\documentclass[11pt]{article}
\usepackage[left=1in,top=1in,right=1in,bottom=1in,head=0in,nofoot]{geometry}

\usepackage{setspace,url,bm,amsmath} % For double-spacing, URL font, math symbols
\usepackage{scalerel}

 \usepackage{multirow}
 \usepackage{arydshln}
\usepackage{titlesec} % Section header formatting
\titlelabel{\thetitle.\quad} % Section header formatting
\titleformat*{\section}{\bf\Large\center}

\usepackage{float}
\usepackage{graphicx} % Graphics scaling
\usepackage{bbm}
\usepackage{latexsym}
\usepackage{caption}
\usepackage[margin=20pt]{subcaption}
\usepackage{enumerate}

\graphicspath{ {./Graphs/} }
\usepackage{xr}

\def\ncss{four}
\def\ntot{five}
\def\ncssest{eight}
\def\ntotest{ten} 
\def\formds{ \qquad \text{for} \quad \mds}
\def\glzz{\gamma_{\lin,z}}
\def\sumii{\sum_{i=1}^I}
 \def\suml{\sum_{l=1}^{n_i}}
\def\sumicz{\sum_{i: C_i = 1, Z_i = z}} 
\def\sumic{\sum_{i: C_i = 1}} 

 \def\sumk{sum_{k=1}^K}
\def\sumq{\sum_{q=1}^Q}

      \newcommand{\tauj}{\tau_{[k]}}

\def\wif{w_{i, \ff}}
\def\wit{w_{i,\tt}}
\def\witc{w_{i,\tc}}

\def\htheta{\hat\theta}
\def\tc{{\textup{tc}}}
\def\ff{{\textup{f}}}
\def\tt{{\textup{t}}}

\def\htlm{\htau_{\lin, (m)}}

\def\xiz{x_i^0}
\def\agg{\textup{ag}}
\def\ttl{\hat\tau^\agg_\lin}
\def\ttlz{\hat\tau^{\agg,0}_\lin}
\def\otj{\otimes_{j=1}^J}

\def\sxyz{S_{xY(z)}}\def\bxz{\bx^0}

\def\hgsim{\hg_{\md}^\imp}

 \def\xiz{x_i^0}

\def\htfm{\htau_{\fisher,(m)}}
\def\htlm{\htau_{\lin,(m)}}
\def\htsm{\hat\tau_{\md, (m)}}
\def\htnm{\hat\tau_{\neyman, (m)}}
\def\xmp{x^\mp}

\def\hts{\htau_\md }
\def\hgl{\hg_\lin}
\def\cinf{c_{\infty}}
\def\ci{c_{\infty}}
  \def\mI{\mathcal I}
\def\summ{\sum_{m\in\mathcal M}}
\def\hsesq{\hat{\text{se}}^2}

\def\hses{\hat{\text{se}}_\md}

\def\hse{\hat{\text{se}}}
\def\mmp{\mm}
\def\mp{\textup{mp}}
\def\cv{\textup{ccov}}
\def\cc{\textup{cc}}
\def\mi{\textup{mim}}
\def\mim{\textup{mim}}

\def\cim{\textup{cim}}
\def\uc{\textup{ic}}

\def\im{\textup{imp}}
\def\imp{\textup{imp}}

\def\htscv{\hts^\cv}
\def\ccov{\cv}

\def\sumk{\sum_{k=1}^K}

\def\xicv{x_i^{\cv}}

\def\bxcv{\bx^\cv}

\def\cxis{\{\xis(c) - \bxs(c)\}}

\def\hgfs{\hg_{\fisher}^\dg}
\def\hgos{\hg_{\lin,1}^\dg}
\def\hgzs{\hg_{\lin,0}^\dg}
\def\hgzzs{\hg_{\lin,z}^\dg}

\def\hgfc{\hg_{\fisher}^\cc}
\def\hgoc{\hg_{\lin,1}^\cc}
\def\hgzc{\hg_{\lin,0}^\cc}

\def\htfcv{\htf ^\cv}
\def\htlcv{\htl^\cv}

\def\xidg{x_i^\dg}

\def\dg{{\dagger}}
\def\htss{\htau_{\md}^\dg }
\def\hgss{\hg_{\md}^\dg }
\def\hgfs{\hg_{\fisher}^\dg }
\def\hgls{\hg_{\lin}^\dg }

\def\dgsi{\dagger\in\{\cc, \cv,  \im, \mim, \mp \} }

\def\dgsss{\dagger \in \{ \im, \mi\}}

\def\sxxc{S_{xx}^\cc}
\def\sxyc{S_{xY}^\cc}
\def\hsxxc{\hat S^\cc_{xx}}
\def\hsxyc{\hat S^\cc_{xY}}

\def\sxximp{S_{xx}^\imp}
\def\sxxmim{S_{xx}^\mim}

\def\byu{\by^\uc}
\def\bxc{\bar x^\cc}
\def\byc{\by^\cc}

\def\hxc{\hat x^\cc}
\def\hyc{\hy^\cc}

\def\tb{{\im}}
\def\bc{\bar C}
\def\bxz{\bar x^0}

\def\bcy{\overline{CY}}
\def\bcx{\overline{ Cx}}
\def\mds{\md = \fisher,\lin}
\def\sumq{\sum_{q=1}^Q}
\def\hbt{\htau}
\def\bt{\tau}

\def\hglc{\hg_{\lin}^\cc }
\def\glc{\gamma_{\lin}^\cc }
\def\gsc{\gamma_{\md}^\cc }

\def\hgsc{\hg_{\md}^\cc }
\def\htsc{\htau_{\md}^\cc }

\def\bc{\bar C}

\def\nzzc{N_{z}^\cc }

\def\hc{\hat C}

\def\hs{\hat S}\def\dc{D_c}
\def\meaniz{N_z^{-1}\sumiz}
\def\meani{\ninv\sumi}

\def\baj{\bar {A}_j}
\def\hgsim{\hg_{\md}^\im }
\def\md{{\star}}

\def\htsim{\hts^\im}
\def\htsc{\hts^\cc}
\def\htsmi{\hts^\mi}

\def\gsmi{\gamma_{\md}^\mi }

\def\gsim{\gamma_{\md}^\im }

\def\barm{\bar M}
\def\xij{x_{ij}}
\def\mj{\mathcal {J}}

\def\hxmi{\hx^\mi}
\def\xis{x_i^\tb }

\def\cxis{\{\xis(c) - \bxs(c)\}}

\def\sxxdg{\sxx^\dg }
\def\sxydg{\sxy^\dg }

\def\sxxim{\sxx^\im }
\def\sxxmi{\sxx^\mi }

\def\sxyim{\sxy^\im }

\def\sxymi{\sxy^\mi }

\def\ninv{N^{-1}}

\def\bu{\bar u}

\def\sqrtn{\sqrt N}
\def\mn{\mathcal{N}}
\def\hajxj{\widehat{Ax_j}}

\def\bajxj{\overline{Ax}_j}

\def\xs{x^\tb }
\def\bxs{\bx^\tb }
\def\hxs{\hx^\tb }

\def\hsxxs{\hsxx^\tb }
\def\hsxys{\hsxy^\tb }

\def\hsxxmi{\hsxx^\mi}
\def\hsxymi{\hsxy^\mi}

\def\tm{M}
\def\sxx{S_{xx}}

\def\nzinv{\nz^{-1}}

\def\tz{\tau^{\cc}}
\def\tone{\tau^{\uc}}

\def\ncinv{(N^\cc)^{-1}}
\def\nc{N^\cc}

 \def\dci{D_{c,\infty}}

\def\cji{c_{j,\infty}}

\def\hgc{\hg_{\fisher}^\cc }

\def\hsxx{\hat S_{xx}}
\def\hsxy{\hat S_{xY}}

\def\nz{N_z}

\def\hd{\hat\delta}
\def\hmu{\hat\mu}

\def\nzc{N_{0,\cc }}

\def\bxc{\bar x ^\cc}
\newcommand{\hb}{\hat{\beta}}
\newcommand{\hg}{\hat\gamma}

\def\htlc{\htau_{\lin}^\cc }

\def\htlc{\htau_{\lin}^\cc }

\def\htfim{\htau_\fisher^\im }
\def\htlim{\htau_\lin^\im }

\def\htfmi{\htf ^\mi }

\def\hy{\hat Y}
\def\hgzc{\hg_{\lin,z}^\cc }
\def\hgzrc{\hg_{\lin,0}^\cc }
\def\hgoc{\hg_{\lin,1}^\cc }

\def\by{\bar Y}

\def\nc{N^\cc}

\def\gf{\gamma_\fisher}
\def\gl{\gamma_\lin}
\def\ont{1_N^\T}
\def\beginp{\begin{pmatrix}}
\def\endp{\end{pmatrix}}
\def\on{1_N}

\def\sumiz{\sum_{i:Z_i = z}}

\def\hcxx{\widehat{Cxx}}
\def\bcxx{\overline{Cxx}}
\def\hcxxt{\widehat{Cxx }}
\def\bcxxt{\overline{Cxx }}
\def\hcxyt{\widehat{CxY }}

\def\hcyy{\widehat{CYY }}
\def\bcyy{\overline{CYY }}

\def\bcxyt{\overline{CxY}}
\def\bcxy{\overline{CxY}}

\def\ba{\bar A}
\def\ha{\hat A}
\def\hcy{\widehat{CY}}
\def\hcx{\widehat{Cx}}

\def\hax{\widehat{Ax}}

\def\bax{\overline{Ax}}

\def\hm{\hat M }

\def\mcar{\textsc{mcar}}

\usepackage{xcolor}
\newcommand{\hx}{\hat x}

\def\begini{\begin{itemize}}
\def\endi{\end{itemize}}

\def\mm{\mathcal{M}}

\def\rs{\rightsquigarrow}

\newcommand{\ols}{\textsc{ols}}

\newcommand{\op}{{ o_P(1)}}

\newcommand{\hgf}{\hat\gamma_\fisher}
\newcommand{\hglo}{\hat\gamma_{\lin,1}}
\newcommand{\hglz}{\hat\gamma_{\lin,0}}
\newcommand{\hglzz}{\hat\gamma_{\lin,z}}

\newcommand{\glo}{\gamma_{\lin,1}}
\newcommand{\glz}{\gamma_{\lin,0}}

\newcommand{\htf}{\hat\tau_\fisher}
\newcommand{\htl}{\hat\tau_\lin}

\newcommand{\htn}{\hat\tau_\neyman}

\newcommand{\sxy}{S_{xY}}

\newcommand{\ep}{\epsilon}

\newcommand{\asim}{\overset{.}{\sim}}

\newcommand{\sumi}{\sum_{i=1}^N}

\newcommand{\mt}{\mathcal{T}}

\def\T{{ \mathrm{\scriptscriptstyle T} }}

\def\neyman{{\textsc{n}}}
\def\fisher{{\textsc{f}}}
\def\lin{{\textsc{l}}}

\newcommand{\obs}{{\rm obs}}

\newcommand{\htau}{\hat \tau}

\newcommand{\htx}{\hat \tau_x}

\newcommand{\Yi}{Y_i}
\newcommand{\Zi}{Z_i}
\newcommand{\hY}{\hat Y}

\newcommand{\sumN}{\sum_{i=1}^N}

\newcommand{\bx}{\bar x}

\newcommand{\stt}[1]{S^{#1}_{\tau\tau,\md}}

\newcommand{\ot}[1]{1, \ldots,#1}

\DeclareMathOperator{\var}{var}

\DeclareMathOperator{\cov}{cov}

\DeclareMathOperator{\diag}{diag}

\DeclareMathOperator*{\plim}{plim}

\allowdisplaybreaks

\def\beginy{\begin{eqnarray}}
\def\endy{\end{eqnarray}}
\def\begina{\begin{eqnarray*}}
\def\enda{\end{eqnarray*}}
\def\begine{\begin{enumerate}}
\def\ende{\end{enumerate}}

\usepackage{enumerate} 
\newcommand{\GG}[1]{}

\usepackage{amsthm}
\usepackage{amssymb}
\usepackage{amsmath}
\usepackage{color}

\usepackage{comment}
\theoremstyle{definition}

\newtheorem*{theorem*}{Theorem}

\newtheorem*{rmk*}{remark}
\newtheorem{proposition}{Proposition}
\newtheorem{lemma}{Lemma}
\newtheorem{example}{Example}

\newtheorem{condition}{Condition}

\newtheorem{corollary}{Corollary}
\newtheorem*{corollary*}{Corollary}

\usepackage{natbib} % ASA citation style
\bibpunct{(}{)}{;}{a}{}{,} % ASA citation style

\usepackage{etoolbox} % Bibliography underfull/overfull box fix
\apptocmd{\sloppy}{\hbadness 10000\relax}{}{} % Bibliography underfull/overfull box fix

\usepackage{multibib}
\newcites{sec}{References}

\usepackage{color}
\usepackage{listings}
\usepackage[hidelinks]{hyperref}
\usepackage{booktabs}

\usepackage{lscape}

\RequirePackage[normalem]{ulem}

\begin{document}

\onehalfspacing

\title{\bf \Large
To adjust or not to adjust? Estimating the average treatment effect in randomized experiments with missing covariates
}
\author{Anqi Zhao and Peng Ding
\footnote{Anqi Zhao, Department of Statistics and Data Science, National University of Singapore, 117546, Singapore (E-mail: staza@nus.edu.sg). Peng Ding, Department of Statistics, University of California, Berkeley, CA 94720 (E-mail: pengdingpku@berkeley.edu). Peng Ding was partially funded by the U.S. National Science Foundation (grant \# 1945136). 
%We thank Fan Li from Duke University for helpful discussions. 
}
}
\date{}
\maketitle

\begin{abstract}
Complete randomization allows for consistent estimation of the average treatment effect based on the difference in means of the outcomes without strong modeling assumptions on the outcome-generating process. Appropriate use of the pretreatment covariates can further improve the estimation efficiency.  However, missingness in covariates is common in experiments and raises an important question: should we adjust for covariates subject to missingness, and if so, how? The unadjusted difference in means is always unbiased.  The complete-covariate analysis adjusts for all completely observed covariates and improves the efficiency of the difference in means if at least one completely observed covariate is predictive of the outcome.  Then what is the additional gain of adjusting for covariates subject to missingness? A key insight  is that the missingness indicators act as fully observed pretreatment covariates as long as  missingness is not affected by the treatment, and can thus be used in covariate adjustment to  bring additional estimation efficiency.  This motivates adding the missingness indicators to the regression adjustment, yielding the missingness-indicator method as a well-known but not so popular strategy in the literature of missing data. We recommend it due to its many advantages. First, it removes the dependence of the regression-adjusted estimators on the imputed values for the missing covariates. Second, it improves the estimation efficiency of the complete-covariate analysis and the regression analysis based on only the imputed covariates. Third,  it does not require modeling the missingness mechanism and yields a consistent and efficient estimator even if the missing-data mechanism is related to the missing covariates and unobservable potential outcomes. Lastly, it is easy to implement via standard software packages for least squares. We also propose modifications to the missingness-indicator method based on asymptotic and finite-sample considerations. To reconcile the conflicting recommendations in the missing data literature, we analyze and compare various strategies  for analyzing randomized experiments with missing covariates under the design-based framework.  This framework treats randomization as the basis for inference and does not impose any modeling assumptions on the outcome-generating process and missing-data mechanism. 

 \end{abstract}

\section{Introduction}\label{sec:intro}

Ever since the seminal work of \citet{Fisher35}, randomization has become the gold standard for estimating treatment effects without strong modeling assumptions. It justifies the simple comparison of outcome means across treatment groups   \citep{Neyman23} and allows for additional gains in efficiency by appropriate adjustment for pretreatment covariates \citep{Fisher35, Lin13}. In particular, \citet{Lin13} showed that the coefficient of the treatment from the ordinary least squares ({\ols}) fit of the outcome on the treatment, centered covariates, and their interactions is a consistent and asymptotically efficient estimator for the average treatment effect, and moreover, the associated Eicker--Huber--White robust standard error is a convenient approximation to the true standard error. Importantly, \citet{Lin13}'s theory holds even if the linear model is misspecified.

Missingness in covariates, however, is ubiquitous in field experiments in biomedical and social sciences, and imposes a dilemma to subsequent analyses. We can simply ignore the covariates and proceed with the unadjusted difference in means, which is unbiased and consistent under complete randomization.  An immediate improvement is the {\it complete-covariate analysis} that adjusts for only the subset of covariates that are observed for all units.  It is more efficient than the unadjusted estimator if at least one completely observed covariate is prognostic to the outcome.  Although some existing methods for missing covariates are likely to further improve efficiency,  they rely on additional modeling assumptions on the outcome model or the missingness mechanism \citep{rubin1987multiple, little, ibrahim2005missing, don}.  Due to the unverifiable additional assumptions, these more sophisticated methods do not strictly dominate the simpler unadjusted estimator and complete-covariate analysis in randomized experiments.  Then a natural question arises: should we adjust for the missing covariates or not?

Our answer to the question above is yes. We propose to simply impute the missing covariates with zeros, augment the imputed covariates with the missingness indicators, and then apply \citet{Lin13}'s estimator for covariate adjustment.  This method becomes intuitive if the missingness is not affected by the treatment such that we can view the missingness indicators as a set of fully observed pretreatment covariates. The resulting {\it missingness-indicator method} has many advantages.  First, unlike the {\it complete-case analysis} that discards units with any missing covariates, it is consistent without assuming that the complete cases are representative of the whole population and can be much more efficient when the proportion of missingness is substantial.  Second, unlike the {\it single imputation} method, it is invariant to the imputed values for the missing covariates, so the convenient choice of imputing by zeros results in no loss of generality.  Third, it is asymptotically more efficient than the complete-covariate analysis and   single imputation  if the missingness indicators are prognostic to the outcome.  Fourth, it does not require modeling of the missingness mechanism and, surprisingly, remains  consistent for the average treatment effect even when the missingness mechanism depends on the missing covariates and unobservable potential outcomes, a scenario analogous to missing not at random under the super-population framework \citep{rubin1976inference, don}.  Lastly, it can be easily implemented via standard software packages for {\ols} and robust standard errors.

This missingness-indicator method is not entirely new. \citet[][Chapter 7]{cohen} proposed to use it in regression analysis; \citet[][Appendix B]{rr1984} suggested it in matching for causal inference with observational studies. 
%\cite{rubin2000} applied the idea of \cite{rr1984} and argued that the propensity score should condition on the missingness indicators if they are prognostic to the outcome. 
See \cite{rubin2000}, \citet[][Sections 9.4 and 13.4]{rosenbaum2010design}, \citet{mattei2009estimating}, and \citet{fogarty2016discrete} for applications of this method to observational studies, and  see \cite{anderson} and \cite{review} for a review. 
However, this method was criticized by \citet{greenland}, \citet{donders}, and \citet{yang2019causal}. 
\citet{greenland} evaluated this method in the context of logistic regression in observational studies and reported severe bias even when the data are missing completely at random; \cite{donders} offered a similar discussion.  \citet{yang2019causal} argued that the underlying assumption in \cite{rr1984} for this method is unreasonable in observational studies. \citet{miettinen} acknowledged its convenience for application, but also pointed out its limitation of representing only partial control when applied to confounders. 
However, our recommendation is not in contradiction with the existing literature. 
The fundamental difference between randomized experiments and observational studies explains the seemingly contradictory recommendations.
Although the missingness-indicator method can be problematic in observational studies, it has appealing theoretical guarantees mentioned above in randomized experiments. 
Randomization balances the pretreatment covariates as well as the missingness indicators on average across different treatment groups. Using them in \citet{Lin13}'s procedure thus ensures consistent and efficient estimation for the average treatment effect. 
Our results echo \citet{white2005} and \citet[][Section 2.4]{carpenter2007missing} without assuming that the covariates and outcomes are normally distributed. 
Rather, we analyze and compare different methods under the design-based framework, also known as the randomization-based framework, free of any modeling assumptions. Therefore, the above theoretical guarantees of the missingness-indicator method hold even if the outcome model is misspecified.  
This is also a key distinction between our results and those for correctly-specified regression models with missing covariates
\citep{rubin1987multiple, little, robins1994, jones, ibrahim2005missing}.

Moreover, we propose modifications to the missingness-indicator method based on asymptotic and finite-sample considerations. First, the {\it missingness-pattern method} stratifies the data based on the missingness patterns and applies \citet{Lin13}'s estimator based on the available covariates within each stratum.  It is closely related to post-stratification \citep{luke} if we view the type of missingness pattern as a discrete covariate, and can be seen as an extension of the \emph{available-covariate analysis} proposed by \cite{wilks1932moments}, \cite{matthai1951estimation}, \cite{glasser1964}, and  \cite{haitovsky}.   The resulting estimator allows for heterogeneous adjustments across different missingness patterns, and thereby promises additional asymptotic efficiency over the missingness-indicator method if the covariates and missingness patterns affect the treatment effects in non-additive ways.  We recommend this method if the sample sizes within all missingness patterns are large enough to justify the application of \citet{Lin13}'s estimator. In finite samples, however, both the missingness-indicator and missingness-pattern methods can have substantial variability   due to estimating many regression coefficients when the number of covariates is large.  This motivates the {\it complete-case-indicator method} and the {\it missingness-count method} that augment the imputed covariates with only the scalar complete-case indicator and the missingness-count variable, respectively, instead of all missingness indicators. The resulting estimators may lose asymptotic efficiency but can improve the finite-sample properties.

We start with the completely randomized treatment-control experiment in Section \ref{sec:basic_setup}, and outline {\ntot}  strategies for handling missing covariates in randomized experiments in Section \ref{sec::6 strategies}. We analyze and compare their asymptotic properties from the design-based perspective in Section \ref{sec:asym}. 
We use simulation and an application  to illustrate the methods and theory in Sections \ref{sec:simu} and \ref{sec:application}.  
%%%
We then extend the theory to cluster randomization and stratified randomization  
in Section \ref{sec:ext}. 
We conclude  in Section \ref{sec::discussion} and relegate all technical details to the Supplementary Material.

The following notation facilitates the discussion.
For a finite population $\{(u_i,v_i): i\in \mI\}$, let $\bu  = |\mI|^{-1} \sum_{i \in \mI}  u_i $, $\bar v  =  |\mI|^{-1} \sum_{i \in \mI} v_i $, and 
$S_{uv}   = (|\mI|-1)^{-1} \sum_{i \in \mI} (u_i  -\bu)  (v_i  - \bar v)^\T$ be the finite-population means and covariance, respectively. We specify the composition of $\mI$ in the context. %, with $\mI$ as the index set and abbreviate $S_{u, v}$ as $S_{uv}$ when no confusion would arise. 
In the case of $u_i = v_i$, we also occasionally write $S_{uu}$ as $S_u^2$. % for notational simplicity. 
For $v_i \in \mathbb R$ and $u_i \in \mathbb R^J$, let $  v_i  \sim u_i  $ denote the {\ols} fit of $v_i$ on $u_i$ over $i \in \mathcal I$.
For example, for $\{(Y_i, Z_i, x_i): i =1,\ldots, N\}$ with $Y_i \in \mathbb{R}$, $Z_i \in \mathbb{R}$, and $ x_i \in \mathbb{R}^J$, let $Y_i \sim 1 + Z_i +x_i $ denote the additive  {\ols} fit of $Y_i$ on $(1, Z_i, x_i)$ over $i =1,\ldots, N$ with regressor vector $  (1, Z_i, x_i^\T)^\T$; let $Y_i \sim 1 + Z_i + (x_i - \bar x) + Z_i(x_i - \bar x)$ denote the fully interacted {\ols} fit of $Y_i$ on $\{1, Z_i, (x_i - \bar x)\}$  over $i =1,\ldots, N$ with centered covariates, $x_i - \bar x$,  and regressor vector $ (1, Z_i, (x_i - \bar x)^\T, Z_i(x_i - \bar x)^\T)^\T$. 
Importantly, we do not invoke the modeling assumptions for {\ols}, but evaluate the sampling properties of the numerical outputs from the design-based perspective. We focus on the robust standard errors from {\ols} because the classic standard errors do not have the desired  design-based properties even in the simple cases \citep{Freedman08a, Lin13}.

\section{Basic setup under the treatment-control experiment}\label{sec:basic_setup}

\subsection{Regression-adjusted estimators with completely observed covariates}\label{sec:hts}
Consider an intervention of two levels, $z = 0, 1$, and a finite population of $N$ units, $i =1,\ldots, N$.  Let $Y_i(z)$ be the potential outcome of unit $i$ under treatment $z$. 
The individual treatment effect is $\tau_i  = Y_i(1) - Y_i(0)$, and the  finite-population average treatment effect  is $\tau = N^{-1} \sumN  \tau_i = \by(1) - \by(0)$, where $\by(z) = \meani Y_i(z)$.  
%We focus on finite-population inference and will refer to $\tau$ as the average treatment effect when no confusion would arise.

The designer  assigns $\nz $ units to receive level $z$ with $N_0 + N_1 = N$ and $(e_0, e_1) =(N_0/N , N_1/N).$ 
Let $Z_i$ denote the treatment level received by unit $i$, with $Z_i=1$ for treatment and $Z_i=0$ for control. 
Complete randomization samples $ (Z_1, \ldots, Z_N)$ uniformly from all permutations of $N_1$ 1's and $N_0$ 0's.
The observed outcome is $ Y_i  = \Zi\Yi(1) + (1-\Zi)\Yi(0) $ for unit $i$.

Let $\hY(z) = \nz ^{-1}\sum_{i:Z_i = z}  Y_i $ be the sample average of the outcomes under treatment $z$. 
The difference-in-means estimator $\hat{\tau}_\neyman = \hY(1) - \hY(0)$ is unbiased for $\tau$,  and equals the coefficient of $Z_i$ from the simple {\ols} fit of $Y_i \sim 1 + Z_i$ over $i = \ot{N}$. 
The presence of covariates affords the opportunity to further improve the efficiency. 
Let $x_i = (x_{i1}, \ldots, x_{iJ})^\T$ be the $J$-dimensional covariate vector for unit $i$. 
\citet{Fisher35} suggested an estimator $\htau_\fisher$ for $\tau$, which equals  the coefficient of $Z_i$ from the additive {\ols} fit $Y_i\sim 1 + Z_i + x_i$ over $i = \ot{N}$. 
\citet{Lin13} recommended an improved estimator, $\htau_\lin$, as the coefficient of $Z_i$ from the fully interacted {\ols} fit $Y_i\sim 1 + Z_i + ( x_i - \bar x) + Z_i(x_i - \bar x)$ over $i = \ot{N}$ with centered covariates and treatment-covariates interactions, and showed its asymptotic efficiency over $\hat{\tau}_\neyman$ and $\htau_\fisher$.  We summarize the results in Lemma \ref{lem::basic-x} below, with
the subscripts ``\neyman",  ``\fisher", and ``\lin" signifying \cite{Neyman23}, \cite{Fisher35}, and \cite{Lin13}, respectively. 
We adopt the finite-population design-based framework conditioning on the potential outcomes, but the theory extends to the super-population framework with minor modifications \citep{tsiatis2008covariate, negi2021revisiting}. The following regularity condition is standard for inference under the design-based framework \citep{DingCLT}.

\begin{condition}\label{asym_basic} As $N \to \infty$, 
(i) $e_z  $ has a limit in $  (0,1)$ for $z =  0,1$, 
(ii) the finite-population first two moments of $\{Y_i(0), Y_i(1), x_i\}_{i=1}^N $  have finite limits; the limit of $S_x^2$ is positive definite, and 
(iii) there exists a  $u_0 < \infty$ independent of $N$ such that $\meani \|x_i\|_4^4 \leq u_0$ and $\meani Y_i^4(z) \leq u_0$ for $z = 0,1$.
\end{condition}

\def\sszs{S^2_{z, \md}}

Let $\gamma_{\lin,z} = (S^2_x)^{-1}\sxyz $ be the coefficient of $x_i$ from the {\ols} fit $Y_i(z) \sim 1+x_i$ over $i = \ot{N}$.
%, where $\sxyz$ indicates the finite-population covariance of $\{x_i, Y_i(z)\}_{i=1}^N$. 
Let $\gf  =    e_0\gamma_{\lin,0} + e_1\gamma_{\lin,1}$, and let $S_{z,\neyman}^2$, $S^2_{z, \fisher}$, and $S^2_{z, \lin}$  be the finite-population variances of $\{Y_i(z)\}_{i=1}^N$, $\{Y_i(z) - x_i^\T\gf \}_{i=1}^N$, and $\{Y_i(z) - x_i^\T \gamma_{\lin,z} \}_{i=1}^N$, respectively, for $z = 0,1$. 
Let $S_{\tau,\neyman}^2 = S^2_{\tau,\fisher}$ and $S^2_{\tau,\lin}$ be the finite-population variances of $(\tau_i)_{i=1}^N$ and $\{\tau_i -x_i^\T (\gamma_{\lin,1} - \gamma_{\lin,0})\}_{i=1}^N$, respectively. 
Condition \ref{asym_basic} ensures $e_z$,  $\sszs $, and $S^2_{\tau,\md}$ all have finite limits for $z = 0,1$ and $\md = \neyman, \fisher, \lin$. 
We will also use  the same symbols  to denote their respective limits when no confusion would arise. 

\begin{lemma}\label{lem::basic-x}
Assume complete randomization and Condition \ref{asym_basic}. We have  $\sqrtn (\hts - \tau) \rs \mn(0, v_\md)$ for $\md =  \neyman, \fisher, \lin$ with 
$$
v_\md = e_0^{-1}S^2_{0,\md} + e_1^{-1}S^2_{1, \md}  - S^2_{\tau, \md}
$$ 
and $v_\lin \leq v_\neyman$, $v_\lin \leq v_\fisher$. 
Further let $\hses$ be the robust standard error of $\hts$ from the corresponding {\ols} fit.  We have $N\hsesq _\md  - v_\md = S^2_{\tau, \md} + \op$ with $S^2_{\tau, \md} \geq 0$.
\end{lemma}

Lemma \ref{lem::basic-x} reviews the classic results from \cite{Neyman23}, \cite{Freedman08a}, and \cite{Lin13}.
The  simple, additive, and fully interacted {\ols} fits thus give consistent and asymptotically normal estimators for $\tau$, and the corresponding robust standard errors are asymptotically conservative for estimating the true standard errors. Asymptotically, \citet{Lin13}'s estimator is the most efficient, whereas \cite{Fisher35}'s estimator can be even less efficient than the unadjusted difference in means. 
 
\citet{LD20} showed that $v_\lin =v_\neyman (1-R^2)$, 
 where $R^2$ is the squared multiple correlation between $\htn$ and the difference in means of the covariates, $\hat\tau_x = \hat x(1) - \hat x(0)$ with $\hat x(z) = \meaniz x_i$.
Therefore, including more covariates in \citet{Lin13}'s estimator will never decrease the asymptotic efficiency. Lemma \ref{lem:eff_a} below gives a more general result.

\begin{lemma}\label{lem:eff_a}
Let $\htau_{\lin}^a$ and $\htau_{\lin}^b$ be \citet{Lin13}'s estimators based on covariates $a_i $  and $b_i $, respectively. 
Let $A$ and $B$ be the matrices with $a_i$ and $b_i$ being the $i$th row vectors, respectively. 
If the column space of $A$ is a subset of that of $B$,
then the asymptotic variance of $\htau_{\lin}^a$ is greater than or equal to that of $\htau_{\lin}^b$ under complete randomization and Condition \ref{asym_basic}.  
%If $\sp(a_i)_{i=1}^N \subseteq \sp(b_i)_{i=1}^N$, 
% $\var_\infty(\sqrt N \htau_{\lin}^a) \geq \var_\infty(\sqrt N \htau_{\lin}^b)$, where $\var_\infty(\cdot)$ denotes the asymptotic variance under complete randomization and Condition \ref{asym_basic}.  
\end{lemma}

Lemma \ref{lem:eff_a} reduces the comparison of asymptotic efficiency of two \cite{Lin13}'s estimators to that of  the column spaces of the respective covariate matrix. It affords the basis for comparing the asymptotic efficiency of various estimators; see Section \ref{sec:asym}. 
In finite samples, however, adjusting for more covariates may increase both bias and variance. This motivates us to consider practical estimators suitable for moderate sample sizes; see Section \ref{sec:ext}.

\subsection{Missing data and running examples}
The proposals by \cite{Fisher35} and \cite{Lin13} assume the $J$ covariates 
are fully observed for all $N$ units. 
Of interest is how to adapt when some covariates are only partially available. 

Let $M_i = (M_{i1}, \ldots, M_{iJ})^\T \in \{0,1\}^J$ be the missingness indicators for unit $i$ with $M_{ij} = 1$ if $x_{ij}$ is missing and $M_{ij} = 0$ if otherwise. 
The possible values of $M_i$ defines $2^J$ possible  missingness patterns, indexed by $m = (m_1, \ldots, m_J)^\T\in \{0,1\}^J$.
Not all $2^J$ missingness patterns need to be present in any given data set. 
Let $N_{(m)} =      N_{(m_1  \ldots m_J)}  = \sumi 1(M_i = m)$  be the number of units with missingness pattern $m = (m_1, \ldots, m_J)^\T $, and let $\mm = \{m: N_{(m)} > 0\}$ be the set of missingness patterns that are present in the data set. We use parentheses in the subscript  of $  N_{(m)} $ to differentiate it from the treatment group sizes, $(N_0, N_1)$.

We use the following running examples for illustrating the key concepts throughout the text. 
For simplicity, we use ``obs'' and ``mis'' to denote whether a covariate is observed or missing.

\begin{example}\label{ex:1}
Consider the case with $J=1$ covariate, $x_i = x_{i1}$, for $i = \ot{N}$. 
The missingness indicators satisfy $M_i = M_{i1}  \in   \{0, 1\}$, and suggest two possible missingness patterns, $m \in \{0, 1\}$, with $N_{(0)} = \sumi 1(M_i = 0) = N - \sumi M_i $ and $N_{(1)} = \sumi 1(M_i = 1) =  \sumi M_i $: 
\begin{center}
\begin{tabular}{|c|  c|c |}\hline
missingness pattern $(M_i)$ & $x_{i}$     & \multicolumn{1}{c|}{number of units}\\\hline
0  & obs   & $N_{(0)}  $ \\\hline
1  & mis   & $N_{(1)} $\\\hline
\end{tabular}
\end{center}
\end{example}

\begin{example}\label{ex:2}
Consider the case with $J=2$ covariates, $x_i = (x_{i1}, x_{i2})^\T$, for $i = \ot{N}$.
The missingness indicators satisfy $M_i = (M_{i1}, M_{i2})^\T  \in   \{0, 1\}^2$, and suggest $2^2 = 4$ possible missingness patterns, $m = (m_1, m_2)^\T \in \{0, 1\}^2$, with $ N_{(m)}  = \sumi 1(M_{i1} = m_1, M_{i2} = m_2)$: 
\begin{center}
\begin{tabular}{|c|cc |c|}\hline
missingness pattern ($M_i$) & $x_{i1}$ & $x_{i2}$   & number of units\\\hline
$(0,0)^\T$ & obs & obs  & $N_{(00)}$ \\\hline
$(0,1)^\T$ & obs& mis     & $N_{(01)}$ \\\hline
$(1,0)^\T$ & mis& obs      & $N_{(10)}$\\\hline
$(1,1)^\T$  & mis& mis    & $N_{(11)}$\\\hline
\end{tabular}
\end{center}
\end{example}

Depending on the specific application under consideration, the missingness may or may not depend on the treatment assignment.  
To simplify the presentation, we start with the case where missingness is unaffected by the treatment assignment in Sections \ref{sec::6 strategies} and \ref{sec:asym}, and then extend to the case where $M_i$ can be dependent on $Z_i$ in Section \ref{sec:ext}. 
Let $\tm_i(z) = (M_{i1}(z), \ldots, M_{iJ}(z) )^\T$ be the potential value of $M_i$ if unit $i$ were assigned to treatment $z$.
Condition \ref{cond:a} below formalizes the notion of treatment-independent missingness in terms of these potential values. 

\begin{condition}\label{cond:a}
$ \tm_{i}(0) = \tm_i(1) = M_i$ for all $i = \ot{N}$.
\end{condition}

Condition \ref{cond:a} states that the potential missingness does not depend on the treatment assignment such that the $M_i$'s are effectively a set of fully observed pretreatment covariates unaffected by $Z_i$'s. 
If the missingness happens before the treatment assignment,  it cannot be affected by the treatment and Condition \ref{cond:a} holds automatically. 
If the covariates are collected retrospectively after the experiment, the missingness indicators may be affected by the treatment and Condition  \ref{cond:a} may be violated. 
The former case is arguably more common in randomized experiments \citep{white2005, carpenter2007missing, sullivan2018should}. It is thus our focus.

\section{Five strategies for handling missing covariates}\label{sec::6 strategies}

We focus on {\ntot}  strategies for handling missing covariates. The simplest strategies are the complete-case analysis that utilizes units with completely observed covariates and the complete-covariate analysis that utilizes covariates completely observed for all units. 
We then review the single imputation method that first fills in the missing covariates with some values and then proceeds with the standard complete-data analysis with the imputed data. 
When the missingness indicators act as pretreatment covariates, it is also natural to include them directly in {\ols}, motivating the  missingness-indicator method. 
We end this section by proposing the missingness-pattern method as an  alternative to the missingness-indicator method that ensures additional asymptotic efficiency.

\subsection{Complete-case analysis}\label{sec:cc}
Most standard software routines adopt the \emph{complete-case analysis} as default by dropping all units with any missingness in covariates. 
Let $C_i  =  1(M_i = 0_J)$ be the complete-case indicator for unit $i$, with $C_i = 1$ if and only if the $J$-vector $x_i$ is fully observed. 
The  \emph{complete-case analysis}  uses only the $N^\cc = \sumi C_i$ complete cases, indexed by $ \{i: C_i = 1\}$.
%; we use the superscript ``\cc" to signify the complete cases.   
%Let $\bxc = \ncinv \sum_{i: C_i = 1} x_{i}  $  be the average of $x_i$ over the complete cases. 
Let $\bxc =  \ncinv \sum_{i: C_i = 1} x_{i}  $  be the average of $x_i$'s over the complete cases. 
The resulting analysis fits
\beginy\label{eq:cc_fisher}
&&Y_i \sim 1+Z_i +x_i,\\
\label{eq:cc_lin}
&& Y_i \sim 1 + Z_i  +  (x_i- \bxc  )  + Z_i  (x_i- \bxc )
\endy
over $\{i: C_i = 1\}$ 
under the additive and fully interacted specifications, respectively, 
and uses  the coefficients of $Z_i$, denoted by $\htau_{\fisher}^\cc $ and $\htlc$,  respectively, to estimate $\tau$.

In Example \ref{ex:1} with $J = 1$ covariate, we have $C_i = 1 - M_i$ for $i = \ot{N}$, and the complete-case analysis uses units with $M_i=0$. In Example \ref{ex:2} with $J=2$ covariates, we have $C_i = (1-M_{i1})(1-M_{i2})$ for $i = \ot{N}$, and the complete-case analysis uses units with $M_{i1}=M_{i2}=0$.

\subsection{Complete-covariate analysis}

Another straightforward option is the \emph{complete-covariate  analysis}  that omits any covariates that are not completely observed for all units. Denote by $\mj = \{j: M_{ij} = 0 \ \text{for} \  i =\ot{N}  \}$ the set of complete covariates, and let $\xicv = (x_{ij})_{j\in\mj}$ be the subvector of $x_i$ corresponding to the covariates in $\mj$. 
The complete-covariate analysis fits 
\beginy\label{eq:cv_fisher}
&&Y_i  \sim  1+Z_i +\xicv,\\
\label{eq:cv_lin}
&&Y_i \sim  1 + Z_i + (\xicv - \bxcv) + Z_i(\xicv - \bxcv)  % \quad \text{where}\quad \bxcv = \frac{1}{N}\sumi \xicv, 
\endy 
over $i = \ot{N}$ under the additive and fully interacted specifications, respectively,  and uses the coefficients of $Z_i$, denoted by $\htfcv$ and $\htlcv$, respectively, to estimate $\tau$. 
%The resulting $\htfcv$ and $\htlcv$ are essentially the \cite{Fisher35}'s and \cite{Lin13}'s estimators based on the truncated covariates $(x_i^\ccov)_{i=1}^N$, respectively. 
In the case of $\mathcal J = \emptyset$, both \eqref{eq:cv_fisher} and \eqref{eq:cv_lin} reduce to $Y_i \sim 1+Z_i$ with $\htfcv = \htlcv = \htn$.

\subsection{Single imputation}\label{sec::imputation}
The \emph{single imputation} strategy imputes the missing covariates based on the observed data, and analyzes the imputed data by standard methods. It allows for the inclusion of all cases and covariates in the analysis.

Consider a covariate-wise imputation that imputes all missing $x_{ij}$'s along the $j$th dimension by some prespecified $c_j$ that may depend on the observed data.  
Common choices include $c_j = 0$ and the covariate-wise observed average $c_j = \hx_j^\obs = \sumi (1- M_{ij}) x_{ij} / \sumi (1- M_{ij})$.
Denote by $  \xis(c)    = (x^\imp_{i1}(c_1), \ldots, x^\imp_{iJ}(c_J))^\T$ the resulting imputed covariates with $c = (c_j)_{j=1}^J$ and $x^\imp_{ij}(c_j) = (1-M_{ij}) x_{ij} + M_{ij} c_j $. % and we use the superscript ``\im"  to indicate imputation.
We can proceed with fitting
\beginy
\label{eq:im_f}
&&Y_i \sim 1+Z_i + \xis(c), \\
\label{eq:im_l} 
&&Y_i \sim  1 + Z_i +  \cxis + Z_i \cxis  % \  \text{where}\  \bxs(c) = \frac{1}{N}\sumi \xis(c),   \qquad 
\endy
over $i = \ot{N}$, 
respectively, and estimate $\tau$ by the coefficients of $Z_i$, denoted by  $ \htfim(c)$ and $\htlim(c)$, respectively.

The above covariate-wise imputation enforces identical imputation value  for all missing values along the same covariate, and could thus appear quite restrictive at first glance. 
Other common choices for single imputation include the treatment-specific sample means of the observed covariates \citep{schemper1990}  and other conditional sample means of the observed covariates based on either only the observed covariates or both the observed covariates and outcomes  \citep{little}. 
We will nevertheless focus on the simple covariate-wise imputation in this text due to its sufficiency for randomized experiments; see \citet{sullivan2018should}, \citet{kayembe2020imputation}, and \citet{reiter} for numerical evidence on the insensitivity of standard analyses to imputation methods in randomized experiments.  We leave the theory on more sophisticated imputation methods to future work.
Moreover, we will show that the robust standard errors from {\ols} are convenient approximations to the true standard errors without any adjustment.  Therefore, we also omit the discussion of multiple imputation as a tool to assess imputation uncertainty in general scenarios \citep{rubin1987multiple}.

\subsection{Missingness-indicator method}\label{sec::mim}

The \emph{missingness-indicator method}  augments \eqref{eq:im_f} and \eqref{eq:im_l} under single imputation by also including $M_i$ as additional regressors to account for the missingness information.
% \citep{cohen,miettinen}.  
%{\color{red} this is too complicate. just tell people what to do... we do not need to do literature review here and tell the readers how we get this method. we can comment on it but do not need their proposal to motivate ours..}
%\citet{cohen} proposed to impute all missing covariates by some arbitrary constant $c_0\in\mathbb R$, and fits the additive regression as $Y_i \sim 1 + Z_i + x_i^\imp(c_0 1_J) + M_i$. covariate-wise as in Section \ref{sec::imputation}
In particular, we first impute the missing $x_{ij}$'s by covariate-specific $c_j$'s $(j=1,\ldots, J)$, and then fit
\beginy
&&Y_i  \sim  1 + Z_i + x_i^\imp(c) + M_i, \label{eq:mi_f}\\
&&Y_i  \sim  1 + Z_i + (x_i^\imp(c) - \bar x^\imp(c))  +  (M_i - \bar M)    + Z_i  (x_i^\imp(c) - \bar x^\imp(c)) + Z_i (M_i - \bar M) \label{eq:mi_l}
\endy
%where $\bar x^\imp(c) = \meani x_i^\imp(c)$ and $\bar M = \meani M_i$, 
over $i = \ot{N}$ to construct the regression estimators as  the  coefficients of $Z_i$, denoted by $\hat\tau^\mim_\fisher(c)$ and $\hat\tau^\mim_\lin(c)$, respectively. This is equivalent to running the additive and fully interacted regressions based on the augmented covariate vector $x_i^\mim(c) = (  x_i^\imp(c) ^\T, M_i^\T)^\T \in \mathbb R^{2J}$. 
%We use the superscript ``\mim" to indicate the missingness-indicator method.

Strictly speaking, $M_{ij} = 0$ for all $i = \ot{N}$ for $j\in\mj$ such that we need to include only $M_{ij}$ for the incomplete covariates. 
In addition, in case $M_{ij} = M_{ij'}$ for all $i = \ot{N}$  for some $j \neq j'$, we need to include only one of them to avoid collinearity. 
Acknowledging the need to adjust for these complications on a case-by-case basis, we will use $M_i$ to represent the vector of missingness indicators added to the model after appropriate adjustment for notational simplicity. This causes little confusion because standard software packages for {\ols} automatically drop redundant regressors. 

As it turns out, the resulting point estimators  $\hts^\mim(c) \ (\mds)$  and their associated robust standard errors $\hse_\md^\mim(c)  \ (\mds)$ are invariant to the choice of the imputation vector $c$. This is a numeric merit due to the inclusion of the missingness indicators, and allows us to construct $\hts^\mim(c)$ by simply imputing all missing covariates as 0. We formalize the intuition in Lemma \ref{lem:invar_mim} below. 

%Let  $\hse_\md^\mim(c)$ be the robust standard errors associated with $\hts^\mim(c)$ from regressions \eqref{eq:mi_f} and \eqref{eq:mi_l}. 
Let  $x^0_i = x_i^\imp(0_J)$  be the imputed covariate vector where we fill in all missing $x_{ij}$'s with 0.  
Let $\hts^\mim  = \hts^\mim(0_J)$ and $\hses^\mim = \hses^\mim(0_J) $  be the regression estimators and associated robust standard errors from
%\beginy
%&&Y_i  \sim  1 + Z_i + x_i^0 + M_i, \label{eq:mi_f_0}\\
%%&&Y_i  \sim  1 + Z_i + \dt x_i^0  +\dt M_i + Z_i \dt x_i^0 + Z_i \dt M_i \label{eq:mi_l_0}
%&&Y_i  \sim  1 + Z_i + (x_i^0- \bxz ) + (M_i   - \bar M) + Z_i (x_i^0- \bxz ) + Z_i (M_i   - \bar M),   \label{eq:mi_l_0}
%\endy
\beginy
\label{eq:mi_f_0}\
&& Y_i   \sim   1 + Z_i + x_i^0 + M_i \\
%&\sim& 1 + Z_i + x_i^\mim , \label{eq:mi_f_0}\\
\label{eq:mi_l_0}
&& Y_i   \sim    1 + Z_i + (x_i^0- \bxz ) + (M_i   - \bar M) + Z_i (x_i^0- \bxz ) + Z_i (M_i   - \bar M) 
%&\sim& 1 + Z_i + (x_i^\mim - \bx^\mim) + Z_i (x_i^\mim - \bx^\mim)   \label{eq:mi_l_0}
\endy
over $i = \ot{N}$, respectively, for $\mds$.
%where 
%$x_i^\mim = x_i^\mim(0_J) = ( (x_i^0)^\T, M_i^\T)^\T$, $\bxz  = \meani x^0_i$,  and $\bx^\mim = \meani x_i^\mim = ((\bxz)^\T, \bar M^\T)^\T$. 
These are effectively regression adjustments with covariates $x_i^\mim  = ( (x_i^0)^\T, M_i^\T)^\T$.

\begin{lemma}\label{lem:invar_mim}
%\prearbitrary, we have 
$\hts^\mim(c) = \hts^\mim$ and $\hse_\md^\mim(c) =  \hses^\mim$ for $\mds$ for all $c\in\mathbb R^J$. 
\end{lemma}

\cite{cohen} hinted at the invariance of $\htf^\mim(c)$ to the choice of $c$. Lemma \ref{lem:invar_mim} formalizes the results for both $\htf^\mim(c)$ and $\htl^\mim(c)$ as well as their corresponding robust standard errors $\hse_\fisher^\mim(c)$ and $\hse_\lin^\mim(c)$.

\subsection{Missingness-pattern method}\label{sec:mp}

The missingness-indicator method factors in information in missingness by including $M_i$ as additional regressors in {\ols}. 
We propose the \emph{missingness-pattern method} that goes one step further  and performs one separate analysis for  each missingness pattern based on all available covariates. 
Let $\rho_{(m)} = N_{(m)}/ N$ denote the proportion of units with missingness pattern $m \in   \{0,1\}^J$. 
We use Examples \ref{ex1:mp} and \ref{ex2:mp} (continued) below with $J=1$ and $J=2$ to illustrate the basic idea, and then formalize the method for general $J$.

\setcounter{example}{0}
\begin{example}[continued]\label{ex1:mp}
Consider the case of $J = 1$ covariate and two missingness patterns.
%, $m \in \mathcal{M} = \{0, 1\}$. 
%Assume  $N_{(m)} \geq 4$ under each pattern. 
We can fit one additive regression for each missingness pattern to obtain the coefficients of $Z_i$ as follows:
\begine[(i)]
\item  regress $Y_i$ on $(1, Z_i,  x_i)$ over $\{i: M_i = 0\}$ to obtain   $\htau_{\fisher, (0)}$; 
% with observed covariates and compute the coefficient of $Z_i$ as the pattern-specific effect of the complete cases, denoted by $\htau_{\fisher, (0)}$; 
\item regress $Y_i$ on $(1, Z_i)$ over $\{i: M_i = 1\}$ to obtain   $\htau_{\fisher, (1)} =  \htau_{\neyman,(1)}$. 
% with missing covariates and compute the coefficient of $Z_i$ as the pattern-specific effect of the incomplete cases, denoted by $\htau_{\fisher, (1)} = \htau_{\neyman,(1)}$.
\ende
The weighted average $\htau_{\fisher}^\mp  = \rho_{(0)} \htau_{\fisher, (0)} + \rho_{(1)} \htau_{\fisher, (1)}$  gives an estimator for $\tau$.
Analogously, we can also obtain $\htau_{\lin}^\mp$ by running one fully interacted regression for each missingness pattern. 
\end{example}

\begin{example}[continued]\label{ex2:mp}
Consider the case of $J=2$ covariates, $x_i = (x_{i1}, x_{i2})^\T$, and four missingness patterns.
%, $m \in \mathcal{M} = \{(0,0)^\T, (0,1)^\T, (1,0)^\T, (1,1)^\T\}$.  
%Assume enough sample size under each pattern. 
We can fit one additive regression for each missingness pattern to obtain the coefficients of $Z_i$ as follows:
\begine[(i)]
\item regress $Y_i$ on $(1, Z_i,  x_{i1} , x_{i2})$ over $\{i: M_i = (0, 0)^\T\}$  to obtain  $\htau_{\fisher, (0,0)}$; 
%with both covariates observed, and compute the coefficient of $Z_i$ as the subgroup  effect, denoted by $\htau_{\fisher, (0,0)}$; 
\item regress $Y_i$ on $(1, Z_i, x_{i1})$ over $\{i: M_i = (0, 1)^\T\}$ to obtain   $\htau_{\fisher, (0, 1)}$; 
%with $x_{i2}$ missing, and compute the coefficient of $Z_i$ as the pattern-specific effect, denoted by $\htau_{\fisher, (0, 1)}$; 
\item regress $Y_i$ on $(1, Z_i, x_{i2})$ over $\{i: M_i = (1,0)^\T\}$ to obtain  $\htau_{\fisher, (1,0)}$; 
%with $x_{i1}$ missing, and compute the coefficient of $Z_i$ as the pattern-specific effect, denoted by $\htau_{\fisher, (1,0)}$; 
\item regress $Y_i$ on $(1, Z_i)$ over $\{i: M_i = (1, 1)^\T\}$ to obtain   $\htau_{\fisher, (1, 1)} = \htau_{\neyman,(1,1)}$.
%with both covariates missing, and compute the coefficient of $Z_i$ as the pattern-specific effect, denoted by $\htau_{\fisher, (1, 1)} = \htau_{\neyman,(1,1)}$.
\ende
The weighted average $\htau_{\fisher}^\mp  = \rho_{(0,0)} \htau_{\fisher,(0,0)} + \rho_{(0,1)} \htau_{\fisher,(0,1)} + \rho_{(1,0)} \htau_{\fisher,(1,0)}+ \rho_{(1,1)} \htau_{\fisher, (1,1)}$ gives an estimator of $\tau$.  Analogously, we can also obtain $\htau_{\lin}^\mp$ by running one fully interacted regression for each missingness pattern. 
\end{example}

Extensions to general $J$ are immediate.
Let $\xmp_i   = (x_{ij})_{j: M_{ij} = 0}$ be the vector of available covariates for unit $i$.
In the above Example \ref{ex2:mp} (continued), we have 
(i) $\xmp_i  = x_i  = (x_{i1}, x_{i2})^\T$ for units with $M_i = (0, 0)^\T$; 
(ii) $\xmp_i  = x_{i1} $ for units with $M_i = (0, 1)^\T$;
(iii) $\xmp_i  = x_{i2} $ for units with $M_i = (1, 0)^\T$;
and 
(iv) $\xmp_i  = \emptyset$ for units with $M_i = (1, 1)^\T$. 
For units with missingness pattern $m$, the missingness-pattern method adjusts for the available covariates by fitting
%The additive and fully interacted regressions for missingness pattern $m \in \mm$ adjust for the available covariates $\xmp_i  $ for all units with $M_i = m$: 
\beginy
&&
Y_i \sim 1 + Z_i + \xmp_i, \label{eq:mp_f_m}\\
&& Y_i \sim 1 + Z_i + (\xmp_i - \bar x^\mp_{(m)}) + Z_i  (\xmp_i - \bar x^\mp_{(m)})\label{eq:mp_l_m} \qquad\
%\quad \text{where} \quad \bar x^\mp_{(m)} = N_{(m)}^{-1}\sum_{i: M_i = m} \xmp_i, 
\endy
over $\{i: M_i = m\}$, where $\bar x^\mp_{(m)} = N_{(m)}^{-1}\sum_{i: M_i = m} \xmp_i$. 
Let $\htfm $ and $\htlm $ be the coefficients of $Z_i$ from \eqref{eq:mp_f_m} and \eqref{eq:mp_l_m}, respectively, with $\htau_{\fisher, (1_J)} = \htau_{\lin, (1_J)} = \htau_{\neyman, (1_J)}$ equaling the coefficient of $Z_i$ from $Y_i \sim 1+Z_i$ over $\{i: M_i = 1_J\}$ for $m = 1_J$. 
The weighted averages 
\beginy\label{eq:htau_mp}
\hts  ^\mp  = \sum_{m \in \mmp } \rho_{(m)} \htsm  \qquad (\mds) 
\endy
then give two covariate-adjusted estimators of  $\tau$. 

The missingness-pattern method differentiates between all $|\mathcal M|$ missingness patterns like the missingness-indicator method, yet does so by using $|\mathcal M|$ missingness-pattern-specific {\ols} fits.
It can be seen as a hybrid of the complete-case and complete-covariate analyses,  factoring in all available covariates without the need of imputation or augmentation.
With $\htau_{\md,(0_J)}$ coinciding with $\hts^\cc$ from the complete-case analysis, it can also be seen as an ensemble variant of $\hts^\cc$, averaging over estimators from not only the complete cases but also other missingness patterns as well. 

The idea of missingness-pattern-specific analysis dates back to \citet{wilks1932moments}, \citet{matthai1951estimation}, and \citet[Appendix B]{rr1984}, yet its use for analyzing experiments with missing covariates remains mostly unexploited to the best of our knowledge.
A key intuition is that the missingness pattern acts as a discrete pretreatment covariate, and thus allows for post-stratified estimators by averaging over estimators within missingness patterns. 
\citet{luke} demonstrated the asymptotic efficiency gain of post-stratification based on the simple stratum-specific differences in means without adjusting for additional covariates. The $\htl ^\mp$ in \eqref{eq:htau_mp} averages over regression-adjusted estimators within missingness patterns, and promises additional large-sample efficiency  over the missingness indicator method by allowing heterogeneous adjustments across different missingness patterns. 
We quantify the intuition in Section \ref{sec:asym}.

Despite the desired gain in large-sample efficiency, the missingness pattern method can be demanding on the missingness-pattern-specific sample sizes in finite samples even with a moderate $J$. Denote by $J_{(m)} = \sum_{j=1}^J (1-m_j)$ the number of available covariates under missingness pattern $m$. The pattern-specific additive estimator $\htfm $ is well defined  only if $N_{(m)} \geq J_{(m)} +2$; 
the pattern-specific fully interacted estimator $\htlm $ is well defined  only if  $\min\{ N_{(m), 0}, N_{(m), 1} \} \geq J_{(m)} + 1 $, where $N_{(m), z}$ denotes the number of units with missingness pattern $m$ that receive treatment $z \in \{ 0,1\}$. 
When some $\htsm $'s are not well-defined due to these sample size constraints, we can replace them by $\htnm$ as the difference in means within missingness pattern $m$,  and construct the final estimators by averaging over a mixture of adjusted and unadjusted pattern-specific estimators. 
Alternatively, we can collapse small similar missingness patterns and form regression estimators based on coarsened missingness patterns and imputed covariates. A related idea has been exploited by  \citet{pashley2021insights} for variance estimation in finely stratified experiments. 
Nevertheless, we  will focus on $\htau_{\fisher}^\mp $ and $\htau_{\lin}^\mp $ for simplicity and leave the more complex estimators to future work. 
%We will focus on the purebred $\htau_{\fisher}^\mp $ and $\htau_{\lin}^\mp $ in the remainder of the text for simplicity, and leave the more complex hybrid estimators to future work. 
In case some $\htsm$'s are not well-defined due to the sample size constraints, we recommend going back to the missingness-indicator method   to ensure finite-sample feasibility. 
%or single imputation

Last but not least, 
performing $|\mm|$ missingness-pattern-specific regressions could be cumbersome in practice when $|\mm|$ is large. 
As it turns out, regression adjustment with $x_i^\imp(c), M_{i1}, \ldots, M_{iJ},$ and all their interactions recovers $\{\htf^\mp, (\hse_\fisher^\mp)^2\}$ and $\{\htl^\mp, (\hse_\lin^\mp)^2\}$ via one aggregated {\ols} fit each.  Section \ref{sec:asym} gives more details. 
%a fully interacted regression of $Y_i$ on $Z_i, x_i^\imp(c), M_{i1}, \ldots, M_{iJ},$ and all their interactions, namely terms like  $M_{i1}M_{i2}$, $x_i^\imp(c)M_{i1}$, $Z_i x_i^\imp(c)M_{i1}M_{i2}$, etc., 
%after proper centering, recovers $\{\htl^\mp, (\hse_\lin^\mp)^2\}$ via one least-squares fit; likewise for $\{\htf^\mp, (\hse_\fisher^\mp)^2\}$. 
%We formalize the intuition in Section \ref{sec:asym}. 

\subsection{Summary of the covariate-adjusted regression estimators}

Sections \ref{sec:cc}--\ref{sec:mp} give a total of {\ntotest}  covariate-adjusted regression estimators, $\htss$, as the combinations of {\ntot}  missing-data strategies, $\dgsi$,  and two model specifications, $\mds$. Table \ref{tb:strategies} summarizes them. 
For notational simplicity, we suppress the suffix ``$(c)$" in $\{\hts^\imp(c), x^\imp_i(c)\}$   when using $\htss$ and $x_i^\dg$  to represent the union of the {\ntotest}  estimators and their corresponding covariate vectors.

Of interest is their respective validity and efficiency for inferring $\tau$.
We address this question in Section \ref{sec:asym}. 
A key observation is that $\emptyset \subseteq x_i^\ccov \subseteq x_i^\imp(c) \subseteq x_i^\mim(c)$  such that the model specifications under the complete-covariate analysis and single imputation can be seen as restricted variants of those under the missingness-indicator method. 
By Lemma \ref{lem:eff_a}, this elucidates the efficiency of $\htl^\mim$ over $\htl^\imp(c)$, $\htl^\ccov$, and $\htn$ if the $x_i^\dg$'s act as standard covariates satisfying Condition \ref{asym_basic}. We formalize the intuition in Section \ref{sec:asym}.

\begin{table}[t]\caption{\label{tb:strategies} 
 Ten estimators $\htss$ under {\ntot}  missing-covariate strategies,  $\dgsi$, and two model specifications, $\mds$.
%, along with the set of units included under each strategy, denoted by $\mathcal I^\dg$, and the respective covariates for forming the additive and fully interacted regressions under each strategy as $Y_i \sim 1+Z_i + w_i $ and $Y_i \sim 1+Z_i + (w_i -\bar w) + Z_i  (w_i -\bar w)$ over $i \in \mathcal I^\dg$, respectively, with $\bar w = |\mathcal I^\dg|^{-1}\sum_{i\in\mathcal I^\dg} w_i$. 
% is the average of $w_i$ over $i \in \mathcal I^\dg$. 
%In particular, $\mathcal I^\dg = \{\ot{N}\}$ for $\dg \in\{ \ccov, \im, \mim, \cim, \mp\}$ and $\mic  = \{i: C_i = 1\}$. 
The complete-case analysis uses units with $ C_i = 1$, and other methods uses all units.
}
\begin{center}
\begin{tabular}{|l|l|l|}\hline
$\htss$  & \multicolumn{1}{|c|}{missing-covariate strategy}    & \multicolumn{1}{c|}{covariates  for regressions}\\\hline
% -------------------------------------------------------------------
$\hts ^\cc $ & use complete cases and all covariates & $  x_i $     \\\hline
% -------------------------------------------------------------------
$\hts ^\cv $ & use all units and complete covariates   & $ x_i^\cv = (x_{ij})_{j \in \mj}$    \\\hline
% -------------------------------------------------------------------
$\hts ^\im (c) $ &impute the missing $x_{ij}$'s with  $c_j$;    & $  x_i^\imp(c)$ with\\
&  run regressions with all imputed covariates & $x^\im  _{ij} (c_j)= (1-M_{ij}) x_{ij} + M_{ij}c_j $ \\
\hline 
%&with  $c_j$ && where $x^\im  _{ij} (c_j)= (1-M_{ij}) x_{ij} + M_{ij}c_j $ \\\hline
%%%%%%%%%%%%%%%%%%%%%%%%%
$\hts ^\mi $ & impute the missing $x_{ij}$'s with $0$;   & $x_i^\mim =( (x_i^0)^\T, M_i^\T)^\T$, where  \\
&  augment the regressions with $M_{ij}$'s   & $x_i^0 = x_{i}^\imp(0_J)$ \\  \hline
%%%%%%%%%%%%%%%%%%
$\hts ^\mp$ & run missingness-pattern-specific {\ols}    & $\xmp_i   = (x_{ij})_{j: M_{ij} = 0} $
\\ \hline
  \end{tabular}
\end{center}
\end{table}

\section{Design-based theory}\label{sec:asym}

 We quantify in this section the design-based properties of the estimators in Table \ref{tb:strategies}. 
In particular, the regression-based covariate adjustment delivers not only point estimators but also their associated robust standard errors, denoted by $\hses^\dg$. 
 We focus on the validity of $(\htss, \hses^\dg)$ for the large-sample Wald-type inference of $\tau$, which concerns the construction of confidence intervals based on the consistency and asymptotic normality of $\htss  $ and the asymptotic conservativeness of $\hses^\dg$ for estimating the true standard error.  Although $\htss$ and $\hses^\dg$ are originally motivated by linear models, we will show that their theoretical guarantees hold under the design-based framework irrespective of whether the corresponding linear models are correctly specified or not.

We focus on individual strategies in Sections \ref{sec:cc_asym}--\ref{sec:mp_asym}, and then unify them under a hierarchy of model specifications with increasing complexity in Section \ref{sec:unify}. 
In a nutshell, we do not recommend  $\hts ^\cc \ (\mds)$ due to their inconsistency without a strong additional assumption. 
We recommend  $\htl^\mi $ in general due to its simplicity, invariance to $c$, and efficiency over $\htn$, $\htf^\mim$, $\hts ^\cv $, and $\hts ^\im (c)$ for $\mds$. 
%When the missingness indicators are highly correlated, we recommend $\htl^\ac (c)$ for better finite-sample properties. 
When the missingness-pattern-specific sample sizes permit, we recommend $\htl^\mp$ due to its additional gain in asymptotic efficiency.

With a slight redundancy of notation, 
let $A_{ij} = 1 -M_{ij} $ indicate the availability of $x_{ij}$ with  $A_i = (A_{i1}, \ldots, A_{iJ})^\T = 1_J - M_i$ and $A_i \circ  x_i  = (A_{i1}  x_{i1}, \ldots, A_{iJ}  x_{iJ})^\T=x^0_i$.
 %let $A_{ij} = 1 -M_{ij} $ indicate the availability of $x_{ij}$ with $\bar A_j = \meani A_{ij} = 1- \bar M_j$. Let $A_i = (A_{i1}, \ldots, A_{iJ})^\T = 1_J - M_i$ with $A_i \circ  x_i  = (A_{i1}  x_{i1}, \ldots, A_{iJ}  x_{iJ})^\T=x^0_i$.
We need the following regularity conditions for asymptotic analysis.
% on the design, potential outcomes, covariates,  and missingness. 

\begin{condition}\label{asym_1} 
Assume Condition \ref{cond:a}. 
As $N \to \infty$, 
(i) $e_z  $ has a limit in $  (0,1)$ for $z = 0,1$,  
(ii) the first two moments of $\{Y_i(z), x_i, M_i, C_i, C_i Y_i(z),  C_i x_i, A_i \circ  x_i : z= 0,1 \}_{i=1}^N $  have finite limits, with  $ \bar C = \meani C_i$ having a limit in $(0,1]$, and 
(iii) there exists a  $u_0 < \infty$ independent of $N$ such that $\meani \|x_i\|_4^4 \leq u_0$ and $\meani Y_i^4(z) \leq u_0$ for $z = 0,1$.
\end{condition}

\subsection{Complete-case analysis}\label{sec:cc_asym}

We derive in this section the asymptotic properties of $\htsc \ (\mds)$ from \eqref{eq:cc_fisher} and \eqref{eq:cc_lin}. 
The result suggests $\htsc$ is in general not consistent for estimating $\tau$ unless the average treatment effect of the complete cases equals that of the incomplete cases asymptotically. 
This is a strong assumption, and can be problematic whenever the missingness is correlated with the potential outcomes. A common example in randomized clinical trials is that the more severely ill patients are more likely to have missing pretreatment covariates.
We thus do not recommend the complete-case analysis.

For the $\nc  = \sumi C_i = N \bar C $  complete cases with $C_i = 1$, let  $S_{xx}^\cc$ be the  finite-population covariance of the completely observed  covariates, $(x_i)_{i: C_i = 1}$;  
let $\byc(0)$ and $\byc(1)$ be the averages of the potential outcomes, $\{Y_i(z)\}_{i:C_i = 1}$, for $z = 0,1$, respectively, and let $\tau^\cc   = \by^\cc (1) - \by^\cc (0)$ be the average treatment effect. 
For the $N^\uc = N - \nc $ incomplete cases with $C_i = 0$, we can similarly define $\by^\uc(1)$, $\by^\uc(0)$, and $\tau^\uc$. 
By definition, we have $\bar Y(z) = \bar C \byc(z) + (1-\bar C)\byu(z)$ and $\tau =  \bc\tz  +  (1-\bc)\tau^\uc $. 
Further define
$S_{C, \tau} = N^{-1}\sumi (C_i - \bar C) (\tau_i  - \tau)$ as the finite-population covariance of $(C_i, \tau_i)_{i=1}^N$.
% over the whole population $i = \ot{N}$. Without introducing new symbols, we
Under Condition \ref{asym_1}, they all have finite limits. We will also use the same symbols 
to denote their respective limits when no confusion would arise.

Condition \ref{cond:cc} below gives an intuitive quantification of the representativeness of the complete cases from the finite-population perspective. It imposes a strong restriction on the missingness mechanism. 

\begin{condition}\label{cond:cc}
Assume Condition \ref{cond:a}. 
As $N \to \infty$, we have $\tau^\cc = \tau$,  with two equivalent conditions being (i) $ (1-\bar C)(\tz - \tau^\uc) = 0$
or (ii) $S_{C, \tau} = 0$. 
\end{condition}

 Let $v_{\md}^\cc $ and $ \stt{\cc}$ be the analogs of $v_\md$ and $S^2_{\tau, \md}$ in Lemma \ref{lem::basic-x} defined over $\{Y_i(0), Y_i(1), x_i \}_{i: C_i = 1}$ for $\md =  \fisher, \lin$.

\begin{proposition}\label{prop:cc_1}
Assume complete randomization, Condition \ref{asym_1}, and the limit of $S_{xx}^\cc$  is  positive definite. We have $\sqrt {N^\cc}(\hts^\cc - \tau^\cc ) \rs \mn(0, v_{\md}^\cc  )$ % for $\md =  \fisher, \lin$,
%\begina
%\sqrt {N^\cc}(\hts^\cc - \tau^\cc ) \rs \mn(0, v_{\md}^\cc  ) \qquad  \text{for} \quad \md =  \fisher, \lin
%\enda 
with  $v_{\lin}^\cc \leq v_{\fisher}^\cc $. In addition,  $N^\cc(\hses^\cc)^2 -  v_\md ^\cc= \stt{\cc} + \op$, where $\stt{\cc} \geq 0$.
The $\htsc$ is consistent for $\tau$ if and only if Condition \ref{cond:cc} holds. 
\end{proposition} 

Proposition \ref{prop:cc_1} shows that $\hts^\cc$ is consistent for $\tau$ if and only if the complete cases are representative of the whole population in the sense of Condition \ref{cond:cc}. 
Under Condition \ref{cond:cc}, we can use $(\hts^\cc, \hses^\cc)$ for the large-sample Wald-type inference for $\tau$. 
This is intuitive and is coherent with the results under the classic Gauss--Markov model  \citep{jones}. 
The equivalent condition $S_{C,\tau} = 0$, on the other hand, allows for connections between Condition \ref{cond:cc} and the notion of missing completely at random (\mcar) from the super-population framework \citep{rubin1976inference}.
Inspired by the requirement of missingness being independent of the potential outcomes under {\mcar}, 
we can quantify the dependence of missingness on the potential outcomes  under the design-based framework by the finite-population covariance of  $\{M_i, Y_i(z)\}_{i=1}^N$, denoted by $S_{M, Y(z)} = (N-1)^{-1}\sumi (M_i - \bar M) \{Y_i(z) - \bar Y(z)\}$, for $z = 0,1$, and use the condition ``$S_{M, Y(z)} = 0$ for $z=0,1$" as the finite-population analog for the independence between $M_i$ and $Y_i(z)$ under {\mcar}.
In the case of $J = 1$, 
we have $S_{C, \tau} =  S_{M, Y(0)} - S_{M, Y(1)}$ such that the condition ``$S_{M, Y(z)} = 0$ for $z = 0,1$" ensures Condition \ref{cond:cc} and thus
the consistency of $\hts^\cc$.

\subsection{Complete-covariate analysis}\label{sec:ccov_1}

%Proposition \ref{prop:ccov_1} below justifies the large-sample Wald-type inference based on $(\htscv, \hses^\cv)$.
Recall $\xicv = (x_{ij})_{j \in \mj}$ as the vector of complete covariates used by the complete-covariate analysis. The finite-population covariance of $(x_i^\ccov)_{i=1}^N$, denoted by  $S_{xx}^\ccov$,  
%Let $S_{xx}^\ccov = (N-1)^{-1}\sumi (x_i^\ccov - \bx^\ccov)(x_i^\ccov - \bx^\ccov)^\T$ be the finite-population covariance of $(x_i^\ccov)_{i=1}^N$. 
%We have $S_{xx}^\ccov$ 
has a finite limit under Condition \ref{asym_1} as long as $\mathcal J$ converges to a limit as $N$ tends to infinity. 
Let $v_{\md}^\cv $ and $\stt{\cv} $ be the analogs of $v_\md$ and $S^2_{\tau, \md}$ in Lemma \ref{lem::basic-x} defined over $\{Y_i(0), Y_i(1), x_i^\cv\}_{i=1}^N$ for $\md =  \fisher, \lin$. 
%We relegate their explicit expressions   to the Supplementary Material. 

\begin{proposition}\label{prop:ccov_1}
Assume complete randomization, Condition \ref{asym_1},   and the limit of $S_{xx}^\ccov$ is  positive definite with $\mathcal J$ converging to a limit. We have $\sqrt N(\htscv - \tau ) \rs \mn(0, v_{\md}^\cv  )$
%\begina
%\sqrt N(\htscv - \tau ) \rs \mn(0, v_{\md}^\cv  ) \qquad  \text{for} \quad \md =  \fisher, \lin
%\enda  
with  $v_{\lin}^\cv \leq v_{\fisher}^\cv   $ and  $v_{\lin}^\cv \leq v_\neyman$. 
%Further let $\hses^\cv$ be the robust standard error associated with $\htscv$ for $\mds$ from \eqref{eq:cv_fisher} and \eqref{eq:cv_lin}, respectively. 
In addition,  $N(\hses^\cv)^2 -  v_\md ^\cv  = \stt{\cv}  + \op$, where $\stt{\cv}  \geq 0$. 
%We relegate the explicit expressions of $v_{\md}^\cv $ and $\stt{\cv} $  to the Supplementary Material. 
\end{proposition}

Proposition \ref{prop:ccov_1} justifies the large-sample Wald-type inference based on $(\hts^\ccov, \hses^\ccov)$ regardless of whether Condition \ref{cond:cc} holds or not. 
This, together with the efficiency of $\htl^\ccov$ over $\htf^\ccov$ and the unadjusted $\htn$,   follows from applying Lemma \ref{lem::basic-x}  to the finite population of $\{Y_i(0), Y_i(1), x_i^\ccov\}_{i=1}^N$, and illustrates the advantage of including all units in the analysis even at the cost of discarding all information in the incomplete covariates. 
Importantly, all theoretical guarantees hold even when the missingness is related to the missing covariates and unobservable potential outcomes, a scenario analogous to missing not at random  under the super-population framework \citep{rubin1976inference}.  
%  the {\red missing-not-at-random missingness}

\cite{schemper1990} referred to the complete-covariate analysis as ``an even less justifiable method" than the complete-case  analysis.  Whereas their comment could be valid for observational studies when incomplete covariates include some key confounders, we give the opposite recommendation for randomization experiments. 
Intuitively, randomization precludes the possibility of confounding by enforcing independence between the treatment assignment and pretreatment covariates, ensuring valid simple comparisons even without covariate adjustment.
The exclusion of the incomplete covariates thus does not affect the validity of the complete-covariate analysis while allowing for additional efficiency over $\htn$ as long as one complete covariate is prognostic. 
We consider this as the baseline strategy for benchmarking the alternative strategies.

\subsection{Single imputation}\label{sec:imp}

%Recall $x_i^\imp (c)$  as the imputed covariates  under the covariate-wise single imputation. 
Under single imputation, a subtle point is that the imputed covariates $\xis(c)$ may depend on the treatment indicators via the choice of $c$ and are thus not necessarily true covariates in the strict sense. 
We focus on imputations $c = (c_j)_{j=1}^J$ with finite probability limits $\ci = (\cji)_{j=1}^J$: 
$$
\mathcal C = \{c \in \mathbb{R}^J:  \plim c_j = \cji <\infty\ (j=1,\ldots, J) \text{ under complete randomization and Condition \ref{asym_1}} \}.
$$
The constant imputation with $c_j = 0$ for $j = \ot{J}$  is a special case with $\ci = 0_J$. 
The unconditional sample mean imputation with  $c_j = \hx_j^\obs$  is  a special case with $\cji  = \lim_{N\rightarrow \infty} \sum_{i=1}^N A_{ij} x_{ij}/  \sum_{i=1}^N A_{ij}$.

%Proposition \ref{prop:imp_1} below justifies the large-sample Wald-type inference based on  $(\htsim(c), \hse _{\md}^\im (c) ) $. For $c\in \mc$ with $\plim c = \ci$,   
Let  $\sxximp  (c)$ be the finite-population covariance  of    $\{x_i^\imp  (c)\}_{i=1}^N$. 
Condition \ref{asym_1} ensures that $\sxximp  (c)$ has a finite limit for all $c \in \mathcal {C}$.  
Let $v_{\md}^\imp(\ci) $ and $\stt{\imp}(\ci)$ be the analogs of $v_\md$ and $S^2_{\tau, \md}$ in Lemma \ref{lem::basic-x} defined over $\{Y_i(0), Y_i(1), x_i^\imp(\ci)\}_{i=1}^N$  for $ \md =  \fisher, \lin$. 
%We relegate their explicit expressions   to the Supplementary Material.

\begin{proposition}\label{prop:imp_1}
Assume complete randomization, Condition \ref{asym_1}, and $c\in \mathcal C$ with the limit of $\sxximp (c)$ being positive definite.   
We have $\sqrt N\{\htsim(c) - \tau \} \rs \mn\{0, v_{\md}^\im (\cinf)\}  $
%\begina
%\sqrt N\{\htsim(c) - \tau \} \rs \mn\{0, v_{\md}^\im (\cinf)\}  \qquad  \text{for} \quad \md =  \fisher, \lin
%\enda
 with 
$v_{\lin}^\im (\cinf) \leq v_{\fisher}^\im (\cinf)$ and $v_{\lin}^\im (\cinf) \leq v_{\lin}^\ccov \leq v_\neyman$. 
%Further let  $\hse _{\md}^\im (c)$ be the robust standard error  associated with $\htsim(c)$  for $\mds$ from \eqref{eq:im_f} and \eqref{eq:im_l}, respectively. 
In addition, 
$N \{\hse _{\md}^\im (c)\}^2 - v_{\md}^\im (\cinf) = \stt{\imp}(\cinf)+ \op$, where $\stt{\imp}(\cinf)\geq 0$. 
 %We relegate the explicit expressions of  $v_{\md}^\im (\cinf)$  and $S^\im_{\tau,\md}(\cinf)$ to the Supplementary Material. 
\end{proposition} 

%%%%%%%%%%%%%%%%%%%%%%%%%%
Echoing the comments after Proposition \ref{prop:ccov_1}, Proposition \ref{prop:imp_1} 
justifies the large-sample Wald-type inference based on $\{\hts^\imp(c), \hses^\imp(c)\}$ for $\mds$ without  Condition \ref{cond:cc} or any restrictions on the dependence of $M_i$'s on $\{Y_i(0), Y_i(1), x_i\}_{i=1}^N$. 
Intuitively, the regularity conditions therein guarantee the imputed covariates act as standard covariates that satisfy Condition \ref{asym_basic} as $N$ tends to infinity, ensuring consistency, along with the  efficiency of $\htl^\imp(c)$ over $\htf^\imp(c)$, by Lemma \ref{lem::basic-x}. 
The efficiency of $\htl^\imp(c)$ over $\hts^\ccov \ (\mds)$ and $\htn$, on the other hand, follows from Lemma \ref{lem:eff_a} with $x_i^\imp(c) \supseteq x_i^\ccov \supseteq \emptyset$,  and illustrates the advantage of including all covariates in the analysis even after some basic imputation.

% namely $v_\md^\imp(\ci)$ and $N\{\hse _{\md}^\im (c)\}^2$
The true and estimated variances both depend on the choice of $c$. 
Computationally, we can minimize them over $c$ to obtain the optimal imputation. However, we do not go into details of this route because the imputation method is strictly dominated by the missingness-indicator method as shown in the next subsection.

\subsection{Missingness-indicator method}\label{sec:mim}
Recall $x_i^\mim  = ((x_i^0)^\T, M_i^\T)^\T$ as the covariates for forming regressions \eqref{eq:mi_f_0} and \eqref{eq:mi_l_0} under the missingness-indicator method. 
Let $\sxxmim$ be the finite-population covariance  of   $(x_i^\mim)_{i=1}^N$. 
It has a finite limit under Condition \ref{asym_1}. 
%Recall $\hses^\mim$ as the robust standard error associated with $\hts^\mim$ for $\mds$ from \eqref{eq:mi_f_0} and \eqref{eq:mi_l_0}, respectively. 
%Proposition \ref{prop:mim_1} parallels Propositions \ref{prop:ccov_1} and \ref{prop:imp_1}, and justifies the large-sample Wald-type inference based on $(\hts^\mim, \hses^\mim)$.
Let $v_{\md}^\mim $ and $\stt{\mim} $ be the analogs of $v_\md$ and $S^2_{\tau, \md}$ in Lemma \ref{lem::basic-x} defined over $\{Y_i(0), Y_i(1), x_i^\mim\}_{i=1}^N$ for $\md =  \fisher, \lin$. 
%We relegate their explicit expressions   to the Supplementary Material. 

\begin{proposition}\label{prop:mim_1}
Assume complete randomization, Condition \ref{asym_1},  and the limit of $\sxxmim$ is  positive definite.  We have $\sqrt N(\hts^\mim - \tau ) \rs \mn(0, v_{\md}^\mim   ) $
%\begina
%\sqrt N(\hts^\mim - \tau ) \rs \mn(0, v_{\md}^\mim   )  \qquad  \text{for} \quad \md =  \fisher, \lin
%\enda
 with  $v_{\lin}^\mim \leq v_\fisher^\mim$ and $v_\lin^\mim \leq v_\lin^\imp(\ci) \leq v_\lin^\ccov \leq v_\neyman$ for all $c \in \mathcal C$.
In addition, $N(\hses^\mim)^2 -  v_\md ^\mim   = S^\mim _{\tau,\md} + \op$, where $S^\mim _{\tau\tau,\md} \geq 0$. 
 %We relegate the explicit expressions of $v_{\md}^\mim  $ and $S^\mim _{\tau,\md}$  to the Supplementary Material. 
\end{proposition} 

%%%%%%%%%%%%%%%%%%%%%%%%%%
Similar to Propositions \ref{prop:ccov_1} and \ref{prop:imp_1}, 
Proposition \ref{prop:mim_1}  holds without
Condition \ref{cond:cc} or any restrictions on the dependence of $M_i$'s on $\{Y_i(0), Y_i(1), x_i\}_{i=1}^N$. 
This echoes the observation by \cite{review} under the super-population framework.
Complete randomization balances both covariates and missingness indicators across treatment groups, and ensures consistency of the regression adjustments based on them irrespective of the missingness mechanism.
The efficiency of $\htl^\mim$ over $\htl^\imp(c)$, $\htl^\ccov$, and $\htn$, on the other hand, follows from Lemma \ref{lem:eff_a} with $ x_i^\mim(c) \supseteq x_i^\imp(c)  \supseteq x_i^\ccov \supseteq \emptyset $, and highlights the  second advantage of including the missingness indicators in addition to the invariance to the imputed values.

The model-based literature  holds different opinions on the consistency of $\htf^\mim$.  \cite{jones} assumed
that the outcomes follow the classic Gauss--Markov model $Y_i = \mu + Z_i \beta + x_i \gamma + \ep_i$ with $\beta$ being the constant treatment effect and thus the model-based analog of $\tau$. He showed that $\htf^\mim$ is  biased for estimating $\beta$ even asymptotically. 
The discrepancy between  \cite{jones} and Proposition \ref{prop:mim_1} is due to the difference in assumptions on the data-generating process and  source of randomness. In particular, \cite{jones} assumed the Gauss--Markov model as the correct model whereas Proposition \ref{prop:mim_1}  requires no Gauss--Markov assumptions and holds even if the linear models are misspecified.
In addition,  \cite{jones} conditioned on the treatment indicators $Z_i$'s whereas Proposition \ref{prop:mim_1} averages over them.

\subsection{Missingness-pattern method}
\label{sec:mp_asym}

\subsubsection{Conditional properties under post-stratification}
Recall $\mmp  = \{m: N_{(m)} > 0\}$ as the set of missingness patterns present in the study population, with $\rho_{(m)} = N_{(m)}/N$ as the proportion of units with missingness pattern $m$. 
Let $\hse_{\md,(m)}$ be the robust standard error associated with $\htsm $ from the pattern-specific {\ols} under missingness pattern $m \in \mathcal M$, where $\mds$. 
As in the definition of $\htau_{\md,(1_J)}$, 
let $\hse _{\fisher,(1_J)} = \hse_{\lin,(1_J)} = \hse_{\neyman, (1_J)}$ be the robust standard error associated with $\htau_{\neyman, (1_J)}$ from $Y_i \sim 1+Z_i$ over  $\{i: M_i = 1_J\}$.
The weighted average 
\beginy \label{eq:hse_mp}
(\hse _\md ^\mp)^2  = \summ \rho_{(m)}^2 \hsesq_{\md, (m)}
\endy  affords an intuitive estimator of the sampling variance of $\htau_\md ^\mp    $  from \eqref{eq:htau_mp} for $\mds$.
%= \sum_{m \in \mm} \rho_{(m)} \htsm

Recall $N_{(m), z}$ as the number of units with missingness pattern $m$ that receive treatment $z \in \{0,1\}$. 
Conditioning on $\mathcal D = \{N_{(m), z}: m \in \mmp , \ z = 0, 1\}$ with $N_{(m), z} > 0$   for all $m \in \mmp $ and $z = 0,1$, we have $ | \mathcal{M} | $ independent completely randomized experiments, one within each missingness pattern \citep{luke}. Under regularity conditions within each missingness pattern, Lemma \ref{lem::basic-x} ensures the asymptotic normality of $\htsm $, the asymptotic efficiency of $\htau_{\lin,(m)}$ over $\htau_{\fisher, (m)}$, and the asymptotic conservativeness of $\hse_{\md,(m)}$ for estimating the true standard error of $\htsm$ for $m\in \mathcal{M}$. Consequently, $\htau_\md ^\mp$ is also asymptotically normal because it is a linear combination of the missingness-pattern-specific estimators, and moreover, 
$\htl^\mp$ is asymptotically more efficient than $\htf^\mp$ with $\hse _\md ^\mp$'s being asymptotically conservative for their respective true standard errors.

The above conditional theory for the missingness-pattern method is straightforward and elegant. A fair comparison with other methods, however, requires quantification of its  asymptotic behaviors without conditioning on $\mathcal D$. This is our goal for the next subsubsection.

\subsubsection{Unconditional properties via aggregate regression}

The unconditional theory becomes intuitive if we can rewrite $(\htau_\md ^\mp, \hse _\md ^\mp)$ as  outputs from one aggregate regression as hinted at by Section \ref{sec:mp}. 
 We formalize below the intuition for $(\htl^\mp, \hse_\lin^\mp)$ and relegate the analogous results for $(\htf^\mp,  \hse_\fisher^\mp)$ to the Supplementary Material.  
 % to Section \ref{sec:mp_aggregate_app} in
 
Let $u_i^\mp(c)$ be the covariate vector that includes $x_i^\imp(c)$, $M_{i1}, \ldots, M_{iJ}$, and all their interactions up to some adjustment for collinearity. 
% ---  namely all terms of the form $ (\prod_{j \in \mathcal J'} M_{ij})$ and $x_i^\imp(c) (\prod_{j \in \mathcal J'} M_{ij})$ with $\mathcal J'$ being an arbitrary subset of $\{\ot{J}\}$ ---  up to some adjustment for collinearity.
A key observation is that $ u_i^\mp(c)$ includes $ x_i^\mim(c) = \{x_i^\imp(c), M_{i1}, \ldots, M_{iJ}\}$ as a subset.
%\beginy\label{eq:mp_mim_cim}
%u_i^\mp(c) \supseteq x_i^\mim(c).
%\endy
We give the explicit forms of $u_i^\mp(c)$ for $J = 1, 2$ in Examples \ref{ex:u_mp_1} and \ref{ex:u_mp_2} (continued) below, and then state its utility for recovering $(\htl^\mp,  \hse_\lin^\mp)$ via one aggregate regression in Proposition \ref{prop:mp_agg}.

\setcounter{example}{0}
\begin{example}[continued] \label{ex:u_mp_1}
For $J = 1$ with $M_i = M_{i1}$ and $x_i^\imp(c) =c M_i  +x_i (1-M_i)$, we have 
\begina
u_i^\mp(c) = \{x_i^\imp(c), M_i, x_i^\mp(c) M_i\} = \{x_i^\imp(c), M_i\} 
\enda 
given $x_i^\imp(c) M_i = c M_i$ is collinear with $M_i$. This ensures $u_i^\mp(c) = x_i^\mim(c)$. 
\end{example}

\begin{example}[continued]\label{ex:u_mp_2}
For $J = 2$ with $M_i = (M_{i1}, M_{i2})^\T$ and $x_i^\imp(c) = (x_{i1}^\imp(c_1), x_{i2}^\imp(c_2))^\T$, we have
\beginy\label{eq:u_mp_2}
u_i^\mp(c) &=& \{ x_i^{\text{imp}}(c) , M_{i1}, M_{i2}, M_{i1} M_{i2},   x_i^{\text{imp}}(c)M_{i1},   x_i^{\text{imp}}(c)M_{i2},  x_i^{\text{imp}}(c) M_{i1} M_{i2}\}\nonumber\\
&=&  \{ x_i^{\text{imp}}(c) , M_{i1}, M_{i2}, M_{i1} M_{i2}, x_{i1}^{\text{imp}}(c_1)M_{i2},   x_{i2}^{\text{imp}}(c_2)M_{i1}\}
\endy
given $x_{ij}^\imp(c_j) M_{ij} = c_jM_{ij}$ is collinear with $M_{ij}$ for $j = 1, 2$ and $x_i^{\text{imp}}(c) M_{i1} M_{i2} = (c_1, c_2)^\T M_{i1}M_{i2}$ is  collinear with $M_{i1}M_{i2}$. 
This ensures $u_i^\mp(c)$ includes $x_i^\mim(c)  =(\{x_i^\imp(c)\}^\T, M_{i1}, M_{i2})^\T$ as a subset. 
\end{example}

\begin{proposition}\label{prop:mp_agg}
%Let $\bar u^\mp(c) = \meani u_i^\mp(c)$ be the average of $u_i^\mp(c)$ over $i = \ot{N}$. 
The fully interacted missingness-pattern estimators
$\htl^\mp$ and $\hse_\lin^\mp$ from \eqref{eq:htau_mp} and \eqref{eq:hse_mp} equal the coefficient of $Z_i$ and its associated robust standard error from 
 \beginy\label{eq:mp_agg}
Y_i \sim 1 + Z_i + \{u_i^\mp(c) - \bar u^\mp(c)\} + Z_i\{u_i^\mp(c) - \bar u^\mp(c)\}  %\qquad \text{over}\quad  i = \ot{N} 
\endy
over $i = \ot{N} $. The result holds
for arbitrary $c \in \mathbb R^J$. 
\end{proposition}

Refer to \eqref{eq:mp_agg} as the fully interacted aggregate specification for the missingness-pattern method. 
Proposition \ref{prop:mp_agg} shows $(\htl^\mp, \hse_\lin^\mp)$ as direct outputs of its {\ols} fit, with $u_i^\mp(c)$ as the effective covariate vector analogous to $x_i^\mim(c)$, $x_i^\imp(c)$, and $x_i^\ccov$. 
This ensures the equivalence of $\htl^\mp$ and $\htl^\mim$ when $J = 1$, and allows us to derive the unconditional asymptotic properties of $\htl^\mp$.
%given Example \ref{ex:u_mp_1} (continued) above without post-stratification

\begin{corollary}\label{cor:mp}
For $J = 1$, it follows from $u_i^\mp(c) = x_i^\mim(c)  $ that $\htl^\mp = \htl^\mim  $. 
\end{corollary}

Let $u_i^\mp = u_i^\mp(0_J)$ be the value of $u_i^\mp(c)$ at $c = 0_J$. % with finite-population mean $\bar u^\mp$ and covariance $S_{uu}^\mp$. 
% $\bar u^\mp(z) = \meani u_i^\mp(z)$. 
%Let $S_{uu}^\mp$  be the finite-population covariance of $(u_i^\mp)_{i=1}^N$. 
%Condition \ref{asym_u} below extends Condition \ref{asym_1} and ensures $(u_i^\mp)_{i=1}^N$ satisfy the regularity conditions required by Lemmas \ref{lem::basic-x} and \ref{lem:eff_a}.
%This ensures the asymptotic normality and efficiency of $\htl^\mp$ without post-stratification. 
%We formalize the intuition in   Propositions  \ref{prop:mp_unconditional} and \ref{prop:eff_mp} below. 
%\begin{condition}\label{asym_u}
%As $N \to \infty$, 
%(i) $e_z  $ has a limit in $  (0,1)$ for $z =  0,1$, 
%(ii) the first two moments of $\{Y_i(0), Y_i(1), u^\mp_i\}_{i=1}^N $  have finite limits;
%the limit of $S_{uu}^\mp$ is positive definite, and 
%(iii) %$\max_{1\leq i \leq N; \ z = 0, 1} Y_i^2(z)/N = o(1)$  and $\max_{1\leq i \leq N} \|x_i\|_2^2 /N = o(1)$.
%there exists a  $u_0 < \infty$ independent of $N$ such that $\meani Y_i^4(z) \leq u_0$ and $\meani \|x_i\|_4^4 \leq u_0$ {\red [--- not a typo, bounded $\meani \|x_i\|_4^4$ ensures $\meani \|u_i\|_4^4$ ---]}.
%\end{condition}
Let $v_\lin^\mp $ and $S^\mp_{\tau\tau,\lin} $ be the analogs of $v_\lin$ and $S^2_{\tau, \lin}$ in Lemma \ref{lem::basic-x} defined over $\{Y_i(0), Y_i(1), u_i^\mp\}_{i=1}^N$.

\begin{proposition}\label{prop:mp_unconditional}
Assume complete randomization and Condition \ref{asym_basic} for  $\{Y_i(0), Y_i(1), u^\mp_i\}_{i=1}^N $.
 We have 
$\sqrt N(\htl^\mp - \tau ) \rs \mn(0, v_\lin^\mp) .$
In addition,  $N(\hse^\mp_\lin)^2 -  v_\lin ^\mp  = S_{\tau\tau, \lin}^{\mp}+ \op$, where $S_{\tau\tau, \lin}^{\mp}\geq 0$.
\end{proposition}

Proposition \ref{prop:mp_unconditional} is a direct consequence of Lemma \ref{lem::basic-x} and Proposition \ref{prop:mp_agg}, and justifies the large-sample Wald-type inference based on $(\htl^\mp, \hse_\lin^\mp)$ irrespective of the missingness mechanism. 
%A caveat is that the analogous additive aggregate specification for recovering $(\htf^\mp, \hse_\fisher^\mp)$ is not just $Y_i \sim 1 + Z_i + u_i^\mp(c)$ but a different restricted variant of \eqref{eq:mp_agg} after removing interactions between $Z_i$ and $x_i^\mp(c)$;  we relegate the details to the Supplementary Material.
The asymptotic efficiency of saturated model over its restricted variants ensures the asymptotic efficiency of $\htl^\mp$ over $\htf^\mp$. 
Heuristically, the unconditional variance $\var_\infty(\hts^\mp)$ is close to $\var_\infty(\hts^\mp \mid \mathcal{D})$ up to some higher order terms \citep{HoltSmith, luke}.
% because
%$$
%\red \var(\hts^\mp) = E\{  \var(\hts^\mp \mid \mathcal{D})\}+\var\{E(\hts^\mp \mid \mathcal {D})\}
%\approx   \var_\infty(\hts^\mp \mid \mathcal{D}) \qquad  \text{for} \quad \md =  \fisher, \lin .
%$$
%{\blue [--- LHS $\var(\hts^\mp)$ no $\infty$ but RHS $\var_\infty(\hts^\mp \mid \mathcal{D})$ has $\infty$? how best to annotate the $\infty$? why OK to let go of the outer E in $E\{  \var(\hts^\mp \mid \mathcal{D})\} \approx \var_\infty(\hts^\mp \mid \mathcal{D})$? ---]}
%See \citet{HoltSmith} and \citet{luke} for related discussions. 
Therefore, the asymptotic efficiency of $\htl^\mp$ over $\htf^\mp$ follows from $ \var_\infty(\htl^\mp \mid \mathcal{D}) \leq  \var_\infty(\htf^\mp \mid \mathcal{D})$ by established results under post-stratification. 

%$\var(\hts^\mp) = E\{  \var(\hts^\mp \mid \mathcal{D})\}+\var\{E(\hts^\mp \mid \mathcal {D})\} $ suggests that $ \var_\infty(\hts^\mp) \approx  E\{\var_\infty(\hts^\mp \mid \mathcal{D})\}$ for $\mds$. 
%The asymptotic efficiency of $\htl^\mp$ over $\htf^\mp$ follows from $N\var_\infty(\htl^\mp \mid \mathcal{D}) \leq N\var_\infty(\htf^\mp \mid \mathcal{D})$ by established results under post-stratification.}
%
%\td
%\begini
%\item \red move the efficiency comparison of $\htl^\mp$ and $\htf^\mp$ forward --- need to state the efficiency of $\htl^\mp$ over $\htf^\mp$ first to claim $\htl^\mp$ is the best among the eight? 
%\endi
%\td

Proposition \ref{prop:eff_mp}, on the other hand,  gives the efficiency hierarchy between  $\htl^\dg$'s for the {\ncss}  consistent strategies, $\dg \in\{ \ccov, \imp, \mim, \mp\}$.

\begin{proposition}\label{prop:eff_mp}
%Assume complete randomization, Condition \ref{asym_u}, and $c\in\mathcal C$ with the limits of $S_{xx}^\cc$, $S_{xx}^\imp(c)$, and $S_{xx}^\mim$  all positive definite. 
The asymptotic variances of $ \htl^\dg  $ satisfy  
$
v_\lin^\mp \leq  v_\lin^\mim \leq v_\lin^\imp(c_\infty) \leq v_\lin^\cv \leq v_\neyman. 
$
\end{proposition}

Proposition \ref{prop:eff_mp} follows from $u_i^\mp(c) \supseteq x_i^\mim(c)  \supseteq x_i^\imp(c) \supseteq x_i^\ccov \supseteq \emptyset$ and, together with the efficiency of $\htl^\dg$ over $\htf^\dg$ for each individual strategy, ensures the asymptotic efficiency of $\htl^\mp$ among all {\ncssest} consistent estimators in Table \ref{tb:strategies}. 
Compare the definitions of $u_i^\mp(c)$ with $x_i^\mim(c)$ to see that $u_i^\mp(c)$ includes interaction terms like $x_i^\imp(c) M_{ij}$, $x_i^\imp(c) M_{ij}M_{ij'}$, etc. that are not in $x_i^\mim(c)$.  
Intuitively, this suggests the advantage of $\htl^\mp$ over $\htl^\mim$ whenever the covariates interact with the missingness pattern in affecting the treatment effect.  
%We illustrate the intuition using simulation in Section \ref{sec:simu}.  
%{\blue Computationally, the matrix of $u_i^\mp(c)$'s, denoted by $U^\mp(c) = (u_1^\mp(c), \ldots, u_N^\mp(c))^\T$,  can be obtained by one line of R code as
%\begin{center}
%\texttt{model.matrix($\sim$ X$^\imp$(c) * M$_{\cdot1}$* $\cdots$ * M$_{\cdot\texttt{J}}$)}
%\end{center}
%where $X^\imp(c) =(x_1^\imp(c), \ldots, x_N^\imp(c))^\T $ and $M_{\cdot j} = (M_{1j}, \ldots, M_{Nj})^\T$ denote the concatenations of $x_i^\imp(c)$ and $M_{ij}$ over $i = \ot{N}$, respectively.}

\subsection{Summary and a hierarchy of model specifications}\label{sec:unify}

Under Condition \ref{cond:a}, Sections \ref{sec:ccov_1}--\ref{sec:mp_asym} establish the validity of $(\htss, \hses^\dg)$ for $\dg \in\{ \ccov, \imp, \mim, \mp\}$ and $\mds$ regardless of the relation between $M_i$ and $\{Y_i(0), Y_i(1), x_i\}$, the correctness of the linear models, and the choice of the imputed values.
The results, though much more general than one might expect, 
are no surprise but a direct implication of Lemma \ref{lem::basic-x}.
% due to \cite{Lin13}.
%may not be surprising after all.  A number of previous works have shown the consistency of additive and fully interacted regressions under mild regularity conditions on the covariates irrespective of whether the models are correctly specified or not; see, for example, \cite{Lin13} and \cite{ZDfrt}.  Condition \ref{asym_1} ensures that the covariate vectors used in these {\ncss}  strategies, namely $x_i^\dg$, all satisfy these regularity conditions, and thereby guarantees consistent estimators of $\tau$ irrespective of whether the corresponding models are correctly specified or not. with $M_i$ unaffected by the treatment
Consistency as such is a rather weak criterion for evaluating the performance of regression estimators, rendering the possible inconsistency under the complete-case analysis  all the more undesirable.

The intuition from Lemma \ref{lem:eff_a} further allows to us to quantify the asymptotic efficiency between the {\ncss}  consistent methods. 
In particular, Proposition \ref{prop:eff_mp} gives the efficiency hierarchy between  $\htl^\dg$'s for the {\ncss}  consistent strategies, $\dg \in\{ \ccov, \imp, \mim, \mp\}$. This, together with the efficiency of $\htl^\dg$ over $\htf^\dg$ for each individual strategy, ensures the asymptotic efficiency of $\htl^\mp$ among all {\ncssest} consistent estimators in Table \ref{tb:strategies}.

This concludes our discussion on the design-based properties of the {\ntotest} estimators in Table \ref{tb:strategies}. 
The $\htl^\mp$ under the missingness-pattern method effectively uses a fully interacted {\ols} with covariates including $x_i^\imp(c), M_{i1}, \ldots, M_{iJ} $, and all their interactions, 
% all higher order interactions between $\{Z_i, x_i^\imp(c), M_{i1}, \ldots, M_{iJ} \}$,   in finite samples. 
and ensures asymptotic efficiency at the cost of being the most demanding on pattern-specific sample sizes.
The missingness-indicator method simplifies the specification by excluding interactions between $\{x_i^\imp(c), M_{i1}, \ldots, M_{iJ} \}$. 
%The complete-case-indicator method simplifies the specification by using $C_i = 1 - \prodj (1-M_{ij})$ to summarize the joint effect of $\{M_{i1}, \ldots, M_{iJ}\}$. 
The single imputation simplifies the specification by discarding all terms involving the missingness indicators.  
The complete-covariate analysis simplifies the specification by further discarding all dimensions in $x_i^\imp(c)$ that involve imputed covariates. 
This unifies the complete-covariate analysis, the single imputation, and the missingness-indicator method as various restricted variants of the missingness-pattern method.
% and, echoing the discussion in Remark \ref{rmk:mim}, implies a whole range of alternative simplifications possible along the spectrum. 
%We discuss the issue in more detail in Section \ref{sec:ext}.  

\section{Simulation}\label{sec:simu}
We now turn to simulation to illustrate the finite-sample performance of the proposed strategies.
Consider a treatment-control experiment with $(N_0, N_1) = (0.8 N, 0.2 N)$. 
For each $i$, independently draw a latent indicator, $\xi_i \sim$ Bernoulli$(0.2)$,  to divide the units into two latent classes, the ``severely ill" group with $\xi_i = 1$ and the ``less ill" group with $\xi_i = 0$; independently draw a $J=3$ dimensional covariate vector $x_i =  (x_{i1}, x_{i2}, x_{i3})^\T   \sim \mn( \xi_i 1_{J}, I_J)$. 
Assume Condition \ref{cond:a} with covariate 1 being the only complete covariate. 
We 
generate $M_i = (M_{i1}, M_{i2}, M_{i3})^\T$  independently for $i = \ot{N}$ as
$M_{i1} = 1$ and $ M_{ij}  \sim \textup{Binomial}\{0.1 \xi_i + 0.05(1-\xi_i)\}$ for $j = 2, 3$, and 
consider three scenarios for generating the potential outcomes   to highlight different aspects of the theoretical results.  

Scenario (i) sets $N = 500$ and generates the potential outcomes as independent normals as $Y_i(z)  \sim \mathcal{N}( 5\xi_i + 2 x_i^\T \gamma_{z | \xi_i}, 1)$ with 
$(\gamma_{1|1}, \gamma_{0|1}) =(1_J, -1_J) $  and $(\gamma_{1|0}, \gamma_{0|0}) =(0.5 \cdot 1_J, - 0.5 \cdot 1_J) $. 
The data-generating process ensures that the severely ill group has both a higher chance of missing covariates  and on average greater values of covariates and potential outcomes. 
This exemplifies the case where missingness is correlated with covariates and potential outcomes and in general leads to unequal $\tau^\cc $ and $\tau^\uc $.  
For illustration simplicity, we center the $\{Y_i(z)\}_{i=1}^N$ for $z = 0,1$, respectively, to have $\tau = 0$. 
The true subgroup average causal effects among the complete and incomplete cases equal $\tau^\cc  =-0.402 $ and $\tone = 1.393$ based on the simulated data.
We use this scenario to illustrate the inconsistency of complete-case analysis when $\tz \neq \tone$. 

Fix $\{Y_i(0), Y_i(1), M_{i}, x_i\}_{i=1}^N$ in the simulation. We draw a random permutation of $N_1$ $1$'s and $N_0$ $0$'s to obtain the completely randomized assignment, and then use the resulting observed outcomes to compute the estimators in Table \ref{tb:strategies}.
% proceed with the regression analysis under the {\ntot}  missing-covariate strategies, namely the complete-case analysis, the complete-covariate analysis, single imputation, the missingness-indicator method,  and the missingness-pattern method, respectively, each with both the additive and fully interacted formulations. 
%This gives us a total of {\ntotest}  regression estimators as summarized in Table \ref{tb:strategies}.
We set the imputation values at $c = 0_J$ for single imputation. 
The results under the unconditional sample mean imputation with $c_j = \hat x_j^\obs$ are similar and thus omitted. 
We do not include the unadjusted $\htn$ from the simple regression $Y_i \sim 1+Z_i$  due to space limit. Its inferiority to the fully interacted complete-covariate analysis  is an established result. %; see, for example, \cite{Lin13, ZDfrt}. 
Figure \ref{fig}(i) shows the distributions of the ten estimators over 1,000 independent assignments under scenario (i).
The complete-case analysis yields biased inferences whereas all the other {\ncss}  methods are consistent. 
The efficiency of the fully interacted regressions over their respective additive counterparts is coherent across different strategies except the missingness-pattern method. 
The long tails of $\htl^\mp$ are not surprising but the consequence of small sample sizes under a subset of missingness patterns.
%{\blue In particular, the sample sizes under the four missingness patterns are $N_{(000)} = 388$, $N_{(001)} = 49$, $N_{(010)} = 39$, and $N_{(011)} = 24$, respectively. [--- this is not that small?]}

Scenario (ii) inherits most settings from scenario (i), yet generates the potential outcomes as $Y_i(z) \sim \mathcal{N}\{\mu_i(z), 1\}$, where $\mu_i(z) = 5\xi_i + x_i^\T \gamma_{z | \xi_i} + 2 M_{i}^\T 1_J$,  prior to centering, rendering   $M_i$ an important predictor of the potential and observed outcomes. 
Figure \ref{fig}(ii) shows the distributions of the resulting estimators over 1,000 independent assignments.
The improvement of $\htl^\mim$ under the missingness-indicator method  over $\htl^\imp(c)$ under single imputation in terms of efficiency is visible, illustrating the benefit of augmenting the regression analysis with information in $M_i$'s.

The benefits of missingness-pattern method, on the other hand, manifest in large samples when the interactions between $x_i$ and $(M_{i2}, M_{i3})$ are non-negligible.
Scenario (iii) inherits most settings from scenario (i), yet generates the potential outcomes as $Y_i(z) \sim \mathcal{N}\{\mu_i(z), 1\}$, where $\mu_i(z) = 5\xi_i + x_i^\T \gamma_{z | \xi_i} +  M_{i}^\T 1_J + M_{i2}M_{i3} + 5M_{i2} \sum_{i=1}^3 x_{ij}$,  prior to centering for $N = 10,000$ units. 
The effect of covariates and missingness indicators on the potential outcomes is no longer additive but involves interaction terms like $M_{i2}M_{i3}$ and $5M_{i2} \sum_{i=1}^3 x_{ij}$.  
 Figure \ref{fig}(iii) shows the distributions of  the resulting estimators over 1,000 independent assignments. 
The improvement of $\htl^\mp$ over $\htl^\mim$ in terms of efficiency is visible.

\begin{figure}\caption{\label{fig}Violin plots of $\htss$'s over 1,000 independent assignments under scenarios (i)--(iii). The estimators are labeled as ``$\dg.\md$" along the x-axis. The true $\tau$ is 0.}
\begin{center}
(i) Scenario (i) with $N = 500$
\includegraphics[width=.8\textwidth]{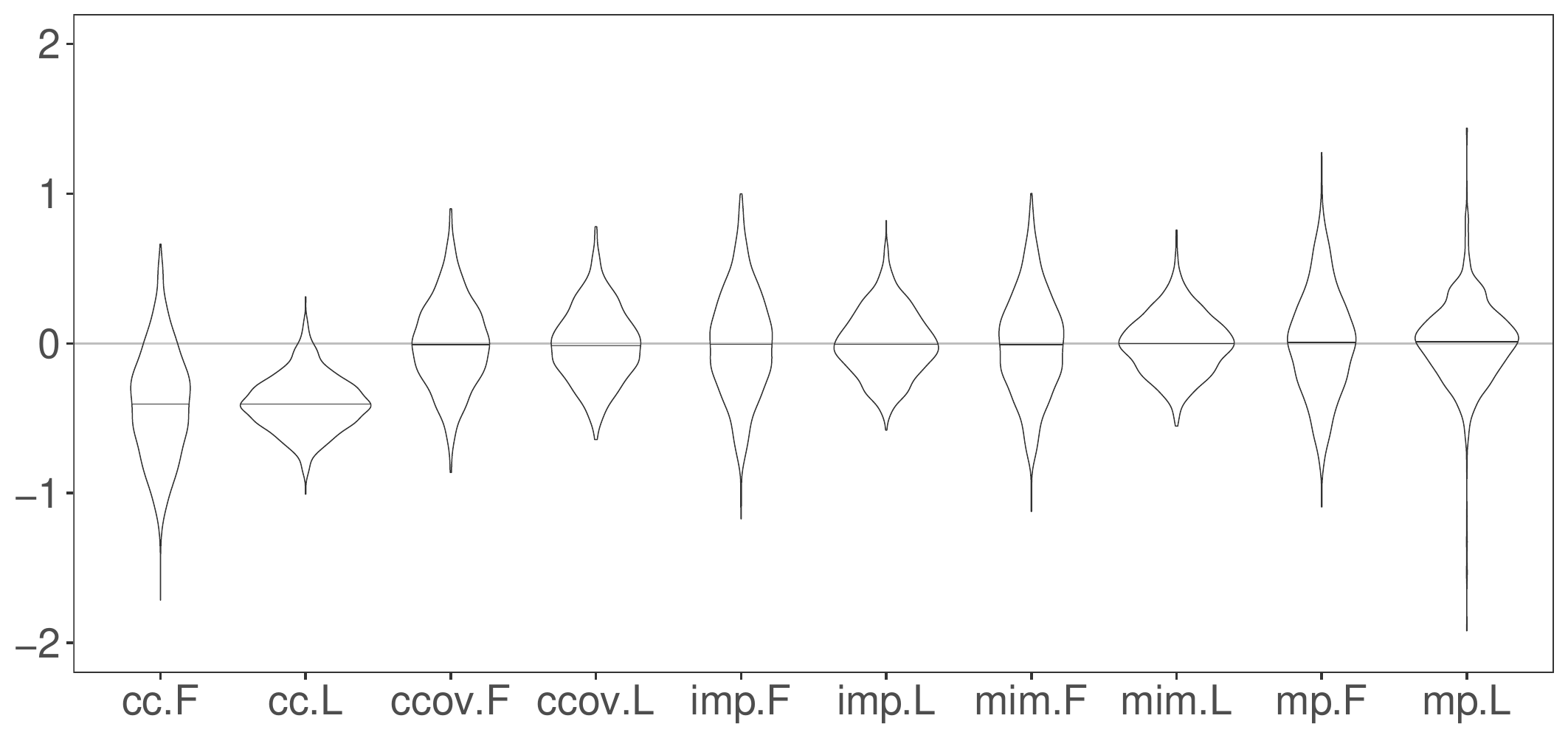}
\end{center}
\begin{center}
(ii) Scenario (ii) with $N = 500$
\includegraphics[width=.8\textwidth]{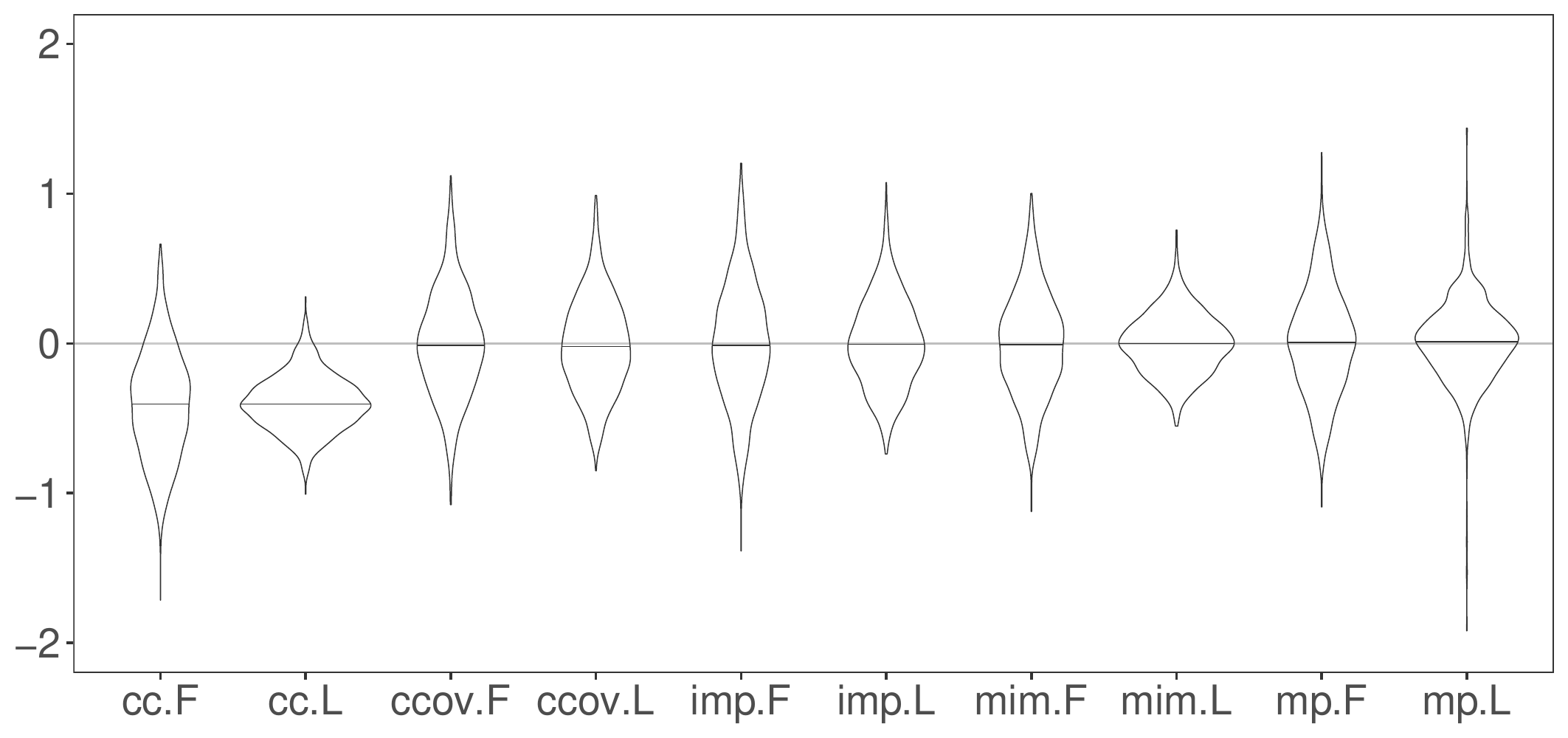}
\end{center}
\begin{center}
(iii) Scenario (iii) with $N = 10,000$
\includegraphics[width=.8\textwidth]{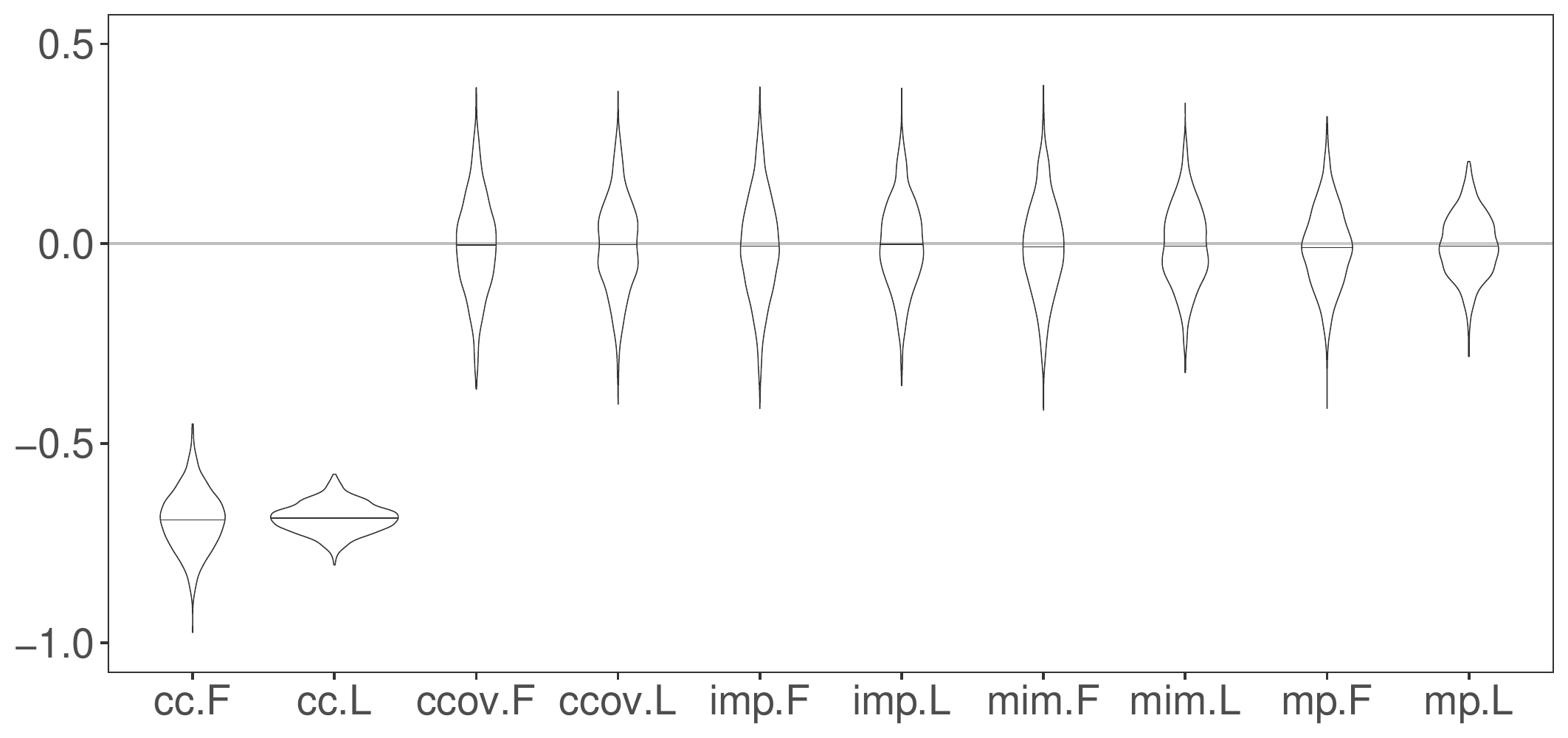}
\end{center}
\end{figure}

\section{Application}\label{sec:application}
%\cite{duflo} conducted a randomized experiment on 877 (season 1) / 757 (=179+208+135+133+102 for season 2 if I understood it correctly) farmers to study the effect of a time-limited utility cost reduction on fertilizer use.

\cite{duflo} conducted a randomized field experiment in Western Kenya to study the effect of a time-limited utility cost reduction on fertilizer use.
The experiment took place over two seasons from 2003 to 2005 after a series of small-scale pilot programs. %, and implemented the reduction in  cost, as the intervention of interest, by offering free delivery of fertilizer. 
We use their data from the first season to illustrate the methods for handling missing covariates. 

  % after stratification by the school and class of their children
The first season began after the 2003 short rain harvest to facilitate fertilizer purchase for the 2004 long rain season.
The treatment was implemented in the form of free delivery of fertilizer as a time-limited reduction in the cost to acquiring fertilizer. 
One major outcome of interest is the fertilizer use in the first season after long rains 2004. 
Farmers were randomly assigned to receive either the treatment  or control. 
A total of $N=877$ participant farmers were tracked for a follow-up usage survey, among which $N_1 = 204$ received the treatment and the rest $N_0 = N - N_1 = 673$ received control.
They constitute the study population of the original analysis reported by \cite{duflo}. 

We follow \cite{duflo} to consider $J = 27$ covariates, including  educational attainment, previous fertilizer usage, gender, income, whether the farmer's home has mud walls, a mud floor, or a thatch roof, and whether the farmer has received a starter kit in the past. 
A total of $N^\cc = 716$ farmers have all covariates observed, accounting for $81.64\%$ of the study population. 
Missingness of covariates happens in 7 out of the $J = 27$ covariates with a total of 9 missing patterns summarized in Table \ref{tb:mp}.
\begin{table}[t]
\caption{\label{tb:mp}Counts of the nine missingness patterns for the incomplete covariates}
% present among the $N = 977$ farmers under study. To simply the notation, the missingness patterns are displayed with only dimensions corresponding to the incomplete covariates. Four out of nine strata have sample sizes less than 5.}
\centering
\begin{tabular}{|c|r|}
  \hline
%present missingness pattern&   \\ 
%(based on incomplete covariates) & count\\
 missingness pattern & count\\
  \hline
0000000 & 716 \\ 
  0000100 &  59 \\ 
  0001000 &   1 \\ 
  0010011 &   2 \\ 
  0100000 &  19 \\ 
  0100100 &   1 \\ 
  1000000 &   1 \\ 
  1011111 &  71 \\ 
  1111111 &   7 \\ 
   \hline
\end{tabular}
\end{table}

Table \ref{tb:application} summarizes the results from our re-analysis of the data. 
We exclude the missingness-pattern method from the comparison due to the small sizes of 4 missingness-pattern strata with $N_{(m)} < 5$. 
%[-- Or, we implement modification by combining them with other strata?? say, some can be combined to the `1011111` others to `1111111`.]}
The  remaining four strategies return coherent results about a significant and fairly sizable effect of cost reduction on  fertilizer use. 
The complete-case analysis with \cite{Fisher35}'s specification corresponds to the original analysis in \cite{duflo}, and yields the largest point estimate, 0.144, overall. This, together with the also large result from \cite{Lin13}'s model, 0.135, suggests the possibility of $\tau^\cc > \tau$ in the study population.  
The complete-covariate analysis, on the other hand, yields the most conservative $p$-values, which is coherent with our theory.
%For comparison, the point estimate, robust standard error, and large-sample $p$-value from the unadjusted regression $Y_i \sim 1 + Z_i$ equal $\htn = 0.112$, $\hse_\neyman = 0.039$, and $p = 0.004$, respectively. 
%{\red [--- how to explain $\htl^\ccov$ yields higher p-value than $\htn$? --- ]} 

\begin{table}[t] \caption{\label{tb:application}Re-analysis of the data from \cite{duflo}.
For comparison, $\htn = 0.112$ and $\hse_\neyman = 0.039$ with  $p$-value $ = 0.004$. We exclude the missingness-pattern method.
}
%estimators and robust standard errors from the complete-case analysis, complete-covariate-analysis, single imputation, the missingness-indicator method, and the complete-case-indicator method. Columns ``$p$-value" report the $p$-values from large-sample approximations of the robust $t$-statistics with significance codes 0 `***' 0.001 `**' 0.01 `*' 0.05 `$.$' 0.1 ` ' 1.}
\begin{center}
\begin{tabular}{|c|r|r|r|r|r|r|}
  \hline
   &\multicolumn{3}{c|}{ $\hat\tau_\fisher^\dagger$ } & \multicolumn{3}{c|}{  $\hat\tau_\lin^\dagger $ }\\\hline
strategy & estimate & robust s.e. & $p$-value &  estimate & robust s.e.  & $p$-value  \\ 
  \hline
cc & 0.144 & 0.04 & 0  & 0.135 & 0.041 & 0.001  \\ 
  ccov & 0.106 & 0.038 & 0.006  & 0.096 & 0.038 & 0.012   \\ 
  imp & 0.129 & 0.037 & 0   & 0.123 & 0.037 & 0.001  \\ 
  mim & 0.132 & 0.037 & 0  & 0.119 & 0.037 & 0.001 \\ 
   \hline
\end{tabular}
\end{center}
\end{table}

\section{Extensions}\label{sec:ext}

\subsection{Other regression specifications}

Recall the hierarchy of model parsimony from Section \ref{sec:unify}.
The missingness-pattern method gives the most saturated model under the fully interacted aggregate specification \eqref{eq:mp_agg}, and has the complete-covariate analysis, single imputation, and the missingness-indicator method all as its restricted variants.
% with $u_i^\mp(c) \supseteq x_i^\mim(c) \supseteq x_i^\imp(c) \supseteq x^\ccov_i$. 
An immediate implication is that any subset of $u_i^\mp(c)$, up to a non-degenerate linear transformation, affords a valid covariate vector to form the additive and fully interacted regressions when simplifications are needed. 
This implies a whole spectrum of possible restrictions on \eqref{eq:mp_agg} depending on the nature of the data. We illustrate below two types of restrictions of possible theoretical and practical interests. 

The first type of restrictions are large-sample inference oriented, and focus on trade-offs between the missingness-indicator method and the missingness pattern method when the former alone is inadequate for attaining the desired asymptotic efficiency. 
In cases where evidence suggests the covariates and missingness indicators affect the treatment effects interactively, an intuitive trade-off between the missingness-indicator method and the missingness-pattern method would be 
to use $x_i^0$, $M_i$, and all second-order interactions between $\{x_i^0, M_{ij}: j = \ot{J}\}$, namely terms like $x_i^0 M_{ij}$ and $M_{ij}M_{ij'}$ for all $j \neq j' \in \overline{\mathcal J}$, to form the covariate vector. The resulting estimator has higher asymptotic efficiency than $\htl^\mim$ under the fully interacted specification. 
%The resulting covariate vector, denoted by $x_i'$,  satisfies $u_i^\mp  \supseteq x_i' \supseteq x_i^\mim$, and ensures higher asymptotic efficiency than $\htl^\mim$ under the fully interacted specification. 
 
%%%%%%%%%%%%%%
The second type of restrictions are finite-sample inference oriented, and focus on trade-offs between the missingness-indicator method and single imputation when the former is subject to considerable finite-sample bias.
% too {\red variable/ volatile/ unstable??} in finite samples.  
In particular, the missingness-indicator method involves $(4J + 2)$ coefficients under the fully interacted specification \eqref{eq:mi_l}, and can be demanding on sample sizes when $J$ is large. 
In cases where this results in large finite-sample variations, an alternative is to  include only a subset of $M_i$ with the highest partial $R^2$ in forming the covariate vector.
%, trading off asymptotic efficiency for improved finite-sample performances. 
Alternatively, we could define $J_i = \sum_{j=1}^J (1-M_{ij})$ as the count of observed entries for unit $i$, and form  the additive and fully interacted regressions based on the {\it missingness-count} covariate $x_i^\textup{mc}(c) = (\{x_i^\imp(c)\}^\T, J_i)^\T \in \mathbb R ^{J+1}$. This idea of augmenting the covariates with $J_i$ was first proposed by \cite{rummel} in the context of factor analysis, and can be easily adapted for the current setting of multiple regression model; see \cite{anderson} for a review. 
In addition, recall $C_i$ as the complete-case indicator for unit $i$. We could also augment $x_i^\imp(c)$ with $C_i$
% to account for the information in missingness 
 and form the additive and fully interacted regressions using  $x_i^\cim(c) = (\{x_i^\imp(c)\}^\T, C_i)^\T$.
This gives us three ways to include the missingness information with less covariates than $x_i^\mim(c)$. 
A caveat is that the minimum specification that ensures invariance to $c$ must include all missingness indicators in {\ols}. The three restricted variants thus improve the finite-sample performances at the cost of losing the numerical invariance.

Eventually, the choice of restrictions can be based on the data, and thus necessarily involves a model-selection step. Analyzing this type of procedure is  non-trivial for design-based inference. \citet{bloniarz2016lasso} started the literature by augmenting \citet{Lin13}'s method with a lasso step for variable selection. We can also augment the fully-interacted   {\ols} for the missingness-pattern method with a lasso step. However, deriving its theoretical properties is a challenging research question. We leave the detailed theory to future work.

\subsection{Cluster randomization}\label{sec:cluster}

Consider $N$ units nested in $I$ clusters of sizes $n_i \  (i = \ot{I}; \ \sumii n_i  = N)$.
Cluster randomization randomly assigns $I_1$ clusters to receive the treatment and the rest $I_0 = I - I_1$ clusters to receive the control.
We use $i$ to index the clusters as the randomization units and, with a slight abuse of notation, use $I$ to denote the cluster number when there is no confusion with the identity matrix.

Let $\{Y_{i l}(z): z= 0,1\}$ be the  potential outcomes for the $l$th unit in cluster $i$, also referred to as unit $il \ (i=\ot{I};\ l = \ot{n_i} )$. 
 The finite-population average treatment effect equals
$
\tau  = N^{-1}  \sumii \suml  \{ Y_{i l}(1) - Y_{i l}(0) \}$. 
Let $Z_i $ be the treatment level received by cluster $i$.  
The observed outcome for unit $il $ is $Y_{i l}=Z_i  Y_{i l}(1)+(1- Z_i  )Y_{i l}(0)$. 
We further observe a $J$-dimensional covariate vector $x_{i l} = (x_{il ,1}, \ldots, x_{il ,J})^\T$ for each unit with $\bar x = N^{-1}\sum_{i=1}^I \sum_{l=1}^{n_i} x_{il}$.
Regressions 
\beginy
\label{eq:cr_unit-fisher} 
&&Y_{il} \sim 1 + Z_i + x_{il},\\ 
\label{eq:cr_unit} 
&&Y_{il} \sim 1 + Z_i + (x_{il} - \bar x) +  Z_i (x_{il} - \bar x)
\endy
over $i = \ot{I}; \ l = \ot{n_i}$ 
afford two intuitive specifications to estimate $\tau$ as the coefficients of $Z_i$. 
In fact, they are identical to the specifications under complete randomization with the individual-level treatment indicators, denoted by $Z_{il}$ for unit $il$, satisfying $Z_{il} = Z_i$ under cluster randomization.
\cite{DS} showed the validity of the resulting regression estimators  and their associated cluster-robust standard errors  for large-sample Wald-type inferences.  

In the case where the $x_{il}$'s are only partially observed, 
all {\ntot}  strategies for handling missingness under complete randomization extend to the current setting with no need of modification. 
Assume the missingness is unaffected by the treatment assignment. 
We can derive results in parallel with Propositions \ref{prop:cc_1}--\ref{prop:eff_mp} by assuming the corresponding regularity conditions for cluster randomization.
In particular, let $M_{il}$ be the vector of missingness indicators  for unit $il$, and  $x_{il}^0$ %$x^0_{il} = (1-M_{il})\circ x_{il}$ 
be the imputed covariates with all missing values filled with 0. Replacing $x_{il}$ by $x_{il}^\mim = ((x^0_{il})^\T, M_{il}^\T)^\T$ in   \eqref{eq:cr_unit}   yields \cite{Lin13}'s estimator under the missingness-indicator method.

Whereas the above approach works for both complete and cluster randomizations with identical regression specifications, the peculiarity of cluster randomization allows us to also form regressions based on cluster total data. 
In particular, let $\bar n = N/I$ be the average cluster size, and let  $\tilde{Y}_{i\cdot}(z)=\bar n^{-1} \suml  Y_{i l}(z) $,  $\tilde{Y}_{i\cdot}= \bar n^{-1} \suml Y_{i l}  $, and $\tilde{x}_{i\cdot}=  \bar n^{-1}\suml  x_{i l} $ be the cluster totals of potential outcomes, observed outcomes, and covariates scaled by $\bar n$. 
Then $\tilde Y_{i\cdot} = Z_i \tilde Y_{i\cdot} (1)+ (1-Z_i) \tilde Y_{i\cdot}(0)$ gives the observed analog of $\tilde Y_{i\cdot}(z)$.
This, together with 
$
\tau  
= I^{-1} \sumii 
\{\tilde{Y}_{i\cdot}(1)-\tilde{Y}_{i\cdot}(0)\},
$
ensures that  $( Z_i ,  \tilde{Y}_{i\cdot})_{i = 1}^I$ is equivalent to the observed data from a complete randomization with potential outcomes $\{ \tilde Y_{i\cdot}(0), \tilde Y_{i\cdot}(1)\}_{i=1}^I$ and average treatment effect $\tau$  \citep{MiddletonCl15, DingCLT}. 
The coefficient of $Z_i$ from the cluster-level regression $\tilde{Y}_{i\cdot} \sim 1 + Z_i$ coincides with the difference in means of the $\tilde Y_{i\cdot}$'s and affords an unbiased estimator of $\tau$. \cite{DS} showed that applying \citet{Lin13}'s estimator with scaled cluster totals further improves efficiency, and more importantly, asymptotically dominates the estimators from individual-level regressions \eqref{eq:cr_unit-fisher} and \eqref{eq:cr_unit}. With missing covariates, we first impute all missing covariates with zero, denoted by $x^0_{il}$ for unit $il$, and define the cluster-level covariate vector as  $u_i^\mim = (n_i, (\tilde x^0_{i\cdot})^\T, \tilde M_{i\cdot}^\T)^\T$ where $\tilde M_{i\cdot} = \bar n^{-1}\suml  M_{il}$ and $\tilde x^0_{i\cdot} = \bar n^{-1}\suml  x_{il}^0$. 
Let $\bar u^\mim =  I^{-1} \sum_{i=1}^I  u_{i}^\mim$ be the average of $u_i^\mim$'s over the $I$ clusters. Extending \cite{DS} to the case with missing covariates yields regression
% cluster randomization yields regression [--- Following \cite{DS} yields regression? ---]}
$$
\tilde{Y}_{i\cdot} \sim 1 + Z_i + (u_i^\mim - \bar u^\mim) + Z_i(u_i^\mim - \bar u^\mim)   
$$ 
over $i = \ot{I}$.  The resulting estimator ensures higher asymptotic efficiency than that from \eqref{eq:cr_unit}. It is our final recommendation.

%, and is our recommendation for handling missing covariates under cluster randomization in general.  
%Let $u_i = (n_i, \tilde{x}_{i\cdot}^\T)^\T$ be the vector of scaled cluster total covariates augmented by the cluster size with $\bar u =  I^{-1} \sum_{i=1}^I  u_{i}$. \cite{DS} showed that the coefficient of $Z_i$ from the cluster-level fully interacted regression \beginy\label{eq:cr_cluster}
%\tilde{Y}_{i\cdot} \sim 1 + Z_i + (u_i - \bar u) + Z_i(u_i - \bar u)  %\qquad\text{over}\quad i = \ot{I},
%\endy 
%over $i = \ot{I},$, is asymptotically more efficient than that from the individual-level \eqref{eq:cr_unit}.
%
%
%Extensions to the presence of missing covariates are immediate. 
%To construct the cluster-level regression based on the missingness-indicator method, we can first impute all missing covariates with zero, denoted by $x^0_{il}$ for unit $il$, and then run the cluster-level fully interacted regression \eqref{eq:cr_cluster} based on the cluster-level covariates $u_i^\mim = (n_i, (\tilde x^0_{i\cdot})^\T, \tilde M_{i\cdot}^\T)^\T$, where 
%$\tilde M_{i\cdot} = \bar n^{-1}\suml  M_{il}$ and $\tilde x^0_{i\cdot} = \bar n^{-1}\suml  x_{il}^0$. 
%The resulting regression estimator ensures higher asymptotic efficiency than its counterpart under the individual-level regression, and is our recommendation for handling missing covariates under cluster randomization in general.  

\subsection{Stratified randomization}
Given a study population nested in $K$ strata, indexed by $k = \ot{K}$, stratified randomization conducts an independent complete randomization in each stratum \citep{luke, liu2020regression}. Let $\omega_{[k]} $ be the proportion of units and $\tauj$ be the average treatment effect within stratum $k$ $(k = \ot{K} )$. The finite-population average treatment effect equals  $\tau = \sumk \omega_{[k]} \tauj $. In the case where the covariates are partially observed, we can form $\hat{\tau}^\dg_{\md[k]}$ and  $\hse_{\md[k]}^\dg$ as the basic estimator and robust standard error within each stratum $k$ for $\dgsi$ and $\mds$, and use their respective weighted averages, namely  $\htss = \sumk   \omega_{[k]}\htau_{\md[k]}^\dg$ and $(\hse^\dg_\md)^2 = \sumk  \omega_{[k]}^2 (\hse^\dg_{\md[k]})^2$, 
as our point estimator and the corresponding squared robust standard error. 
%We denote here the estimator and robust standard error under stratified randomization using the same notations as their counterparts under complete randomization. 
%This abuse of notation causes little confusion because $\htss$ and $\hses^\dg$ reduce to their definitions under complete randomization when $K=1$. 
We can derive results in parallel with Propositions \ref{prop:cc_1}--\ref{prop:eff_mp} by assuming the corresponding regularity conditions hold within all strata.

\subsection{Missingness that depends on the treatment assignment}\label{sec:ext_no condition 4}

The discussion so far requires Condition \ref{cond:a} with the missingness unaffected by the treatment assignment, that is, $M_i(0) = M_i(1) = M_i$ for all $i = \ot{N}$.  The resulting missingness indicators are effectively fully observed covariates.
Without Condition \ref{cond:a}, % cause $M_i(1)$ to no longer necessarily equal $M_i(0)$ such that 
$M_{i} = M_{i}(0) + Z_i\{M_{i}(1) - M_{i}(0)\}$ takes different values under different realized values of $Z_i$.
A direct implication is that the vectors of $\xidg$'s that we use to form the additive and fully interacted regressions for $\dgsi$ can no longer be seen as standard covariates unaffected by the treatment even asymptotically, imposing additional complications for quantifying the design-based properties of the resulting estimators.  

As it turns out, among the   estimators in Table \ref{tb:strategies}, only $\hts^\ccov \ (\mds)$ from the complete-covariate analysis remain consistent  without Condition \ref{cond:a}. 
It is thus our recommendation in the absence of Condition \ref{cond:a}. 
Recall from Proposition \ref{prop:eff_mp} that  under Condition \ref{cond:a}, $\htl^\ccov$ has the largest asymptotic variance among the consistent $\htl^\dg$'s. 
Its consistency in the absence of Condition \ref{cond:a} thus gives an analog of the bias-variance trade-off in terms of the asymptotic biases and variances. 
We formalize the results in the Supplementary Material.
%in Section \ref{sec:ext_no4_app} of

\section{Discussion}\label{sec::discussion}

We proposed to use \citet{Lin13}'s model to adjust for missing covariates in randomized experiments by imputing missing covariates with zeros and augmenting the covariates with missingness indicators. When the treatment does not affect the missingness indicators, this missingness-indicator method is consistent for the average treatment effect and more efficient than the unadjusted estimator, the complete-covariate analysis, and the estimators based on imputed covariates alone. It can be conveniently implemented via {\ols}. We also proposed the missingness-pattern method as a modification to reap additional asymptotic efficiency.

We focused on constructing large-sample Wald-type confidence intervals based on consistent point estimators and conservative standard errors. Building upon these results, it is immediate to extend \citet{ZDfrt} to construct robust Fisher randomization tests adjusting for covariates subject to missingness. In particular, using a studentized statistic based on any consistent estimator and the associated conservative standard error in the Fisher randomization test yields a $p$-value that is finite-sample exact under the strong null hypothesis $\tau_i=0$ for all $i$ and asymptotically conservative under the weak null hypothesis $\tau = 0$. By duality, this also gives a confidence interval by inverting a sequence of Fisher randomization tests.

As an alternative to regression, weighting based on the propensity score is another simple yet powerful method to improve efficiency in randomized experiments. \citet{shen2014inverse} and \citet{zeng2021propensity}  have shown that regression and weighting are equivalent asymptotically. With missing covariates, one option is to use the generalized propensity score \citep{rr1984} to construct weighting estimators. We conjecture that the equivalence between regression and weighting also holds even with missing covariates, but leave the theoretical analysis to future work.

%\section*{Acknowledgments}
%Peng Ding was partially funded by the U.S. National Science Foundation (grant \# 1945136). 

\bibliographystyle{plainnat}
\bibliography{refs_missingCov}

\newpage
\setcounter{page}{1}

\onehalfspacing

\setcounter{equation}{0}
\setcounter{section}{0}
\setcounter{figure}{0}
\setcounter{example}{0}
\setcounter{proposition}{0}
\setcounter{corollary}{0}
\setcounter{theorem}{0}
\setcounter{table}{0}
\setcounter{condition}{0}
\setcounter{lemma}{0}
\setcounter{remark}{0}

\renewcommand {\theproposition} {S\arabic{proposition}}
\renewcommand {\theexample} {S\arabic{example}}
\renewcommand {\thefigure} {S\arabic{figure}}
\renewcommand {\thetable} {S\arabic{table}}
\renewcommand {\theequation} {S\arabic{equation}}
\renewcommand {\thelemma} {S\arabic{lemma}}
\renewcommand {\thesection} {S\arabic{section}}
\renewcommand {\thetheorem} {S\arabic{theorem}}
\renewcommand {\thecorollary} {S\arabic{corollary}}
\renewcommand {\thecondition} {S\arabic{condition}}
\renewcommand {\thepage} {S\arabic{page}}

%\doublespacing
%\onehalfspacing
\begin{center}
\bf \Large 
Supplementary Material  
\end{center}

Section \ref{sec:ext_no4_app} gives the more general results under possible violations of Condition \ref{cond:a}. 

%Section \ref{sec:cim} introduces the complete-case-indicator method as a modification to the missingness-indicator method and outlines its design-based properties under complete randomization.  for verifying the results in the main text and Section \ref{sec:ext_no4_app}. 

Section \ref{sec:no_missing_app} reviews the key lemmas. 

Sections  \ref{sec:cc_app}--\ref{sec:mp_aggregate_app} give the proofs of the results in the main text and Section \ref{sec:ext_no4_app} under individual strategies. The proofs for the complete-covariate analysis are short and given in the text. 
We verify the general results without Condition \ref{cond:a} unless specified otherwise.

\section*{Notation}
 
Consider an experiment with two treatment levels, $\mt = \{0,1\}$, and a finite population of size $N$. Given $\{u_i(z), v_i(z): z\in\mt; \ i = \ot{N}\}$, where $u_i(z)$ and $v_i(z)$ are two arbitrary potential outcome vectors for unit $i$, let $\bar u(z) = \meani u_i(z)$, $\bar v(z) = \meani v_i(z)$, and $S_{u(z), v(z')} = \meani \{u_i(z)-\bar u(z)\} \{v_i(z')-\bar v(z')\}^\T$ be the finite-population means and covariances, respectively, for $z, z'\in\mt$. 
We use $N$ instead of $(N-1)$ as the divisors for the covariances; the simplification does not affect the validity of the proofs with $(N-1)/ N = 1 + o(1)$  as $N\to\infty$. 
Abbreviate $u_i(z)$, $\bar u(z)$, and $S_{u(z), v(z')}$ to $u_i$, $\bar u$, and $S_{u, v(z')}$, respectively, if $u_i(z)$ is unaffected by the treatment with $u_i(0) = u_i(1)$ for all $i$. 
Abbreviate $S_{u(z), v(z)}$ as $S_{uv}(z)$ occasionally to simplify the notation. 

Further let $Z_i \in \mt$ be the treatment indicator of unit $i$. 
The observed values of $u_i(z)$ and $v_i(z)$ are $u_i = Z_i u_i(1) + (1-Z_i) u_i(0)$ and $v_i = Z_i v_i(1) + (1-Z_i) v_i(0)$ for unit $i$.
Let $N_z = \sumi 1(Z_i = z)$ be the number of units under treatment $z$.
Let $\hat u(z) = \nzinv \sumiz u_i$, $\hat v(z) = \nzinv \sumiz v_i$, and $\hat S_{u  v}(z) = \nzinv  \sumiz \{u_i - \hat u(z)\} \{v_i - \hat v(z)\}^\T$ be the sample analogs of $\bar u(z)$, $\bar v(z)$,  and $S_{u(z), v(z)}$, respectively, for units under treatment $z$. 
We use $N_z$ instead of $(N_z-1)$ as the divisors for the sample covariances; the simplification does not affect the validity of the proofs with $(N_z-1)/ N_z = 1 + o(1)$ as $N\to\infty$ and $e_z = N_z/ N$ converges to a limit in $(0,1)$. 

Write $a_N \asim b_N$ if $\sqrt N(a_N-b_N) = \op$ for  sequences of random variables $(a_N)_{N=1}^\infty$ and $(b_N)_{N=1}^\infty$. Slutsky's theorem ensures that $\sqrt N a_N$ and $\sqrt N b_N$ have the same limiting distribution as long as $a_N\asim b_N$ and either $( \sqrt N a_N)_{N=1}^\infty$ or $(\sqrt N b_N)_{N=1}^\infty$ has a limiting distribution. 
%{\red [--- no need to center because of the `and either ... or ... has a limiting distribution?---]} 
We suppress the subscript $N$ when no confusion would arise.

Let $\circ$ denote the Hadamard product and $\otimes$ denote the Kronecker product of matrices, respectively. We will repeatedly use the following property of the Kronecker product:
$$(H_1   H_2 ) \otimes (H_3 H_4)  = (H_1 \otimes H_3 ) (H_2\otimes H_4)$$ for matrices $H_k \ (k = 1, 2, 3, 4)$ with compatible dimensions.

\newpage
\section{Extensions to missingness that depends on the assignment}\label{sec:ext_no4_app}
%\subsection{Notation and regularity conditions}\label{sec:notation_no_condition4}

We first introduce the notation without Condition \ref{cond:a}. 
Recall $M_i = (M_{i1}, \ldots, M_{iJ})^\T$ as the observed missingness indicators of unit $i$ with $M_{ij} = 1(\text{$x_{ij}$ is missing})$. 
Recall $M_i(z) = (M_{i1}(z), \ldots, M_{iJ}(z))^\T$ as the potential value of $M_i$ if unit $i$ were assigned to treatment $z$.
% under possible violations of Condition \ref{cond:a}.
%
For $C_i = 1(M_i = 0_J)$ and $A_i = 1_J - M_i$, let  
  $C_i(z) = 1\{ M_i(z) = 0_J\}$ and $A_i(z) = (A_{i1}(z), \ldots, A_{iJ}(z) )^\T= 1_J - M_i(z)$ be the corresponding potential values, respectively.
  %, if unit $i$ were assigned to treatment $z$. 
Let $x_i^0(z) = A_i(z) \circ x_i  = (A_{i1}(z) x_{i1}, \ldots, A_{iJ}(z) x_{iJ})^\T$ be the corresponding imputed variant of $x_i$ with all missing elements replaced by 0.
%Let $\bc(z) = \meani C_i(z)$, $\bax(z) = \meani A_i(z) \circ x_i$,  and $\barm(z) = \meani M_i(z)$ be the averages of $C_i(z)$, $A_i(z) \circ x_i$, and $M_i(z)$ over $i = \ot{N}$. 
%Let $\hc(z) =\meaniz C_i$, $\hax(z)  = \meaniz A_i \circ x_i$, and $\hm(z) = \meaniz M_i$ be the sample analogs of $\bc(z) $, $\bax(z) $,  and $\barm(z) $, respectively, under treatment $z$. 
Condition \ref{asym} generalizes Condition \ref{asym_1} to possible violations of Condition \ref{cond:a}. 
\begin{condition}\label{asym} As $N \to \infty$, 
(i) $e_z  $ has a limit in $  (0,1)$  for $z =  0,1$;  
(ii) the first two moments of $\{Y_i(z), x_i,  M_i(z), C_i(z), C_i(z) Y_i(z),  C_i(z) x_i ,   A_i(z) \circ x_i: z = 0, 1\}_{i=1}^N $  have finite limits,
and 
 (iii) there exists a  $u_0 < \infty$ independent of $N$ such that $\meani \|x_i\|_4^4 \leq u_0$ and $\meani Y_i^4(z) \leq u_0$ for $z = 0,1$. 
\end{condition}

Let $\barm(z)$, $\bc(z)$, and $\bax(z)$ denote the averages of $M_i(z)$, $C_i(z)$, and $A_i(z) \circ x_i$  over $i = \ot{N}$, respectively. 
Let $\hm(z)$, $\hc(z) $, and $\hax(z)  $  be the corresponding sample analogs over units under treatment $z$.

\subsection{Complete-case analysis}
Refer to $\{i: C_i(z) = 1\}$ as the {\it $z$-complete cases} with all covariates observed if assigned to treatment $z$.
Refer to $\{i: C_i = 1\}$ as the {\it observed complete cases} with all covariates observed under the realized assignment. 
The two sets of units coincide under Condition \ref{cond:a}. 
Analogous to the definitions of $\nc$, $
\byc(z)$, $\bxc$, $\sxxc$, and $\tau^\cc$ over the observed complete cases from the main text, let $\nc(z) = \sumi C_i(z)$ be the number of the $z$-complete cases, with $\bar C(z) = \meani C_i(z) =   N^\cc(z) / N$ as the corresponding proportion; let 
$
\byc(z)$, $\bxc(z)$, $\sxxc(z)$, and $\sxyc(z)$
%\beginy\label{eq:yx_cc}
%\byc(z) = \nc(z)^{-1}  \sum_{i: C_i(z) =1}  Y_i(z) , \qquad \bxc(z) =  \nc(z)^{-1}  \sum_{i: C_i(z) =1}   x_i, 
%\endy
%and 
%\beginy\label{eq:s_cc}
%&&\sxxc(z) =\nc(z)^{-1}  \sum_{i: C_i(z) = 1} \{x_i - \bxc(z) \}\{x_i - \bxc(z) \}^\T , \\\nonumber 
%&&\sxyc(z) =\nc(z)^{-1} \sum_{i: C_i(z)=1} \{x_i - \bxc(z) \}\{Y_i(z) - \byc(z) \} 
%\endy
be the corresponding finite-population means and covariances of $Y_i(z)$'s and $x_i$'s over $ \{i: C_i(z) = 1\}$; and let $\tau^\cc   = \by^\cc (1) - \by^\cc (0)$. 
With $C_i(0)$ and $C_i(1)$ no longer necessarily equal under possible violations of Condition \ref{cond:a},
$\byc(0)$ and $\byc(1)$ are now the average potential outcomes over two distinct subsets of units, namely $\{i: C_i(0) = 1\}$ and $\{i: C_i(1) = 1\}$, such that the resulting difference $\tau^\cc$ is no longer necessarily a causal effect.
%Likewise Let $\byu(z)$ be the average of $Y_i(z)$'s over the {\it $z$-incomplete cases}, $\{i: C_i(z) = 0\}$. 
Let \begina
\rho_z^\cc 
=  \frac{e_z \bar C(z)}{e_1 \bar C(1) + e_0  \bar C(0)}  \qquad\text{for}\quad z = 0,1
\enda
  with $\rho_0^\cc + \rho_1^\cc =1$. As $N\to\infty$, $\rho_z^\cc$ gives the probability limit of the proportion of  observed complete cases that receive treatment $z$; we give the details in Lemma \ref{lem:lim_cc}. 
Let 
$\byu(z)$ be the average of $Y_i(z)$'s over the {\it $z$-incomplete cases}, $\{i: C_i(z) = 0\}$, analogous to $\byc(z)$.

We can show that $\by^\cc (z)$, $\bxc(z)$, $\sxxc(z)$, $\sxyc(z)$, $\tau^\cc$, $\rho_z^\cc$, and $\by^\uc(z)$ all have finite limits under Condition \ref{asym}. We also use the same symbols to denote their respective limits when no confusion would arise. 

\begin{proposition}\label{prop:cc_asym}
Assume complete randomization, Condition \ref{asym}, and  the limits of $\{S_{xx}^\cc(z): z= 0,1\}$ are both positive definite. We have  $$\htsc - \tau = (\tau^\cc -\tau) -  \{\bxc(1)  - \bxc(0)  \}^\T \gsc + \op \formds,$$
where 
$
\gf ^\cc  =
\{ \rho_0^\cc  \sxxc(0) +  \rho_1^\cc \sxxc(1) \}^{-1}
\{ \rho_0^\cc  \sxyc(0) +  \rho_1^\cc \sxyc(1) \}$,
and  $\glc =  \rho_0^\cc  \glo^\cc +  \rho_1^\cc \glz^\cc$ with $
\gamma_{\lin,z}^\cc = \{\sxxc(z)\}^{-1} \sxyc(z)$ for $z = 0,1$.
A sufficient and necessary conditions for $\tau^\cc = \tau$ is 
 $\{1 - \bc(1) \}\{ \byc(1)	- \byu(1)\}  = \{1-\bc(0)\}\{  \byc(0) - \byu(0) \}$.
\end{proposition} 
With $\bxc(1)$ and $\bxc(0)$ no longer necessarily equal in the absence of Condition \ref{cond:a}, 
Proposition \ref{prop:cc_asym} generalizes Proposition \ref{prop:cc_1} to  possible violations of Condition \ref{cond:a}, and illustrates the additional source of asymptotic bias, in addition to the difference between $\tau^\cc $ and $\tau$, due to the possible difference between $\bxc(1)$ and $\bxc(0)$.  

\subsection{Complete-covariate analysis}

\begin{condition}\label{cond:J}
The  set of observed complete covariates $\mathcal J = \{j: M_{ij}(Z_i) = 0 \text{ for all } i = \ot{N}\}$ remains unchanged over all possible values of $(Z_i)_{i=1}^N$ under complete randomization  for $N = \ot{\infty} $, with $S_{xx}^\ccov  $ and its  limit under Condition \ref{asym} both being positive definite as $N \to \infty$. 
\end{condition}

\begin{proposition}\label{prop:ccov_no 4}
Assume complete randomization and Conditions \ref{asym}--\ref{cond:J}. 
Proposition \ref{prop:ccov_1} holds regardless of whether Condition \ref{cond:a} holds or not. 
\end{proposition}

Echoing the comments under Proposition \ref{prop:ccov_1}, Condition \ref{cond:J} ensures that the set $\mathcal J$ remains constant over all possible treatment assignments, rendering $x_i^\ccov$ an effective true covariate vector unaffected by the randomization. 
The result  of Proposition \ref{prop:ccov_no 4} follows from applying Lemma \ref{lem::basic-x} to the finite population of $\{Y_i(0), Y_i(1), x_i^\ccov\}_{i=1}^N$ regardless of whether Condition \ref{cond:a} holds or not, and justifies the large-sample Wald-type inference based on $(\hts^\ccov, \hses^\ccov)$. 

\subsection{Single imputation and missingness-indicator method}

Recall  $\xs_i(c)= (\xs_{i1}(c_1), \ldots, \xs_{iJ}(c_J))^\T$ as the imputed covariate vector with $x_{ij}^\im (c_j)  =  (1-M_{ij})  \xij  +M_{ij}  c_j$. 
Single imputation uses $\xs_i(c)$ to form the additive and fully interacted regressions as \eqref{eq:im_f} and \eqref{eq:im_l}. 
The missingness-indicator method uses $x_i^\mim(c) = (\xs_i(c)^\T, M_i^\T)^\T$ to form the additive and fully interacted regressions as \eqref{eq:mi_f} and \eqref{eq:mi_l}. 

Let $\xs_{ij}(z; c_j) = \{1-M_{ij}(z)\} x_{ij} + M_{ij}(z) c_j $ be the potential value of $\xs_{ij}(c_j)$  if unit $i$ were assigned to treatment $z$;  we treat $c_j$ as fixed in defining $\xs_{ij}(z; c_j)$ despite its possible dependence on $(Z_i)_{i=1}^N$.
Let 
$
x^\im  _i(z; c) = \{\xij^\im  (z; c_j)\}_{j=1}^J$ and 
$
x_i^\mim(z; c) =  (x_i^\imp(z; c)^\T, M_i(z)^\T)^\T$
be the corresponding potential values of $\xs_i(c)$ and $x_i^\mim(c)$, respectively. 
%Direct comparison shows $\sxxim(z;\ci)$ and $\sxyim(z;\ci)$ are the upper-left $J\times J$ submatrix and upper $J\times 1$ vector of $\sxxmi (z;\ci)$ and $\sxymi (z;\ci)$, respectively. 
We focus on imputations $c = (c_j)_{j=1}^J$ with finite probability limits $\ci = (\cji)_{j=1}^J$: 
$$
\mathcal C' = \{c \in \mathbb{R}^J:  \plim c_j = \cji <\infty\ (j=1,\ldots, J) \text{ under complete randomization and Condition \ref{asym}} \}.
$$
In particular, $\mathcal C'$ generalizes the definition of $\mathcal C$ in the main text to possible violations of Condition \ref{cond:a}. 
The constant imputation with $c_j = 0$  is a special case of $c \in \mathcal C'$ with $\ci = 0_J$; 
the unconditional mean imputation with $c_j = \hx_j^\obs$  is  also a special case with $\cji  = \{e_1\baj(1) + e_0\baj(0)\}^{-1} \{e_1 \bajxj(1) + e_0 \bajxj(0)\}$, where $\baj(z) = \meani A_{ij}(z) $ and $\bajxj(z) = \meani A_{ij}(z) x_{ij}$ denote the $j$th elements of $\ba (z)$ and $\bax(z)$, respectively. 

Let $\sxxdg (z;c)$ and $\sxydg (z;c)$ be the finite-population covariances of $\{\xidg(z; c)\}_{i=1}^N$  and  $\{\xidg(z; c), Y_i(z)\}_{i=1}^N$ for $\dgsss$. 
For $c \in \mathcal C'$ with  $\plim c = \ci$, 
we can show that $ \sxxdg(z;\ci)  $ and $ \sxydg(z;\ci) $ have finite limits under Condition \ref{asym} for all $z = 0, 1$ and $\dgsss$. 
We will use $\sxxdg(z;\ci)$ and $\sxydg(z;\ci)$ to also denote their respective limiting values when no confusion would arise. Let $\dci   =\diag(\cji)_{j=1}^J$ be the $J \times J$ diagonal matrix. % with $\cji$'s on the diagonal for $c \in \mathcal {C}'$.  

Recall $\hts^\mim = \hts^\mim(0_J)$ as the value of $\hts^\mim(c)$ when we impute all missing covariates with 0. 
Lemma \ref{lem:invar_mim} holds without Condition \ref{cond:a} such that we still have $\hts^\mim(c) = \hts^\mim$ for all $c \in \mathbb R^J$.

\begin{proposition}\label{prop:imp_mim_cim no 4}
Assume complete randomization, Condition \ref{asym}, and $c \in \mathcal C'$ with the limits of $\sxxdg(z ; \ci) \ (\dg = \im, \mi; \ z = 0,1)$  all being positive definite. We have 
\begina
&&\htsim(c) - \tau = - [\bax(1) - \bax(0) + \dci \{\barm(1)-\barm(0)\} ]^\T \gsim(\cinf) + \op, \\
&&\htsmi(c) - \tau = - ( [\bax(1) - \bax(0)  ]^\T, \  \{\barm(1)-\barm(0)\}  ^\T) \gsmi(0_J) + \op
\enda
for $\mds$, where $\gamma_{\fisher}^\dg (\cinf) =  \{ e_0 \sxxdg(0;\ci ) + e_1\sxxdg(1;\ci ) \}^{-1}
\{ e_0 \sxydg(0;\ci ) + e_1\sxydg(1;\ci ) \}$ and $
\gamma_{\lin}^\dg (\cinf) = e_0 \gamma_{\lin,1}^\dg (\cinf) + e_1 \gamma_{\lin,0}^\dg (\cinf)$  
 with  $\gamma_{\lin,z}^\dg (\cinf) = \{\sxxdg(z;\ci)\}^{-1}\sxydg(z;\ci)$  for  $\dgsss$ and $ z= 0,1$. 
\end{proposition}

With $\bax(1) - \bax(0)$ and $\barm(1) - \barm(0)$ no longer necessarily equal to $0_J$, 
Proposition \ref{prop:imp_mim_cim no 4} highlights the possible asymptotic biases of single imputation and the missingness-indicator method  in the absence of Condition \ref{cond:a}. 
We thus recommend using the complete-covariate analysis whenever the validity of Condition \ref{cond:a} is in doubt.
%,  which ensures consistent estimation whatsoever.  

Whereas $\htsmi(c)$ is invariant to $c$ and thus precludes the possibility of bias reduction via crafted choice of $c$, the dependence of $\htsim(c)$ on $c$ promises a way to reduce the asymptotic bias by a data-dependent choice of $c$. 
In particular, 
recall $\bar A_j(z)  $ and $\bax_j(z) $ as the $j$th elements of $\ba (z)$ and $\bax(z)$, respectively. 
A sufficient condition for $\htsim(c)$ to be consistent is $\bax(1) - \bax(0) + \dci \{\barm(1)-\barm(0)\} = 0_J$.
Let  $\hat A_j(z) = \meaniz A_{ij}$  and $\hax_j(z) = \meaniz A_{ij} x_{ij}$ be the sample analogs of  $\bar A_j(z)$  and $\bax_j(z)$, respectively. 
When $\ba_j(0) \neq \ba_j(1)$ for all $j$'s, we can choose
$$
c_j =  \frac{\hajxj(1) - \hajxj(0)}{\ha_j(1) - \ha_j(0)} \quad ( j = \ot{J}) 
$$
to ensure 
$$
\cji  =  \frac{\bajxj(1) - \bajxj(0)}{\ba_j(1) - \ba_j(0)} \quad ( j = \ot{J})
$$
to remove the asymptotic bias. 
When $\ba_j(0) = \ba_j(1)$ for some $j$'s, it is impossible to remove the bias by choosing the $c_j$'s.  

\subsection{Missingness-pattern method}\label{sec:mp_ext}
Recall that $\htl^\mp$ equals the coefficient of $Z_i$ from the aggregate regression \eqref{eq:mp_agg}  for arbitrary $c \in \mathbb{R}^J$. 
Let $u_i^\mp = u_i^\mp(0_J)$ be the value of $u_i^\mp(c)$ at $c = 0_J$. 
Let $u_i^\mp(z)$ be  the potential value of $u_i^\mp$ if unit $i$ were assigned to treatment $z$. It is the vector of $x_i^\imp(z; 0_J)$, $M_{i1}(z), \ldots, M_{iJ}(z)$, and all their interactions. 

Let $\bar u^\mp(z) = \meani u_i^\mp(z)$ be the finite-population mean of $u_i^\mp(z)$ over $i = \ot{N}$. 
Let $S_{uu}^\mp (z)$ and $S_{uY}^\mp(z)$ be the finite-population covariances of $\{u_i^\mp(z)\}_{i=1}^N$  and  $\{u_i^\mp(z), Y_i(z)\}_{i=1}^N$, respectively. 
%Direct comparison shows $\sxxim(z;\ci)$ and $\sxyim(z;\ci)$ are the upper-left $J\times J$ submatrix and upper $J\times 1$ vector of $\sxxmi (z;\ci)$ and $\sxymi (z;\ci)$, respectively. 

\begin{condition}\label{asym_mp}
As $N\to \infty$, for $z = 0,1$, (i) $e_z$ has a limit between $(0, 1)$, and (ii) $\bar u^\mp(z )$, $S_{uu}^\mp (z )$, and $S_{uY}^\mp(z )$ all have finite limits, with $S_{uu}^\mp(z)$ and its limit both being positive definite. 
\end{condition}

\begin{proposition}\label{prop:mp_no 4}
Assume complete randomization and Condition \ref{asym_mp}. 
We have 
\begina
\htl^\mp = \tau - \{\bar u^\mp(1) - \bar u^\mp(0)\}^\T \gamma_\lin^\mp + \op,
\enda
where $\gamma_\lin^\mp = e_1^{-1}\gamma_{\lin,0}^\mp + e_0^{-1}\gamma_{\lin,1}^\mp$ with $\glzz^\mp = \{S_{uu}^\mp (z )\}^{-1}S_{uY}^\mp (z )$ for $z = 0,1$. 
\end{proposition} 

With $\bar u^\mp(0)$ and $\bar u^\mp(1)$ no longer necessarily equal, 
Proposition \ref{prop:imp_mim_cim no 4} highlights the possible asymptotic bias of $\htl^\mp$  in the absence of Condition \ref{cond:a}. 
The same result extends to $\htf^\mp$ from the additive regressions as well. We relegate the details to Section \ref{sec:mp_add_app}. 

\section{Lemmas}\label{sec:no_missing_app}
%\subsection{Overview and notation}
%We verify in Sections \ref{sec:cc_app}--\ref{sec:mp_aggregate_app} below the theoretical results in the main text and Section \ref{sec:ext_no4_app}. 
%We proceed without Condition \ref{cond:a} unless specified otherwise. 

%Recall from Section \ref{sec:notation_no_condition4} the definitions of $M_i(z)$, $C_i(z) $, and $A_i(z)$ as the potential values of $M_i$, $C_i$, and $A_i$, respectively, under possible violations of Condition \ref{cond:a}. 
%
Let $(Ax)_i(z) = A_i(z) \circ  x_i$, $(CY)_i(z) = C_i(z)Y_i(z)$, and $(Cx)_i(z) = C_i(z)x_i$ be shorthand notations for $ A_i(z) \circ  x_i$, $C_i(z)Y_i(z)$, and $ C_i(z)x_i$, respectively, with finite-population means $\bax(z), \bcy(z), $ and $\bcx(z)$ over $i = \ot{N}$.
%\begina
%\bax(z) = \meani (Ax)_i(z), \qquad \bcy(z) = \meani (CY)_i(z), \qquad \bcx(z) = \meani (Cx)_i(z).
%\enda 
Let $(Ax)_i = A_i \circ x_i$, $(CY)_i = C_i Y_i$, and $(Cx)_i = C_i x_i$ be the observed values of $(Ax)_i(z)$, $(CY)_i(z)$, and $(Cx)_i(z)$, respectively, with  $\hax(z) , \hcy(z), $ and $\hcx(z) $
%\begina
%\hax(z) = \meaniz (Ax)_i, \qquad \hcy(z) = \meaniz (CY)_i, \qquad \hcx(z) = \meaniz (Cx)_i
%\enda 
as the sample analogs of $\bax(z)$, $\bcy(z)$, and $\bcx(z)$ over $\{i: Z_i = z\}$. 

%\subsection{Lemmas}
Lemma \ref{lem:gf_gl} gives the numerical expressions of  $\hts \ (\mds)$. % irrespective of the data-generating process. 
The result affords the basis for quantifying the sampling properties of $\htss$'s for $\dgsi$ and $\mds$. % in the presence of missing covariates. 

\begin{lemma}\label{lem:gf_gl}
%For $(Y_i, x_i, Z_i)_{i=1}^N$ from arbitrary data-generating process, 
%The coefficients of $Z_i$ from $Y_i \sim 1+Z_i +x_i$ and $Y_i \sim 1+Z_i +(x_i - \bx) + Z_i(x_i - \bx)$ over $i = \ot{N}$ satisfy
We have 
\begina
&&\htf = \htn - \htx^\T \hg_\fisher,\\ 
&&\htl = \htn - \htx^\T \hg_\lin = \left[ \hy(1) - \{\hx(1) - \bx\}^\T \hglo\right]- \left[\hy(0) - \{\hx(0) - \bx\}^\T \hglz\right],
\enda
where $
\hgf = \{e_0\hsxx(0) + e_1\hsxx(1) \}^{-1} \{e_0 \hsxy(0) + e_1\hsxy(1) \}$ is the coefficient of $x_i$ from $Y_i \sim 1+Z_i +x_i$ over $i = \ot{N}$, and $\hg_\lin = e_0\hglo + e_1\hglz$ with $\hglzz = \{\hsxx(z)\}^{-1}\hsxy(z)$ equaling the coefficient of $x_i$ from the treatment-specific {\ols} fit $Y_i \sim 1+x_i$ over $\{i: Z_i =z\}$. 
\end{lemma}

\begin{proof}[Proof of Lemma \ref{lem:gf_gl}] 
The numerical form of $\htl$ follows from \citet[][Proposition 1]{ZDfrt}.
The numerical form of $\htf$ follows from the basic properties of {\ols} and, in particular, the Frisch--Waugh--Lovell theorem. 
\end{proof}

Lemma   \ref{lem:Ding17}  states the asymptotic  normality of moment estimators under general complete randomization with $Q\geq 2$ treatment levels. 

\begin{lemma}\citep[][Theorems 3 and 5]{DingCLT}\label{lem:Ding17}
In a completely randomized experiment with $N$ units and $Q$ treatment groups of sizes $(N_q)_{q=1}^Q$, let $Y_i(q)$ be the $L$-dimensional potential outcome vector of unit $i$ under treatment $q$, and $S_{qq'} = (N-1)^{-1} \sum_{i=1}^N \{Y_i(q) - \bar Y(q)\}\{Y_i(q') - \bar Y(q')\}^\T$ be the finite-population covariance. 
Let $\bt( \Gamma) =  \sumq   \Gamma_q \bar Y(q)$, where $\Gamma_q$ is an  arbitrary $K\times L$ coefficient matrix for $q = \ot{Q}$. 
The estimator $\hbt( \Gamma) =  \sumq   \Gamma_q \hY(q)$ has mean $\bt( \Gamma)$ and covariance 
$$\cov\{\hbt( \Gamma)\} = \sumq N_q^{-1}  \Gamma_q  S_{qq}  \Gamma_q^\T - N^{-1}  S_{\bt(\Gamma)}^2,$$
where $S_{\bt(\Gamma)}^2$ is the finite-population covariance of $\{\bt_i( \Gamma) =  \sumq   \Gamma_q Y_i(q): \ i=\ot{N}\}$.
If for any $1\leq q, q' \leq Q$, $S_{qq'}$ has a finite limit, $N_q/N$ has a limit in $(0,1)$, and $\max_{1\leq q\leq Q, 1\leq i \leq N} \|Y_i(q)-\bar Y(q)\|^2_2 /N \to 0$, then 
$N\cov\{\hbt (\Gamma)\}$ has a limiting value, denoted by $ V$, and 
$$\sqrtn \{\hbt ( \Gamma) - \bt( \Gamma)\} \rightsquigarrow \mathcal{N}( 0,  V).$$
\end{lemma}

Lemma \ref{lem:joint_dist} follows from Lemma \ref{lem:Ding17} and states the asymptotic joint normality of 
$\{ \hc(z),  \hcy(z), \hcx(z): z = 0,1\}$ and that of $\{\hy(z), \hm(z), \hc(z), \hax(z): z = 0,1\}$ without Condition \ref{cond:a}. 
It affords the basis for verifying the asymptotic normality of $\htss$'s.
 %for $\dg= \cc, \imp, \mim, \cim$ 
%under Condition \ref{cond:a}. 

\begin{lemma}\label{lem:joint_dist}
Assume complete randomization and Condition \ref{asym}. We have 
\begina
\sqrtn \beginp
\hcy(0) - \bcy(0) \\
\hcy(1) - \bcy(1)\\
\hcx(0) -  \bcx(0)\\
\hcx(1) -   \bcx(1)\\
\hc(0) - \bc(0)\\
\hc(1) - \bc(1)
\endp \quad\text{and}\quad
\sqrtn \beginp
\hy(0) - \by(0)\\
\hy(1) - \by(1)\\
\hax(0) -  \bax(0)\\
\hax(1)-  \bax(1)\\
\hm(0) -  \barm(0)\\
\hm(1) -  \barm(1)
\endp
\enda
are asymptotically normal with means $0_{4+2J}$ and $0_{4 + 4J}$, respectively. 
\end{lemma}

\begin{proof}[Proof of Lemma \ref{lem:joint_dist}]
We verify below the result for $\{\hcy(z), \hcx(z), \hc(z):z=0,1\}$.
The proof for $\{\hy(z), \hax(z), \hm(z): z= 0,1\}$ is almost identical and thus omitted. 

%\textbf{Result on $\{\hcy(z), \hcx(z), \hc(z):z=0,1\}$.}\\
 Let $G_i (z) = (C_i(z)Y_i(z), C_i(z) x_i^\T,  C_i(z))^\T$ be the augmented $(2+J)$-dimensional potential outcome vector of unit $i$. 
Condition \ref{asym} ensures that the finite-population covariance matrix of $\{G_i(0), G_i (1)\}_{i=1}^N$ has a finite limit, and it follows from 
\begina
\|G_i(z) - \bar G(z)\|_2^2 
&=& \{C_i(z) Y_i(z) -  \bcy(z)\}^2 + \| C_i (z)  x_i - \bcx(z)\|_2^2 + \{C_i(z) - \bc(z)\}^2 \\
&\leq &  2\{C_i(z) Y_i(z)\}^2 + 2\{ \bcy(z)\}^2 + 2\|C_i(z) x_i\|_2^2 +2 \|\bcx(z)\|_2^2 +  2\{C_i(z) \}^2 +  2\{ \bc(z)\}^2  \\
&\leq &  2\{Y_i(z)\}^2 + 2\{ \bcy(z)\}^2 + 2\|x_i(z)\|_2^2 +2 \|\bcx(z)\|_2^2 + 2 +  2\{ \bc(z)\}^2 
\enda
that $\max_{1\leq i \leq N; z = 0, 1}
\|G_i(z) - \bar G(z)\|_2^2/N = o(1)$. 
The finite population of $\{G_i(0), G_i (1)\}_{i=1}^N$ thus satisfies the regularity conditions in Lemma \ref{lem:Ding17}, and ensures the joint asymptotic normality of $\{\hcy(z), \hcx(z), \hc(z):z=0,1\}$.
\end{proof}

The following lemma states the invariance of {\ols} to  non-degenerate linear transformations of the design matrix. 
Its proof follows from simple linear algebra but the result helps to simplify many proofs below. 
In particular, it guarantees that \citet{Fisher35} and \citet{Lin13}'s estimators are invariant to non-degenerate linear transformations of the covariates. 

\begin{lemma}\label{lem:trivial} 
Consider an $N \times 1$ vector $ Y$  and two $N\times Q$ matrices, $X_1$ and $X_2$, that satisfy $X_2 = X_1\Gamma$ for some invertible $Q\times Q$ matrix $\Gamma$. 
The {\ols} fits
\begina
 Y =  X_1 \hb_1 + \hat\epsilon_1, \quad  Y =  X_2 \hb_2 +\hat\epsilon_2
\enda
yield robust covariances $\hat \Psi_1  $ and $ \hat\Psi_2$. They
satisfy 
$$
\hb_1 = \Gamma\hb_2,\quad
\hat\epsilon_1 = \hat\epsilon_2,\quad
\hat \Psi_1 = \Gamma \hat\Psi_2 \Gamma^\T.
$$ 
\end{lemma}

\section{Complete-case analysis}\label{sec:cc_app}

\subsection{Useful  facts without Condition \ref{cond:a}}
\def\nzc{N_0^\cc}
\def\sumizcc{\sum_{i: C_i = 1, Z_i = z}}
%All results in this subsection hold regardless of whether Condition \ref{cond:a} holds or not. y
%Recall $\{i: C_i(z) = 1\}$ as the $z$-complete cases with all covariates observed if assigned to treatment $z$.
%Recall $\{i: C_i = 1\}$ as the {\it observed complete cases} with all covariates observed under the realized assignment $Z_i$. 
Recall $\bar C(z) = \meani C_i(z)$ as the proportion of $z$-complete cases among all $N$ units.  
 Let  
%\beginy\label{eq:bcx_cc}
%&& \bcx (z) = \meani  C_i(z) x_i,\qquad
%\bcy(z) = \meani C_i(z) Y_i(z), \nonumber \\
%&& \bcxxt (z) = \meani   C_i(z)  x_i x_i^\T,
%\qquad
% \bcxy (z) = \meani   C_i(z)  x_i Y_i(z)
%% , 
%% \quad 
%% \bcyy(z) = \meani C_i(z)Y_i(z) Y_i(z)
%\endy
$\bcx (z)$, $\bcy(z)$, $\bcxxt (z)$, and $\bcxy (z)$ 
be the averages of $ C_i(z)x_i$, $  C_i(z)Y_i(z)$, $C_i(z)x_i x_i^\T$, and $C_i(z)x_i Y_i(z)$  over $i = \ot{N}$, respectively, for $z = 0,1$. 
%Recall from \eqref{eq:yx_cc} and \eqref{eq:s_cc} that $
%\byc(z)  $, $\bxc(z)  $, $\sxxc(z)  $, and $\sxyc(z)  $
%denote the finite-population means and covariances over the $z$-complete cases, $ \{i: C_i(z) = 1\}$. 
With the number of $z$-complete cases satisfying $N^\cc(z) = \sumi C_i(z) = N \bar C(z)$, 
we have
\beginy\label{eq:bridge_cc_population}
&&\bxc(z) %= \frac{\sumi C_i(z) x_i }{\sumi C_i(z)}
= \frac{ \bcx(z)}{\bc(z)}, 
\qquad \byc(z) %= \frac{\sumi C_i(z) Y_i(z)}{\sumi C_i(z)}
= \frac{ \bcy(z)}{\bc(z)}, 
\nonumber\\
%%%%%%%%%%
&& \sxxc(z)  %= \frac{\sum_{i:C_i(z) = 1} x_i x_i^\T}{N^\cc(z)} -  \bxc(z)\bxc(z)^\T
 % \frac{\sumi C_i(z) \{x_i - \bxc(z) \}\{x_i - \bxc(z) \}^\T  }{\sumi C_i(z)}
%  \\
%&=&  \frac{\ninv\sumi C_i(z) \{x_i - \bxc(z) \}\{x_i - \bxc(z) \}^\T  }{\bc(z)} 
=
\frac{\bcxx(z)}{\bc(z)} - \bxc(z)\bxc(z)^\T, 
\qquad 
%%%%%%%%%%%%%%%%
 %\text{and likewise} \quad 
\sxyc(z) 
 %&=& \frac{\sumi C_i(z) \{x_i - \bxc(z) \}\{Y_i - \byc(z) \}  }{\sumi C_i(z)} \\
%&=&  \frac{\ninv\sumi C_i(z) \{x_i - \bxc(z) \}\{Y_i - \byc(z) \}  }{\bc(z)} 
= 
\frac{\bcxy(z)}{\bc(z)} - \bxc(z)\byc(z). 
\endy
%Equation  \eqref{eq:bridge_cc_population} connects the population means and covariances of $Y_i(z)$ and $x_i$ over the $z$-complete cases, namely $
%\byc(z)  $, $\bxc(z)  $, $\sxxc(z)  $, and $\sxyc(z)  $ over $\{i:C_i(z) =1\}$, back to those of $C_i(z)$, $ C_i(z)x_i$, and $  C_i(z)Y_i(z)$ over  the entire finite population, namely $\bc(z)$, $\bcx(z)$, $\bcy(z)$, $\bcxx(z)$, and $\bcxy(z)$ over $i = \ot{N}$. 
%

%Recall $\{i: C_i = 1\}$ as the observed complete cases with all covariates observed under the realized assignment $Z_i$. 
Let $\nzzc = \sum_{i: C_i = 1} I(Z_i = z) = \sumiz C_i$ be the number of observed complete cases assigned to level $z$, with $e_z^\cc  = \nzzc /\nc $ as the corresponding proportion among all observed complete cases. 
Let  $\hyc(z) , \hxc(z) , \hsxxc(z)$, and $\hsxyc(z)$
%$
%\hyc(z) = (\nzzc)^{-1}  \sumizcc    Y_i$, $
%\hxc(z) = (\nzzc)^{-1}  \sumizcc    x_i $, 
%$\hsxxc(z)  = (\nzzc)^{-1}\sumizcc \{x_i - \hxc(z) \}\{x_i - \hxc(z) \}^\T$, 
%and 
%$\hsxyc(z)  = (\nzzc)^{-1}\sumizcc \{x_i - \hxc(z) \}\{Y_i - \hyc(z) \}$
%and 
%\begina
%&&\hsxxc(z)  = (\nzzc)^{-1}\sumizcc \{x_i - \hxc(z) \}\{x_i - \hxc(z) \}^\T, \\
%&&\hsxyc(z)  = (\nzzc)^{-1}\sumizcc \{x_i - \hxc(z) \}\{Y_i - \hyc(z) \}
%\enda
be the corresponding sample means and covariances of $\{Y_i, x_i\}_{i:C_i=1,Z_i=z}$. 
% over the observed complete cases under treatment $z$.
With  
$
\nzzc = \sumiz C_i = \hc(z) \nz 
$
by definition, we have 
%\beginy\label{eq:bridge_cc_sample}
%&&e_z^\cc   = \frac{N_z^\cc}{N^\cc_1 + N^\cc_0} 
%= \frac{e_z \hat C(z)}{ e_1 \hat C(1) + e_0 \hat C(0)}, \nonumber\\
%&& \hyc(z) = \frac{\sum_{i: Z_i =z} C_i  Y_i }{\sumiz C_i }= \frac{ \hcy(z)}{\hc(z)}, \qquad \hxc(z)  = \frac{\sum_{i: Z_i =z} C_i  x_i }{\sumiz C_i }= \frac{ \hcx(z)}{\hc(z)}, \nonumber \\
%&&\hsxxc(z)  =  \frac{\sumicz x_i  x_i^\T}{\nzzc} -  \hxc(z) \hxc(z)^\T 
%% &=& 
%%\left\{ (\nzzc)^{-1}\sum_{i: Z_i = z, C_i = 1} x_i  x_i^\T\right\} -  \hxc(z) \hxc(z)^\T \\
%%  &=& 
%%  \frac{N_z^{-1}\sum_{i: Z_i = z } C_i x_i  x_i^\T}{ \hc(z)} -  \hxc(z) \hxc(z)^\T \\
%= \frac{\hcxx(z)}{\hc(z)} - \hxc(z) \hxc(z)^\T ,\nonumber\\
%&&\text{and likewise} \quad 
%\hsxyc(z) = \frac{\hcxyt(z)}{\hc(z)} - \hxc(z) \hyc(z)^\T. 
%\endy
\beginy\label{eq:bridge_cc_sample}
&&e_z^\cc   = \frac{N_z^\cc}{N^\cc_1 + N^\cc_0} 
= \frac{e_z \hat C(z)}{ e_1 \hat C(1) + e_0 \hat C(0)}, 
\qquad 
\hxc(z)  %= \frac{\sum_{i: Z_i =z} C_i  x_i }{\sumiz C_i }
= \frac{ \hcx(z)}{\hc(z)},
\qquad  \hyc(z) %= \frac{\sum_{i: Z_i =z} C_i  Y_i }{\sumiz C_i }
= \frac{ \hcy(z)}{\hc(z)},  \nonumber \\
&&\hsxxc(z)  %=  \frac{\sumicz x_i  x_i^\T}{\nzzc} -  \hxc(z) \hxc(z)^\T 
% &=& 
%\left\{ (\nzzc)^{-1}\sum_{i: Z_i = z, C_i = 1} x_i  x_i^\T\right\} -  \hxc(z) \hxc(z)^\T \\
%  &=& 
%  \frac{N_z^{-1}\sum_{i: Z_i = z } C_i x_i  x_i^\T}{ \hc(z)} -  \hxc(z) \hxc(z)^\T \\
= \frac{\hcxx(z)}{\hc(z)} - \hxc(z) \hxc(z)^\T ,
%\nonumber\\
%&& \text{and likewise} 
\qquad 
\hsxyc(z) = \frac{\hcxyt(z)}{\hc(z)} - \hxc(z) \hyc(z)^\T, 
\endy
 where
$\hc(z)$, $\hcx (z)$, $\hcy(z)$, $ \hcxxt (z)$, and $\hcxyt (z)$ 
are the sample analogs of $\bc(z)$, $\bcx(z)$,$\bcy(z)$, $\bcxxt(z)$, and $\bcxy(z)$ over $\{i: Z_i = z\}$,  respectively. 
%Equation  \eqref{eq:bridge_cc_sample} connects the sample means and covariances of $Y_i(z)$ and $x_i$ over the observed complete cases, namely $
%\hyc(z)  $, $\hxc(z)  $, $\hsxx^\cc(z)  $, and $\hsxy^\cc(z)  $ over $\{i:C_i  =1, Z_i = z\}$, back to those of $C_i(z)$, $ C_i(z)x_i$, and $  C_i(z)Y_i(z)$ over the entire finite population, namely $\hc(z)$, $\hcx(z)$, $\hcy(z)$, $\hcxx(z)$, and $\hcxy(z)$ over $\{i: Z_i = z\}$. 

\begin{lemma}\label{lem:lim_cc}
Assume complete randomization and Condition \ref{asym}. We have
\begina
&&\hat C(z) - \bar C(z) = \op, \qquad\hcx(z) - \bcx(z) = \op, \qquad \hcy(z) - \bcy(z) = \op, \nonumber \\
&&\hcxxt(z) - \bcxxt(z) = \op, \qquad \hcxyt(z)- \bcxyt(z) = \op
\enda
with $\bc(z)$, $\bcx(z)$, $\bcy(z)$, $\bcxxt(z)$, and $\bcxy(z)$ all have finite limits as $N\to\infty$. This ensures  
\begina
&&e_z^\cc - \rho_z^\cc = \op, \qquad \hyc(z) -\byc(z) =\op, \qquad \hxc(z) - \bxc(z) = \op, \\
&&\hsxxc(z) - \sxx^\cc(z) = \op, \qquad \hsxyc(z) - \sxy^\cc(z) = \op 
\enda
by \eqref{eq:bridge_cc_population} and \eqref{eq:bridge_cc_sample}, 
recalling $
\rho_z^\cc 
=   e_z \bar C(z) / \{e_1 \bar C(1) + e_0  \bar C(0)\}$  
for $z = 0,1$. 
\end{lemma}

\subsection{Probability limits of $\hts^\cc \ (\md = \text{F}, \text{L})$}
We next verify the result of Proposition \ref{prop:cc_asym} without Condition \ref{cond:a}. 
The proof affords important intermediate steps for proving Proposition \ref{prop:cc_1} under Condition \ref{cond:a}.

%Let $\hg_{\fisher}^\cc $ be the coefficient of $x_i$ from the additive regression \eqref{eq:cc_fisher} based on the complete cases.
%Let $\hg_{\lin,z}^\cc $  be the coefficient  of $x_i$ from the treatment-wise regression 
%$Y_i \sim 1 + x_i$ over $\{i: i \in C_i = 1, Z_i = z\}$ based on the complete cases under treatment $z$. 

\begin{proof}[Proof of Proposition \ref{prop:cc_asym}]
By Lemma \ref{lem:gf_gl}, 
we have 
\beginy\label{lem:algebra_cc}
 \htsc = \hyc(1)  - \hyc(0) -  \{\hxc(1)  - \hxc(0) \}^\T \hgsc,
\endy where  
$
\hgc  =
\{ e_0^\cc  \hsxxc(0) + e_1^\cc \hsxxc(1) \}^{-1}
\{ e_0^\cc  \hsxyc(0) + e_1^\cc \hsxyc(1) \}$ and $\hglc = e_0^\cc \hgoc +  e_1^\cc \hgzrc
$ with $
\hgzc = \{\hsxxc(z) \}^{-1} \hsxyc(z) $ for $z = 0,1$.
Plugging the probability limits of $e_z^\cc$, $\hyc(z)$, $\hxc(z)$, $\hsxxc(z)$, and $\hsxyc(z)$ from Lemma \ref{lem:lim_cc} in \eqref{lem:algebra_cc}  verifies
\beginy\label{eq:gamma_lim_cc}
\hgf^\cc - \gf^\cc = \op, \qquad \hglzz^\cc - \glzz = \op, \qquad \hgl^\cc = \gl^\cc + \op
\endy
and thus 
 the probability limits of $\htsc$ for $\mds  $. 

The sufficient and necessary condition for $\tau^\cc = \tau$ follows from $\tau^\cc  = \byc(1) - \byc(0) $ and
$
\by(z)  = \bc(z) \byc(z)	 + \{1-\bc(z)\} \byu(z) ,\quad (z=0,1)$. 
%such that $
%\tau =\bc(1) \byc(1)	 + \{1-\bc(1)\} \byu(1)  - \left[ \bc(0) \byc(0)	 + \{1-\bc(0)\} \byu(0) \right]$ 
%and thus 
%\begina
%\tau^\cc - \tau 
%&=& \byc(1) - \byc(0) - \bc(1) \byc(1)	-\{1-\bc(1)\} \byu(1) +  \bc(0) \byc(0)	 + \{1-\bc(0)\} \byu(0) \\&=&\{1 - \bc(1) \}\{ \byc(1)	- \byu(1)\}  -  \{1-\bc(0)\}\{  \byc(0) - \byu(0) \}.
%\enda 

%When $M_i(1) \neq M_i(0)$ in the absence of Condition \ref{cond:a}, $\bxc(1)$ may or may not equal $\bxc(0)$ such that $\htsc $ is not necessarily consistent for $\tau$ even when $\tau^\cc   = \tau$ with asymptotic bias $\{\bxc(1) - \bxc(0)\}^\T\gamma_\md^\cc$. 
%% 
%When $M_i(1) = M_i(0)$ under Condition \ref{cond:a}, we have $C_i(z) = C_i$, $\bc(z) = \bc = \meani C_i $, and $\bcx(z) = \bcx = \meani C_ix_i$ such that  $\bx^\cc(1 ) = \bx^\cc(0)$ and $
%\htsc  =\tz +\op$. 
\end{proof}

\subsection{Asymptotic normality and variance estimation under Condition \ref{cond:a}}
We next verify the result of Proposition \ref{prop:cc_1} under Condition \ref{cond:a}. Assume throughout this subsection that Condition \ref{cond:a} holds with $C_i(0) = C_i(1) = C_i$ for all $i$. The $z$-complete cases thus equal the observed complete cases, denoted by $\{i:C_i = 1\}$.  
The expressions of
$\bc(z)$, $\bcx(z)$, $\bcy(z)$, $\bcxx(z)$, $\bcxy(z)$ simplify to
\begina
&&\bar C(z) = \meani C_i \equiv \bar C, \qquad
\bcx(z) = \meani C_i x_i \equiv \bcx, \qquad \bcy(z) = \meani C_i Y_i(z),\\
&&\bcxx(z) = \meani C_i x_ix_i^\T \equiv \bcxx, \qquad \bcxy(z) = \meani C_i x_i Y_i(z).
\enda
This, together with \eqref{eq:bridge_cc_population}, ensures
\beginy\label{eq:bridge_population_4}
&& N^\cc(z) = N \bar C = N^\cc, \qquad \byc(z) =  \frac{\bcy(z)}{\bc}, \qquad \bxc(z) =   \frac{\bcx}{\bc} \equiv \bxc, \nonumber\\
&& \sxxc(z) =
\frac{ \bcxx}{\bc}   - \bxc ( \bxc) ^\T \equiv \sxxc, \qquad 
 \sxyc(z) 
=
 \frac{\bcxy(z)}{\bc}  - \bxc \byc(z)
% , \\
%&& \rho_z^\cc = e_z, \qquad \gamma^\cc_{\lin,z} = (\sxx^\cc)^{-1}\sxyc(z), \qquad   \gf ^\cc = e_0\gamma^\cc_{\lin,0} + e_1 \gamma^\cc_{\lin,1}, \qquad  \gamma_\lin^\cc = e_1\gamma^\cc_{\lin,0} + e_0 \gamma^\cc_{\lin,1}\nonumber
\endy
with 
$
\gf ^\cc  =
 \rho_1^\cc  \glo^\cc +  \rho_0^\cc \glz^\cc$ and $
\gamma_{\lin,z}^\cc = (\sxxc)^{-1} \sxyc(z)$ for $z = 0,1$ 
from Proposition \ref{prop:cc_asym}.  
Let 
\begina
&&w_{i, \fisher}(z) = Y_i(z) - \byc(z) - (x_i - \bxc)^\T\gf^\cc, \\
&&w_{i, \lin}(z) = Y_i(z) - \byc(z) - (x_i - \bxc)^\T\glzz^\cc
\enda
be the adjusted potential outcomes with finite-population means $\bar w_\md(z) =0$ and covariances $S^\cc_{zz', \md} $ for $z, z' = 0,1$ over $\{i: C_i = 1\}$.
The explicit forms of $v_\md^\cc$ and $S_{\tau\tau,\md}^\cc$ in Proposition \ref{prop:cc_1} are
\beginy\label{prop:cc_app}
&&S_{\tau\tau,\fisher}^\cc = (N^\cc)^{-1}\sumic (\tau_i - \tau^\cc)^2, \nonumber\\
&&S_{\tau\tau,\lin}^\cc = (N^\cc)^{-1}\sumic \{\tau_i - \tau^\cc - (x_i - \bxc)^\T(\glo^\cc - \glz^\cc)\}^2 \nonumber\\
&&v_\md^\cc  =  e_0^{-1} S^\cc_{00, \md} +e_1^{-1} S^\cc_{ 11, \md} - S^\cc_{ \tau\tau, \md} \formds,
\endy
respectively.  
Intuitively, $S^\cc_{zz',\fisher}$ and $S^\cc_{zz',\lin}$ give the finite-population covariances of $\{ Y_i(z) - x_i^\T\gf^\cc, Y_i(z') -x_i^\T\gf^\cc\}_{i: C_i = 1}$  and $\{ Y_i(z) - x_i^\T\glzz ^\cc, Y_i(z') -x_i^\T\gamma_{\lin, z'} ^\cc\}_{i: C_i = 1}$, respectively; $S_{\tau\tau, \fisher}^\cc$ and $S_{\tau\tau, \lin}^\cc$ give the finite-population variances of $(\tau_i)_{i: C_i = 1}$ and $\{\tau_i -x_i^\T (\gamma^\cc_{\lin,1} - \gamma^\cc_{\lin,0})\}_{i: C_i =1}$, respectively. 
 
\begin{proof}[Proof of Proposition \ref{prop:cc_1}]
Let 
$\hat w_\md(z) = (N_z^\cc)^{-1}\sumicz w_{i,\md}(z)$ be the sample averages of $w_{i, \md}(z)$ over $\{i: C_i =1, Z_i = z\}$ with 
\beginy\label{eq:w_fl}
\hat w_\fisher(z)   =  \hyc(z) - \byc(z) - \{\hxc(z) - \bxc\}^\T \gf^\cc, \nonumber\\
\hat w_\lin(z)   =   \hyc(z) - \byc(z) - \{\hxc(z) - \bxc\}^\T \glzz^\cc. 
\endy
We proceed with the proof in the following five steps:
\begine[(i)]
\item\label{step1:cc}
$
 \hts^\cc - \tau^\cc 
 \asim   
   \hat w_\md(1) - \hat w_\md(0)$ for $\mds$.
\item\label{step2:cc} $ N^\cc \var_\infty( \hts^\cc) = v_\md^\cc$  with $v_\md^\cc =  e_0^{-1} S^\cc_{00, \md} +e_1^{-1} S^\cc_{ 11, \md} - S^\cc_{ \tau\tau, \md} $   defined in \eqref{prop:cc_app}. 
\item\label{step3:cc} $v_\lin^\cc \leq v_\fisher^\cc$. 
\item\label{step4:cc} $N^\cc (\hses^\cc)^2 - v_\md^\cc = S_{\tau\tau,\md}^\cc + \op$. 
\item\label{step5:cc} The sufficient and necessary conditions for $\tau^\cc = \tau$. 
\ende

\paragraph{Proof of step \eqref{step1:cc}.}
Given 
\begina
&& \htf^\cc  - \tau^\cc 
=  \{\hyc(1) -\byc(1)\} - \{\hyc(0) -\byc(0)\} -  \{\hxc(1)  - \hxc(0) \}^\T \hgfc  \nonumber \\
&&\htl^\cc - \tau^\cc 
= \left[\hyc(1) -\byc(1) - \{\hxc(1)  -\bxc \}^\T \hglo^\cc \right]  - \left[ \hyc(0) - \byc(0) -  \{\hxc(0)  - \bxc \}^\T \hglz^\cc \right]
\enda
by \eqref{lem:algebra_cc}, 
we have 
\begina
&& \htf^\cc  - \tau^\cc  - \{ \hat w_\fisher(1) - \hat w_\fisher(0) \}
= 
 - \{\hxc(1)  - \hxc(0) \}^\T (\hgfc- \gf^\cc) , \nonumber\\
  %%%%%%%%%%%%%
  %%%%%%%%%%%%%%%%
&& \htl^\cc - \tau^\cc -\{ \hat w_\lin(1) -\hat w_\lin(0)\}
=
   - \{\hxc(1)  -\bxc \}^\T (\hglo^\cc - \glo^\cc ) + \{\hxc(0)  - \bxc \}^\T (\hglz^\cc - \glz^\cc )  
\enda
by \eqref{eq:w_fl} such that it suffices to verify 
\begina
\sqrt N \{\hxc(1)  -\hxc(0) \}^\T  (\hgf^\cc  - \gf^\cc) = \op, \qquad \sqrt N \{\hxc(z)  -\bxc \}^\T (\hglzz^\cc  - \glzz^\cc) = \op.
\enda
This is indeed correct by \eqref{eq:bridge_cc_sample}, Lemma \ref{lem:joint_dist}, the delta method, and \eqref{eq:gamma_lim_cc}.

%. In particular, Lemma \ref{lem:joint_dist} and the delta method together ensure that 
%\begina
%\sqrt N\{  \hxc(z) - \bxc   \}
%= \sqrt N\left\{ \frac{\hcx(z)}{\hat C(z)}- \frac{\bcx}{\bc}  \right\} 
%\enda
%and  $\sqrt N \{\hxc(1)  - \hxc(0) \} = \sqrt N \{\hxc(1)  - \bxc  \}  + \sqrt N \{\hxc(0)  - \bxc  \} $ are asymptotically normal with the limit of $\bar C > 0$ by Condition \ref{asym_1}. 
%%
%Condition \eqref{eq:ssss} follows from \eqref{eq:gamma_lim_cc} and Slutsky's theorem. 

\paragraph{Proof of step \eqref{step2:cc}.} We verify below the result for $N\var_\infty( \htf^\cc)$. The proof  for $N\var_\infty( \htl^\cc)$ is analogous by replacing  $\gf$ with $\glzz$ in \eqref{eq:U} below. 

Consider pseudo potential outcomes 
\beginy\label{eq:U}
U_i(z) = \left\{\bar C C_i Y_i(z) - \bcy(z) C_i \right\} - (\bar C C_i x_i - \bcx C_i  )^\T \gf^\cc
\endy
for $z = 0, 1$ and $i = \ot{N}$. We can verify that the finite-population means and covariances of $\{U_i(0), U_i(1)\}_{i=1}^N$  equal 
$
\bar U(z) 
%= \meani U_i(z) =  \{\bar C \bcy(z) - \bcy(z) \bc \} -  (\bar C \bcx - \bcx \bc )^\T\gf^\cc 
= 0$ for $z=0,1$ 
and 
\beginy\label{eq:UF}
S_{U(z), U(z')} = \meani U_i(z) U_i(z')  
=  \bar C^3 S^\cc_{zz', \fisher} \qquad \text{for} \quad z, z' = 0,1,
\endy  
respectively.
Let $
\hat U(z) = \meaniz U_i(z)$ 
be the sample analog of $\bar U(z)$ over $\{i: Z_i =z\}$. 
Standard results of complete randomization, together with \eqref{eq:UF},  ensure that 
$\sqrt N( \hat U(0), \hat U(1) )^\T$ is asymptotically normal under Condition \ref{asym_1} with covariance
\beginy\label{eq:covU}
N\cov\left\{ \beginp \hat U(0)\\ \hat U(1)\endp \right\}
& =&  \diag(e_z^{-1} S_{U(z), U(z)})_{z=0,1} - S_{\{U, U\}} \nonumber\\
&=& \bar C^3 \left\{ \diag(e_z^{-1} S^\cc_{zz, \fisher})_{z=0,1} - S^\cc_\fisher \right\} 
=  \bar C^3 (\Phi \circ S^\cc_\fisher),
\endy
where 
\begina
S_{\{U, U\}} = 
\beginp
S_{U(0), U(0)} & S_{U(0), U(1)} \\
S_{U(1), U(0)} & S_{U(1), U(1)} 
\endp, \qquad 
S^\cc_{\fisher}
= \beginp
S^\cc_{00, \fisher} & S^\cc_{01, \fisher}\\
S^\cc_{10, \fisher} & S^\cc_{11, \fisher}
\endp, \qquad \Phi = \diag(e_z^{-1})_{z=0,1} - 1_{2\times 2}. 
\enda

On the other hand, plug \eqref{eq:U} in the definition of $\hat U(z)$ to see
\begina
\hat U(z) = \meaniz U_i(z) 
&=& \left\{\bar C \hcy (z) - \bcy(z) \hc(z) \right\} - \left\{\bar C \hcx(z) - \bcx \hc(z) \right\}^\T\gf^\cc\\
%%%%%%%%%%%
&=& \{\hat C(z) \bc\}  \left[ \left\{\frac{\hcy(z)}{\hc(z)} - \frac{\bcy(z)}{\bc}\right\}  -  \left\{\frac{\hcx(z)}{\hc(z)} - \frac{\bcx}{\bc}\right\}^\T \gf^\cc   \right] \\
%%%%%%%%%%%
&=& \{\hat C(z) \bc\}  \left[ \{\hyc(z) - \byc(z)\}  -  \left\{\hxc(z) - \bxc \right\}^\T\gf^\cc  \right]\\
&=& \{\hat C(z) \bc\}  \hat w_\fisher(z) 
\enda
by \eqref{eq:bridge_cc_sample} and \eqref{eq:bridge_population_4}. 
This, together with Slutsky's theorem, ensures 
$(\hat w_\fisher(0), \hat w_\fisher(1) )^\T 
\asim 
\bc^{-2} 
( \hat U(0),  \hat U(1))^\T$
and thus 
\begina
\htf^\cc - \tau^\cc \asim \bc^{-2}  (-1, 1) \beginp \hat U(0)\\ \hat U(1)\endp 
\enda 
by  step \eqref{step1:cc}. 
This, together with \eqref{eq:covU}, verifies the result for $\var_\infty(\htf^\cc)$ with
\beginy\label{eq:vfcc}
 N \var_\infty(\htf^\cc) &=& \bar C^{-1} (-1, 1) (\Phi \circ S^\cc_\fisher) (-1, 1)^\T \\
&=&\bar C^{-1} \left[ (-1, 1) \{ \diag(e_z^{-1})_{z=0,1} \circ S^\cc_\fisher\} (-1, 1)^\T -  (-1, 1)  S^\cc_\fisher  (-1, 1)^\T \right]\nonumber \\
&=&  \bar C^{-1}  ( e_0^{-1} S^\cc_{00, \fisher} +e_1^{-1} S^\cc_{ 11, \fisher} - S^\cc_{ \tau\tau, \fisher}),\nonumber 
\endy
where $N \bar C = N^\cc$; 
the last equality follows from 
$
(-1, 1) S_\fisher^\cc (-1, 1)^\T = 
 (N^\cc)^{-1}\sumic [\{Y_i(1) - Y_i(0) \} - \{\byc(1) - \byc(0)\} ]^2 = S_{\tau\tau, \fisher}^\cc. 
$

% I restored this because {eq:vlcc} is used in the proof of step (iii) -- anqi 07/27
Almost identical reasoning, after replacing  $\gf$ with $\glzz$ in \eqref{eq:U},  yields
\beginy\label{eq:vlcc}
 N^\cc \var_\infty(\htl^\cc) =   (-1, 1) (\Phi \circ S^\cc_\lin) (-1, 1)^\T =    e_0^{-1} S^\cc_{00, \lin} +e_1^{-1} S^\cc_{ 11, \lin} - S^\cc_{ \tau\tau, \lin} ,
\endy
where $S^\cc_{\lin} = (S^\cc_{zz', \lin})_{z, z' = 0, 1}$ is the $2\times 2$ matrix summarizing $S^\cc_{zz', \lin}$'s in lexicographical order. 

\paragraph{Proof of step \eqref{step3:cc}.} 
The proof is similar to that of Lemma \ref{lem::basic-x}. We give a sketch of it.
Equations \eqref{eq:vfcc} and \eqref{eq:vlcc} together ensure that $$
v_\fisher^\cc - v_\lin^\cc  =  (-1, 1) \{\Phi\circ(S_\fisher^\cc - S_\lin^\cc)\}(-1,1)^\T.$$
With $\Phi$ being positive semi-definite, Schur product theorem \citep[][Theorem 3.1]{schur} ensures that $v_\lin^\cc \leq v_\fisher^\cc$ as long as $S_\fisher^\cc - S_\lin^\cc \geq 0$. This is indeed correct because 
%
%This is indeed correct. In particular, $w_{i, \lin}(z)$ equals the least-squares residual from the population regression $Y_i(z) \sim 1 + (x_i - \bx)$ over $\{i: C_i = 1\}$. 
%The properties of least squares ensure 
%\begina
%(N^\cc)^{-1} \sumic (x_i - \bar x^\cc) w_{i, \lin}(z)  = 0_J \qquad \text{for}\quad z = 0, 1. 
%\enda
%This, together with 
%$
%w_{i, \fisher}(z) - w_{i, \lin}(z) =   (x_i - \bxc)^\T(\glzz^\cc - \gf^\cc) 
%$ by definition, 
%ensures that 
%\begina
%(N^\cc)^{-1} \sumic w_{i, \lin}(z)\{w_{i, \fisher}(z') - w_{i, \lin}(z')\} = 0\qquad \text{for}\quad z, z' = 0,1.
%\enda
%As a result, we have 
%\begina
%S_{zz', \fisher}^\cc
%& =&  (N^\cc)^{-1}\sumic  w_{i,\fisher}(z)w_{i,\fisher}(z')  \\
%&=&  (N^\cc)^{-1}\sumic \left[ \{ w_{i,\fisher}(z) - w_{i,\lin}(z)\} + w_{i,\lin}(z) \right] \left[ \{ w_{i,\fisher}(z') - w_{i,\lin}(z')\} + w_{i,\lin}(z') \right]  \\
%&=&  (N^\cc)^{-1}\sumic \{ w_{i,\fisher}(z) - w_{i,\lin}(z)  \} \{ w_{i,\fisher}(z') - w_{i,\lin}(z')  \}\\
%&& + (N^\cc)^{-1}\sumic  w_{i,\lin}(z)  \{ w_{i,\fisher}(z') - w_{i,\lin}(z')  \}
%+
%(N^\cc)^{-1}\sumic \{ w_{i,\fisher}(z) - w_{i,\lin}(z)  \}   w_{i,\lin}(z') 
% \\
% &&+ (N^\cc)^{-1}\sumic   w_{i,\lin}(z)w_{i,\lin}(z')\\
%&=&(N^\cc)^{-1}\sumic \{ w_{i,\fisher}(z) - w_{i,\lin}(z)  \} \{ w_{i,\fisher}(z') - w_{i,\lin}(z')  \} + S_{zz', \lin}^\cc \qquad \text{for}\quad z, z' = 0, 1 
%\enda
%such that 
\begina
S_\fisher^\cc - S_\lin^\cc 
%&=& \beginp
%S_{00, \fisher}^\cc - S_{00, \lin}^\cc& S_{01, \fisher}^\cc - S_{01, \lin}^\cc\\
%S_{10, \fisher}^\cc - S_{10, \lin}^\cc& S_{11, \fisher}^\cc - S_{11, \lin}^\cc
%\endp\\
&=& (N^\cc)^{-1}\sumic 
\beginp
 w_{i,\fisher}(0) - w_{i,\lin}(0) \\
  w_{i,\fisher}(1) - w_{i,\lin}(1) 
\endp
\beginp
 w_{i,\fisher}(0) - w_{i,\lin}(0), \ 
  w_{i,\fisher}(1) - w_{i,\lin}(1) 
\endp \geq 0. 
\enda

\paragraph{Proof of step \eqref{step4:cc}.}
Let $\ep^\cc_{\fisher,i}$ and $\ep^\cc_{\lin,i}$ be the residuals from \eqref{eq:cc_fisher} and \eqref{eq:cc_lin}, respectively, 
%with 
%\begina
%&& \ep^\cc_{\fisher,i} = Y_i - \hat\mu_{\fisher}^\cc - Z_i \htf^\cc - x_i ^\T \hg_\fisher^\cc, \\
%&& \ep^\cc_{\lin,i} = Y_i - \hat\mu_{\lin}^\cc - Z_i \htl^\cc - ( x_i -\bx^\cc)^\T \hg_{\lin,0}^\cc -  Z_i ( x_i -\bx^\cc)^\T (\hg_{\lin,1}^\cc - \hg_{\lin,0}^\cc) 
%\enda
for $\{i: C_i = 1\}$. 
Let $\hat\ep^\cc_\md(z) $ and 
$
\hat S_{\ep\ep,\md}^\cc(z)$
%Let $\hat\ep^\cc_\md(z) = (N_z^\cc)^{-1}\sum_{i:C_i = 1, Z_i = z} \ep^\cc_{\md,i}$ and 
%$
%\hat S_{\ep\ep,\md}^\cc(z)  = (N_z^\cc)^{-1}\sum_{i:C_i = 1, Z_i = z}\{\ep^\cc_{\md,i} - \hat\ep^\cc_{\md}(z)\}^2 
%$
  be the sample mean and variance of $\ep^\cc_{\md,i}$'s over $\{i: C_i = 1, Z_i = z\}$. 
We have 
\beginy\label{eq:nse}
N^\cc (\hse^\cc_\md)^2 - 
\left\{ (e_0^\cc)^{-1} \hat S_{\ep\ep,\md}^\cc(0) + 
 (e_1^\cc)^{-1} \hat S_{\ep\ep,\md}^\cc(1)
 \right\} = \op \qquad \text{for} \quad\mds 
\endy 
by \cite{LD20} and \cite{ZDfrt}. 

On the other hand, let 
$
\bcyy(z) = \meani C_i Y_i(z)^2$ and $\hcyy(z) = \meaniz C_i Y_i^2$ be the analogs of $\bcxx$ and $\hcxx$ defined on $Y_i(z)$ and $Y_i$.
We have $ \hcyy(z) - \bcyy(z)= \op$ under complete randomization and Condition \ref{asym_1} with $\bcyy(z)$ having a finite limit. 
Let 
\begina
S^\cc_{YY}(z) = (N^\cc)^{-1}\sumic \{Y_i (z) - \byc(z)\}^2= \frac{\bcyy(z)}{\bc} - \{ \byc(z)\}^2  
\enda
be the finite-population variance of $\{Y_i(z):C_i = 1\}$ with sample analog
\begina
\hat S^\cc_{YY}(z) = (N_z^\cc)^{-1}\sumicz \{Y_i - \hy^\cc(z)\}^2 %=\left\{ (N_z^\cc)^{-1}\sumicz Y_i^2\right\}  - \hy^\cc(z)^2 
= \frac{\hcyy(z)}{\hc(z)} -  \hy^\cc(z)^2   =S^\cc_{YY}(z)+ \op .
\enda
We have 
\begina
\hat S_{\ep\ep,\fisher}^\cc(z)
&=&\hat S^\cc_{YY}(z)  + (\hgf^\cc)^\T \hsxx^\cc(z) \hgf^\cc - 2 (\hgf^\cc)^\T   \hsxy^\cc(z)\\
&=& S_{YY}^\cc(z) +(\gf^\cc)^\T \sxx^\cc  \gf^\cc - 2 (\gf^\cc)^\T   \sxy^\cc(z) + \op\\
&=&
(N^\cc)^{-1}\sumic \left[ Y_i(z) - x_i^\T \gf^\cc - \{\by^\cc(z) - (\bxc)^\T\gf^\cc \} \right]^2  +\op\\
&=& S_{zz, \fisher}^\cc + \op
\enda
and likewise $\hat S_{\ep\ep,\lin}^\cc(z) = S_{zz, \lin}^\cc + \op$. 
Plugging these probability limits  in \eqref{eq:nse} verifies the result. 

\paragraph{Proof of step \eqref{step5:cc}.} The equivalence relationship follows from $\tau = \bc \tau^\cc + (1-\bc) \tau^\uc $, $S_{C, \tau}  =  \bar C(\tau^\cc - \tau)$, and simple algebra. We omit the details.
%The equivalence between $\tz = \tau$ and $\tz=\tone$ follows from \eqref{eq:by} with $\tau = \bc \tau^\cc + (1-\bc) \tau^\uc $ under Condition \ref{cond:a}. 
%The sufficiency and necessity of $S_{C,\tau} = 0$ follows from 
%$
%S_{C, \tau}  = N^{-1}\sumi C_i \tau_i - \bar C \tau =  N^{-1}( \nc  \tau^\cc) -  \bar C  \tau =   \bar C(\tau^\cc - \tau)$. 
%
%Further let  
% $S_{C, Y(z)} = N^{-1}\sumi (C_i - \bar C) \{Y_i(z) - \bar Y(z)\}$ be the finite-sample covariance of $(C_i)_{i=1}^N$ and $\{Y_i(z)\}_{i=1}^N$. 
%In the case of $J = 1$, it follows from $S_{M, Y(z)} = - S_{C, Y(z)}$ that  $S_{C, \tau} = S_{C, Y(1)} - S_{C, Y(0)} =S_{M, Y(0)} - S_{M, Y(1)} $.  
\end{proof}

\section{Single imputation and missingness-indicator method}

\subsection{Useful facts without Condition \ref{cond:a}}
%We first introduce some additional notations and useful facts. All results in this subsection hold regardless of whether Condition \ref{cond:a} holds or not. 

Recall  $\xs_i(c)$ and $x_i^\mim(c) $ as the covariates  under  single imputation and the missingness-indicator method, respectively, with $c = (c_1, \ldots, c_J)^\T \in \mathbb{R}^J$. 
We have 
\beginy\label{eq:xi_imp}
\xs_i(c)= A_i \circ  x_i + \dc M_i, \qquad \text{where} \quad \dc   = \diag(c_j)_{j=1}^J. 
\endy
Recall  $\xidg(z; c)$ as the potential value of $\xidg(c)$ for $\dgsss$ if unit $i$ were assigned to treatment $z \in \{0,1\}$. We have 
\begina
x^\im  _i(z; c) = A_i(z) \circ  x_i + \dc M_i(z), \qquad 
x_i^\mim(z; c) =  (x_i^\imp(z; c)^\T, M_i(z)^\T)^\T. 
\enda
Let $  \bx^\dg(z; c)$, $\sxx^\dg (z;c)$, and $\sxy^\dg (z;c)$ be the mean and finite-population covariances of $\{\xidg(z; c) \}_{i=1}^N$ and $\{\xidg(z; c) , Y_i(z)\}_{i=1}^N$ over $i = \ot{N}$, respectively. 
Let 
$\hx^\dg   (z; c)$, 
$  \hsxx^\dg   (z; c)$,  and $\hsxy^\dg   (z;c)$ 
be the sample analogs over $\{i: Z_i = z\}$.
% of $\bx^\dg(z; c)$, $S_{xx}^\dg(z; c)$, and $S_{xY}^\dg(z; c)$, respectively, 
%Let $  \bx^\dg(z; c) = \meani \xidg(z; c)$ be the mean of $\xidg(z; c)$ over $i = \ot{N}$ for $\dgsss$.
% Let 
%$\sxx^\dg (z;c) = \meani  \{ \xidg(z;c) - \bx^\dg(z;c)\}\{\xidg(z;c) - \bx^\dg(z;c)\}^\T$ and $\sxy^\dg (z;c) = \meani  \{ \xidg(z;c) - \bx^\dg(z;c)\}\{ Y_i(z) - \by(z)\} $
%be the finite-population covariances of $\{\xidg(z; c) \}_{i=1}^N$ and $\{\xidg(z; c) , Y_i(z)\}_{i=1}^N$, respectively. 

%Let 
%$\hx^\dg   (z; c) = \meaniz  x_i^\dg   (c)$, 
%$  \hsxx^\dg   (z; c) = \meaniz  \{x_i^\dg   (c) -\hx^\dg   (z; c)\} \{x_i^\dg   (c) -\hx^\dg   (z;c)\}^\T$,  and $\hsxy^\dg   (z;c) = \meaniz  \{x_i^\dg   (c) -\hx^\dg   (z;c)\} \{Y_i  -\hy(z)\}^\T$ 
%be the sample analogs of $\bx^\dg(z; c)$, $S_{xx}^\dg(z; c)$, and $S_{xY}^\dg(z; c)$ for $\dgsss$.

%Recall $\hc(z) =\meaniz C_i$, $\hax(z) = \meaniz (Ax)_i = \meaniz A_i \circ x_i$, and $\hm(z) = \meaniz M_i$ as the sample analogs of $\bc(z) = \meani C_i(z)$, $\bax(z) = \meani A_i(z) \circ x_i$, and $\barm(z) = \meani M_i(z)$, respectively. 

\begin{lemma}\label{lem:plim_imp}
Assume complete randomization, Condition \ref{asym}, and $c \in \mathcal C'$. We have
\begina
&&\hxs(z;c) - \{ \bax(z) + \dci \barm(z) \} =\op, \nonumber\\
&&\hx^\mim(z;c) - ( \{\bax(z) + \dci \barm(z)\}^\T, \bar M(z)^\T )^\T= \op,
\enda
and 
$\hsxx^\dg (z;c)- \sxx^\dg(z;\ci) = \op$ and $\hsxy^\dg (z;c) - \sxy^\dg(z;\ci) = \op$ for $\dgsss$.
\end{lemma}

\begin{proof}[Proof of Lemma \ref{lem:plim_imp}]
Recall $\hm(z)$, $\hc(z)$, and $\hax(z)$ as the sample means of $M_i$, $C_i $, and $A_i \circ x_i$ over $\{i:Z_i = z\}$,  respectively. 
It follows from \eqref{eq:xi_imp} that 
\begina
&&\hxs(z;c) = \hax(z) + \dc \hm(z),\\
&&\hxmi(z; c) =( \hxs(z; c)^\T,  \hm(z)^\T )^\T = ( \{\hax(z) + \dc \hm(z)\}^\T, \hm(z)^\T)^\T.
\enda
Their probability limits   follow from $
\hax(z) - \bax(z) = \op$, $ \hm(z) - \barm(z) = \op$, 
 and $D_c - \dci = \op$.

The probability limits of $ \hsxx^\im  (z; c) $
 and $ \hsxy^\im  (z; c) $ follow from 
\begina
&& \hsxx^\im  (z; c) 
%&=& \meaniz  \{x_i^\im  (c) -\hx^\im  (z; c)\} \{x_i^\im  (c) -\hx^\im  (z;c)\}^\T\\
%&=& \meaniz  \left[ \{A_i \circ  x_i - \hax(z)\} + D_c\{M_i - \hm(z)\} \right] \left[ \{A_i \circ  x_i - \hax(z)\} + D_c\{M_i - \hm(z)\} \right]^\T\\
= \hs_{Ax, Ax}(z) + \hs_{Ax,M}(z)\dc + \dc\hs_{M,Ax}(z)  + \dc \hs_{M,M}(z)\dc, \nonumber\\
&& \hsxy^\im  (z; c) 
%&=& \meaniz  \{x_i^\im  (c) -\hx^\im  (z; c)\} \{Y_i -\hy(z)\} \\
%&=& \meaniz  \left[ \{A_i \circ  x_i - \hax(z)\} + D_c\{M_i - \hm(z)\} \right] \{Y_i -\hy(z)\}\\
= \hs_{Ax,Y}(z) + \dc\hs_{M,Y}(z).
\enda
and 
\begina
&& \sxxim (z;c) = S_{Ax(z),Ax(z)}  + S_{Ax(z),M(z)} \dc + \dc S_{M(z), Ax(z)} + \dc S_{M(z), M(z)} (z) \dc,\nonumber\\
&& \sxyim (z;c) = S_{Ax(z),Y(z)}  +  \dc S_{M(z), Y(z)} .
\enda
%
%
%%%%%%%%%%%%%%%%%%%%%
The probability limits of $ \hsxx^\mim  (z; c) $ and $ \hsxy^\mim  (z; c) $ follow from 
$$
\hsxxmi (z ;c) =
 \beginp
 \hsxxs(z; c) &  \hs_{Ax,M}(z) + D_c \hs_{M,M}(z)\\
 \hs_{M,Ax}(z) +  \hs_{M,M}(z) D_c& \hs_{M,M}(z)
 \endp, 
\quad 
 \hsxymi (z;c) =   \beginp
 \hsxys(z; c)  \\
 \hs_{M,Y}(z)  
 \endp.
$$
and 
$$
 \sxxmi  (z; c)  = 
 \beginp
 \sxxim(z;  c) & S_{Ax(z),M(z)}  + D_c S_{M(z), M(z)}\\
 S_{M(z),Ax(z)} + S_{M(z), M(z)}D_c & S_{M(z), M(z)}
 \endp, \quad 
 \sxymi  (z; c)  =  \beginp
 \sxyim(z;  c)  \\
 S_{M(z),Y(z)}  
 \endp.
$$
\end{proof}

\subsection{Probability limits of $\htss \ (\dgsss; \ \md = \text{F}, \text{L})$ without Condition \ref{cond:a}}\label{sec:im_app}

\begin{proof}[Proof of Proposition \ref{prop:imp_mim_cim no 4}]
Applying Lemma \ref{lem:gf_gl} to $(Y_i, x_i^\dg, Z_i)_{i=1}^N$ for $\dgsss$ ensures that 
\beginy\label{lem:algebra}
\htss(c) =\htn -  \{\hx^\dg  (1;c )  - \hx^\dg  (0;c) \}^\T \hgss(c), 
\endy
%for $\dgsss$ and $\mds$
where 
$
\hgfs(c)  =
\{ e_0 \hsxx^\dg  (0;c ) + e_1\hsxx^\dg  (1;c ) \}^{-1}
\{ e_0 \hsxy^\dg  (0;c ) + e_1\hsxy^\dg  (1;c ) \}$ and $\hgls =  e_0 \hgos(c) + e_1 \hgzs(c)$ with $\hgzzs(c) = \{\hsxx^\dg  (z;c)\}^{-1}\hsxy^\dg  (z;c)$ for $z = 0,1$.
The results follow from  Lemma \ref{lem:plim_imp}  with $\hg_\md^\dg(c) - \gamma_\md^\dg(\ci) = \op$. 
\end{proof}

\subsection{Asymptotic normality and variance estimation under Condition \ref{cond:a}}
Assume Condition \ref{cond:a} throughout this subsection. 
The expressions of $\gamma_\fisher^\dg(\ci)$ and $\gamma_{\lin,z}^\dg(\ci)$ from Proposition \ref{prop:imp_mim_cim no 4} simplify to 
$
\gamma_\fisher^\dg(\ci) = e_0 \gamma_{\lin,0}^\dg(\ci) +  e_1 \gamma_{\lin,1}^\dg(\ci)$ and $ \gamma_{\lin,z}^\dg(\ci) = \{\sxx^\dg(\ci)\} ^{-1}\sxy^\dg(z; \ci)$,
respectively. Let 
\begina
Y_{i,\fisher}^\dg(z; \ci) = Y_i(z) - \{x_i^\dg(\ci)\}^\T \gf^\dg(\cinf), \quad 
Y_{i,\lin}^\dg(z; \ci) = Y_i(z) -\{x_i^\dg(\ci)\}^\T \gamma_{\lin,z}^\dg(\ci)
\enda be the adjusted potential outcomes for $\dgsss$. 
Let 
$\tau^\dg_{i, \md}(\ci) = Y_{i, \md}^\dg(1; \ci) - Y_{i, \md}^\dg(0; \ci) $ be the corresponding individual effects  with 
$
\tau^\dg_{i, \fisher}(\ci)   = \tau_i$. 

Write $v_\md^\mim = v_\md^\mim(\ci)$ and $S_{\tau\tau,\md}^\mim = S_{\tau\tau,\md}^\mim(\ci)$ with $\ci = 0_J$ to unify the notations between Propositions \ref{prop:imp_1} and \ref{prop:mim_1}.  
The $S_{\tau\tau,\md}^\dg(\ci)$ in Propositions \ref{prop:imp_1} and \ref{prop:mim_1} equal the population variances of $\{(\tau^\dg_{i, \md}(\ci)\}_{i=1}^N$ over $i = \ot{N}$, respectively, with $S^\dg  _{\tau\tau,\fisher}(\cinf)=S_{\tau}^2$. 
The $v_\md^\dg(\ci)$ in Propositions \ref{prop:imp_1} and \ref{prop:mim_1} equal 
\begina
v_{\md}^\dg (\ci)  = e_0^{-1}S^\dg  _{00, \md}(\ci) + e_1^{-1}S^\dg  _{11, \md}(\ci) - S^\dg_{\tau\tau, \md}(\cinf),
\enda
where $S^\dg  _{zz,\fisher}(\ci)$ and $S^\dg  _{zz,\lin}( \ci)$ are the population variances of  $\{Y_{i,\fisher}^\dg(z) \}_{i=1}^N$ and $\{Y_{i,\lin}^\dg(z) \}_{i=1}^N$ over $i = \ot{N}$, respectively.

We verify below the results for $\htau_{\md }^\im(c) $ in Proposition \ref{prop:imp_1}. 
The proof for $\hts^\mim$  in Proposition \ref{prop:mim_1}  is almost identical and thus omitted. 

\begin{proof}[Proof of Proposition \ref{prop:imp_1}.]

Let $\hts^\imp(\ci)$ and $\hses^\imp(\ci)$, where $\mds$, be the estimators  and robust standard errors by using $x_i^\im  (\cinf) = A_i \circ  x_i + \dci M_i $ as the covariate vector in forming regressions \eqref{eq:im_f} and \eqref{eq:im_l}.
We proceed with the proof in the following three steps:
\begine[(i)]
\item\label{step3:imp} $\sqrt N\{\htsim(\ci) - \tau\}\rs \mathcal N\{0, v_\md^\imp(\ci)\}$ with $v_\lin^\imp(\ci) \leq v_\fisher^\imp(\ci)$. 
\item\label{step1:imp} $\sqrt  N \{\htsim(c) - \tau\}$ has the same limiting distribution as $\sqrt N \{\htsim(\ci)-\tau\}$.
\item\label{step2:imp} $N\{\hses^\imp(c)\}^2  - v_\md^\imp(\ci) = S^\imp_{\tau\tau,\md}(\ci) + \op$. 
\ende 
With $v_\lin^\imp(\ci)\leq v_\lin^\ccov \leq v_\neyman$ following from Proposition \ref{prop:eff_mp}, steps \eqref{step3:imp}--\eqref{step2:imp} together complete the proof.

\paragraph{Proof of step \eqref{step3:imp}.}
With $\ci$ being a fixed vector in $\mathbb R^J$, the covariate vector $x_i^\imp(\ci)$ is  unaffected by the treatment assignment, and thus a true covariate vector  under Condition \ref{cond:a}.  
With $\htsim(\cinf)$ being the analog of $\hts$ based on $x_i^\imp(\ci)$'s, 
it suffices to verify  that the finite population of $\{Y_i(1), Y_i(0), x_i^\im  (\cinf)\}_{i=1}^N$ satisfies Condition \ref{asym_basic}.  
This is ensured by Condition \ref{asym_1}  by direct comparison.
The result then follows from Lemma \ref{lem::basic-x}.

\paragraph{Proof of step \eqref{step1:imp}.}
Write $
\htsim(c) = \htn -  \{\hx^\im  (1;c)  - \hx^\im  (0;c) \}^\T \hgsim(c)$ by  \eqref{lem:algebra} to see 
\begina
\hts^\imp(c) - \hts^\imp(\ci) = - \{\hx^\im  (1;c)  - \hx^\im  (0;c) \}^\T \hgsim(c) + \{\hx^\im  (1; \ci)  - \hx^\im  (0; \ci) \}^\T \hgsim(\ci). 
\enda
To verify $\hts^\imp(c) \asim \hts^\imp(\ci)$ is thus equivalent to verifying
\beginy\label{eq:ss_imp}
 \{\hx^\im  (1;c)  - \hx^\im  (0;c) \}^\T \hgsim(c)  \asim  \{\hx^\im  (1; \ci)  - \hx^\im  (0; \ci) \}^\T \hgsim(\ci). 
\endy

With $\bax(1) - \bax(0) = \bar M(1)- \barm(0) =0_J$ under Condition \ref{cond:a}, it follows from Lemma \ref{lem:joint_dist} that 
$\sqrt N\{\hax(1) - \hax(0)\}$ and $\sqrt N\{\hm(1) - \hm(0)\}$
%
%$\sqrt N\{\hax(1) - \hax(0)\} = \sqrt N  [ \{\hax(1) - \bax(1)\} - \{\hax(0) - \bax(0)\} ]$ and $\sqrt N\{\hm(1) - \hm(0)\} = \sqrt N  [ \{\hm(1) - \barm(1)\} - \{\hm(0) - \barm(0)\} ]$ 
are both asymptotically normal such that 
\beginy
 \{\hx^\im  (1;c)  - \hx^\im  (0;c) \}^\T \hgsim(c) 
& = &
  \{ \hax(1) - \hax(0)\}^\T\hgsim(c) 
   + \{ \hm(1) - \hm(0)\}^\T D_c^\T \hgsim(c) \nonumber\\
   &\asim& 
 \{ \hax(1) - \hax(0)\}^\T\gsim(\ci) 
   + \{ \hm(1) - \hm(0)\}^\T \dci^\T \gsim(\ci) \nonumber\\
   \label{eq:ss_imp_1}
 \endy
by Slutsky's theorem. 
The same reasoning ensures that 
\begina
 \{\hx^\im  (1; \ci)  - \hx^\im  (0; \ci) \}^\T \hgsim(\ci) 
 \asim 
 \{ \hax(1) - \hax(0)\}^\T\gsim(\ci) 
   + \{ \hm(1) - \hm(0)\}^\T \dci^\T \gsim(\ci) 
\enda
by replacing $c$ with $\ci$ in \eqref{eq:ss_imp_1}. 
This verifies   \eqref{eq:ss_imp}.

 \paragraph{Proof of step \eqref{step2:imp}.} The proof of step \eqref{step2:imp} follows from the same reasoning as that of step \eqref{step4:cc} in the proof of Proposition \ref{prop:cc_1}. 
In particular, let $\ep^\imp _{\fisher,i}(c)$ and $\ep^\imp _{\lin,i}(c)$ be the residuals from \eqref{eq:im_f} and \eqref{eq:im_l}, respectively, 
for $i = \ot{N}$. 
Let $\hat\ep^\imp _\md(z; c) $ and 
$
\hat S_{\ep\ep,\md}^\imp (z; c)$
  be the sample mean and variance of $\ep^\imp _{\md,i}(c)$'s over $\{i: C_i = 1, Z_i = z\}$. 
We have 
\beginy\label{eq:nse_imp}
N  \{\hse^\imp _\md(c)\}^2 -   
\left\{  e_0 ^{-1} \hat S_{\ep\ep,\md}^\imp (0; c) + 
  e_1^{-1} \hat S_{\ep\ep,\md}^\imp (1; c)
 \right\} = \op \qquad \text{for} \quad\mds
\endy
by \cite{LD20} and \cite{ZDfrt}. 
The result follows from 
\begina
\hat S_{\ep\ep,\fisher}^\imp (z; c)
&=& \meaniz \left[ Y_i - \hy(z) - \{x_i^\imp(c) - \hx^\imp(c)\} ^\T \hg_\fisher^\imp(c)\right]^2 \\
&=&\hat S _{YY}(z)  + \{\hgf^\imp (c )\}^\T \hsxx^\imp (z; c) \hgf^\imp(c)  - 2 \{\hgf^\imp (c )\}^\T   \hsxy^\imp (z; c)\\
&=&  S _{YY}(z)  + \{\gf^\imp (\ci)\}^\T \sxx^\imp (z; \ci) \gf^\imp(\ci)  - 2 \{\gf^\imp (\ci )\}^\T   \sxy^\imp (z; \ci)+ \op\\
&=&
N^{-1}\sumi \left[ \{Y_i(z) - \by(z)\} - \{ x_i^\imp(\ci) - \bx^\imp(\ci)\}^\T \gf^\imp(\ci)  \} \right]^2  +\op\\
&=& S_{zz, \fisher}^\imp(\ci)  + \op
\enda
by Lemma \ref{lem:plim_imp}, 
and likewise $\hat S_{\ep\ep,\lin}^\imp (z; c) = S_{zz, \lin}^\imp(\ci)  + \op$. 
\end{proof}

\subsection{Invariance of $\htsmi(c)$ to the choice of $c$}\label{sec:mi_app}

\begin{proof}[Proof of Lemma \ref{lem:invar_mim}]
The result follows from the invariance property of {\ols} in Lemma \ref{lem:trivial}. We give details below.  

Let $M = (M_1, \ldots, M_N)^\T = (M_{ij})_{N\times J}$ and $A\circ X = (A_{ij}x_{ij})_{N\times J}  = (x_1^0, \ldots, x_N^0)^\T$. 
The matrix of  imputed covariates equals 
$X^\im  (c) = (x_1^\im  (c), \ldots, x_N^\im  (c))^\T = A\circ X + M \dc$. %, recalling  $D_c =  \diag(c)= \diag(c_1, \ldots, c_J)$.
Let
\begina
H_c = \beginp I_J \\ D_c & I_J\endp,
\enda 
which is invertible. 
The design matrix of the additive regression \eqref{eq:mi_f} equals  
\begina
\chi_\fisher(c) & =& \big(\on, Z,  X^\im  (c), M \big) =  \big(\on, Z, A\circ X+ MD_c,  M\big)\\
&=& \big(\on, Z, A\circ X, M \big) 
\beginp
1 & & & \\
  & 1 &&\\
  & & I_J & \\
  & & D_c & I_J 
\endp = \chi_\fisher(0_J) \beginp
1 & &   \\
  & 1 & \\
  & &H_c
\endp.
\enda
This ensures 
\begina
\beginp
\hmu_{\fisher}^\mi (c)\\
\htau_{\fisher}^\mi (c)\\  
\hg_{\fisher}^\mi (c)
\endp = \beginp
1 & &   \\
  & 1 & \\ 
  & & H_c
\endp^{-1}\beginp
\hmu_{\fisher}^\mi (0_J)\\
\htfmi(0_J)\\
\hg_{\fisher}^\mi (0_J)
\endp 
\enda
and thus $\htf^\mim(c) = \htf^\mim(0_J)$ for arbitrary $c$; the $\hmu_\fisher^\mi(c)$ and $\hgf^\mim(c)$ denote the coefficients of $1$ and $x^\mim(c) $, respectively. Likewise for the robust standard error. 

Let $P_N = I - N^{-1} \on\ont$ be the $N$-dimensional projection matrix; we suppress the subscript ``$N$" when no confusion would arise. 
Let $D_Z   = \diag(Z_1, \ldots, Z_N)$. 
The design matrix of the fully interacted regression \eqref{eq:mi_l} equals 
\begina
\chi_\lin(c) & =& \big(\on, Z,  PX^\im  (c), PM, D_Z PX^\im  (c), D_ZPM \big) \\
&=&  \big(\on, Z, P( A\circ X+ MD_c), PM,  D_Z P (A\circ X+ MD_c), D_ZPM \big)\\
&=& \big(\on, Z, P( A\circ X), PM,  D_Z P(A\circ X), D_ZPM \big)
\beginp
1 & & & \\
  & 1 &&\\
  & & I_J & \\
  & & D_c & I_J \\
  & & & & I_J \\
  & & & & D_c & I_J 
\endp\\
& =& \chi_\lin(0_J) \beginp
1 & & & \\
  & 1 &&\\
  & & H_c & \\
  & &  & H_c
\endp.
\enda
This ensures 
\begina
\beginp
\hmu_{\lin}^\mi (c)\\
\htau_{\lin}^\mi (c)\\
\hg_{\lin}^\mi (c)\\
\hd_{\lin}^\mi (c)
\endp = \beginp
1 & & & \\
  & 1 &&\\ 
  & & H_c & \\
  & &   & H_c
\endp^{-1}\beginp
\hmu_{\lin}^\mi (0_J)\\
\htau_{\lin}^\mi (0_J)\\
\hg_{\lin}^\mi (0_J)\\
\hd_{\lin}^\mi (0_J)
\endp 
\enda
and thus $\htl^\mim(c) = \htl^\mim(0_J)$ for arbitrary $c$; the $\hmu_\lin^\mi(c)$, $\hg_\lin^\mim(c)$, and $\hd_\lin^\mi(c)$ denote the coefficients of $1$, $\{x^\mim(c) - \bx^\mim(c)\}$, and $Z_i\{x^\mim(c) - \bx^\mim(c)\}$, respectively.  Likewise for the robust standard error. 
\end{proof}

\section{Missingness-pattern method}\label{sec:mp_aggregate_app}

\subsection{Kronocker product notation for OLS} 

For $(u_{i,[1]}, \ldots, u_{i,[K]})_{i\in\mathcal I}$, where $ u_{i, [k]}$ is a $J_k$-vector of non-constant regressors, 
let $  v_i  \sim 1 + u_{i,[1]} + \cdots + u_{i,[K]}$ denote the additive {\ols} fit of $v_i$ on  $(1, u_{i,[1]} , \ldots,  u_{i,[K]})$  with regressor vector  $u_i = ( 1, u_{i,[1]}^\T, \ldots,  u_{i,[K]}^\T)^\T$; let $ v_i \sim \otimes_{k=1}^K (1, u^\T_{i, [k]})^\T$ denote the fully interacted {\ols} fit of $v_i$ on $( 1, u_{i,[1]} , \ldots,  u_{i,[K]})$ and all their interactions with regressor vector 
% let $ v_i \sim u_{i,[1]} * \cdots * u_{i,[K]}$ be a shorthand notation for the fully interacted {\ols} regression of $v_i$ on $( 1, u_{i,[1]} , \ldots,  u_{i,[K]})$ and all their interactions by R notation.
%The corresponding regressor vector equals the  Kronecker product of $(1, u^\T_{i, [k]})^\T$ over $k = \ot{K} $, denoted by 
\begina
u_i = (1, u^\T_{i, [1]})^\T \otimes \cdots \otimes (1, u^\T_{i, [K]})^\T = \otimes_{k=1}^K (1, u^\T_{i, [k]})^\T.
\enda
In particular, for $\{(Y_i, Z_i, x_i): i \in \mathcal I\}$ with $Y_i \in \mathbb{R}$, $Z_i \in \mathbb{R}$, and $ x_i \in \mathbb{R}^J$, let $Y_i \sim 1 + Z_i +x_i $ denote the additive  {\ols} fit of $Y_i$ on $(1, Z_i, x_i)$ over $i\in\mI$ with regressor vector $u_i = (1, Z_i, x_i^\T)^\T$; let $Y_i \sim 1 + (x_i - \bar x) + Z_i + Z_i(x_i - \bar x)  $ denote the fully interacted {\ols} fit  of $Y_i$ on $\{1, Z_i, (x_i - \bar x)\}$  over $i\in\mI$ with regressor vector $u_i = (1, Z_i)^\T \otimes (1, (x_i - \bar x)^\T)^\T  $. 

The Kronocker product notation also facilitate the discussion of different parametrizations of the missingness patterns. Index the $Q = 2^J  $ possible missingness patterns in $ \{0,1\}^J$  in lexicographical order as $\{m^{(q)}: q = \ot{Q} \}$ with $m^{(1)} = 0_J$, $m^{(2)} = (0_{J-1}^\T, 1)^\T$, $m^{(Q)} = 1_J$, etc. 
We have 
\beginy\label{eq:t}
t_i &=& (1-M_{i1}, M_{i1})^\T   \otimes \cdots \otimes ( 1-M_{iJ}, M_{iJ})^\T \nonumber \\
&=& (1(M_i = m^{(1)}), \ldots, 1(M_i = m^{(Q)}) )^\T
\endy 
gives the indicator vector of the $Q$ missingness patterns for $i = \ot{N}$.
Let 
$$
f_i = (1, M_{i1})^\T \otimes \cdots \otimes (1, M_{iJ})^\T
$$ 
be a non-degenerate linear transformation of $t_i$ with 
\beginy\label{eq:ft}
f_i = \Phi  t_i, \quad \text{where} \quad  \Phi = \otj \beginp 1& 1\\ 0 & 1\endp . 
\endy
The $f_i$ consists of 1 and 
 $(\prod_{j \in \mathcal J'} M_{ij})$'s for all $\emptyset\neq \mathcal J' \subseteq \{\ot{J}\}$.
The vector $u^\mp(c)$ in   \eqref{eq:mp_agg} is essentially the subvector of $(1, (x_i^\imp(c)^\T)^\T \otimes (1, M_{i1})^\T \otimes \cdots \otimes (1, M_{iJ})^\T = (1, (x_i^\imp(c)^\T)^\T \otimes f_i $ after excluding the first element 1, with 
\beginy\label{eq:u_mp_f}
(1, u_i^\mp(c)^\T)^\T = (1, x_i^\imp(c)^\T)^\T\otimes f_i.
\endy
%{\red The Kronocker product notation also facilitate the discussion of different parametrization of the missingness patterns. Let $Q =2^J $. 
%Let 
%$$
%f_i = (1, M_{i1})^\T \otimes \cdots \otimes (1, M_{iJ})^\T
%$$ 
%be the $Q$-vector of $(\prod_{j \in \mathcal J'} M_{ij})$'s for all $\mathcal J' \subseteq \{\ot{J}\}$; the definition includes $\mathcal J' = \emptyset$, which corresponds to the first element 1.  
%Let 
%$$
%t_i = (1-M_{i1}, M_{i1})^\T   \otimes \cdots \otimes ( 1-M_{iJ}, M_{iJ})^\T 
%$$
%be a linear transformation of $f_i$ with 
%\beginy\label{eq:ft}
%f_i = \Phi  t_i, \quad \text{where} \quad  \Phi = \otj \beginp 1& 1\\ 0 & 1\endp . 
%\endy
%The $t_i$ is essentially the vector of indicators of the $2^J$ possible missingness patterns. In particular, index the $Q  $ possible missingness patterns in $ \{0,1\}^J$  in lexicographical order as $\{m^{(q)}: q = \ot{Q} \}$ with $m^{(1)} = 0_J$, $m^{(2)} = (0_{J-1}^\T, 1)^\T$, $m^{(Q)} = 1_J$, etc. We have 
%\beginy\label{eq:t}
%t_i = (1(M_i = m^{(1)}), \ldots, 1(M_i = m^{(Q)}) )^\T \quad \text{with} \quad t_i^\T t_{i'} = 1(M_i = M_{i'})
%\endy
% for all $i, i' = \ot{N}$.
%The vector $u^\mp(c)$ in   \eqref{eq:mp_agg} is essentially the subvector of $(1, (x_i^\imp(c)^\T)^\T \otimes (1, M_{i1})^\T \otimes \cdots \otimes (1, M_{iJ})^\T = (1, (x_i^\imp(c)^\T)^\T \otimes f_i $ after excluding the first element 1, with 
%\beginy\label{eq:u_mp_f}
%(1, u_i^\mp(c)^\T)^\T = (1, x_i^\imp(c)^\T)^\T\otimes f_i.
%\endy}
Recall from Examples \ref{ex:u_mp_1} and \ref{ex:u_mp_2} (continued) in Section \ref{sec:mp_asym} of the main text that elements in $(1, u_i^\mp(c)^\T)^\T$ could be collinear.
Whereas we adjusted for collinearity by explicitly removing the collinear terms in the examples, we take a different approach in the proof below for algebraic simplicity.
In particular, we stipulate the coefficient of a regressor that is collinear with some of the earlier regressors to be zero.
This stipulation is consistent with the way that standard software packages for {\ols} handle the collinear regressors, and allows us to verify the correspondence between the aggregate and missingness-pattern-specific regressions using the full vector of $u_i^\mp(c)$ as defined by \eqref{eq:u_mp_f}.

\subsection{Proof of Proposition \ref{prop:mp_agg}}

Denote by $\ttl$ the coefficient of $Z_i$ from 
\beginy
Y_i \sim 1 + Z_i + (u_i^\mp(c) - \bar u^\mp(c)) + Z_i   (u_i^\mp(c) - \bar u^\mp(c)). \label{eq:mp_agg_app_1}
\endy
%where the superscript ``ag" indicates aggregate. 
The goal is to verify $\ttl = \htl^\mp$ as defined in \eqref{eq:htau_mp}.

Let $u_i^\mp = u_i^\mp(0_J)$ be the value of $u_i^\mp(c)$ at $c = 0_J$, with 
\beginy\label{eq:u_mp_f_0}
(1, (u_i^\mp)^\T)^\T 
%= (1, (x_i^0)^\T)^\T \otimes (1, M_{i1})^\T \otimes \cdots \otimes (1, M_{iJ})^\T
= (1, (x_i^0)^\T)^\T \otimes f_i.
\endy 
%Let 
%$
%t_i = (1-M_{i1}, M_{i1})^\T   \otimes \cdots \otimes ( 1-M_{iJ}, M_{iJ})^\T 
%$
%be a linear transformation of $f_i$ with 
%\beginy\label{eq:ft}
%f_i = \Phi  t_i, \quad \text{where} \quad  \Phi = \otj \beginp 1& 1\\ 0 & 1\endp; 
%\endy
%the equality follows from 
%\begina
%\beginp
%1  \\
%M_{ij}
%\endp
%= 
%\beginp
%1 & 1\\
%0 & 1
%\endp
%\beginp
%1 - M_{ij}\\
%M_{ij}
%\endp. 
%\enda
Recall $\bx_{(m)}^\mp= N^{-1}_{(m)} \sum_{i: M_i = m} x_i^\mp$ with dimension $J_{(m)}$.
% as the average covariate under missingness pattern $m = (m_1, \ldots, m_J)^\T \in \{0,1\}^J$. 
Let $\bxz_{(m)} = N^{-1}_{(m)} \sum_{i: M_i = m} \xiz$ with dimension $J$.
% be an augmented variant of $\bx_{(m)}^\mp$, with dimensions that are missing under pattern $m$, namely $\{j: m_j = 1\}$,  filled with 0. 
Consider the below four variants of \eqref{eq:mp_agg_app_1} as key stepping stones for verifying the result. We use ``$\sim$" to indicate equivalent regression formulas up to a reordering of the regressors:
\beginy
Y_i &\sim& 1 + Z_i +  (u_i^\mp  - \bar u^\mp ) + Z_i  (u_i^\mp  - \bar u^\mp ) \nonumber \\
&\sim&  (1, Z_i) \otimes ( 1, (u_i^\mp  - \bar u^\mp ) ^\T)^\T \label{eq:mp_agg_0}\\
\nonumber\\
%%%%%%%%%%%
\text{Formula ``f": }\qquad Y_i &\sim& 1 + Z_i +  u_i^\mp   + Z_i  u_i^\mp\nonumber\\
&\sim& (1, Z_i)^\T \otimes (1, (u_i^\mp)^\T)^\T \nonumber\\
&\sim& (1, Z_i)^\T \otimes  (1, (x_i^0)^\T)^\T \otimes f_i, \label{eq:uc}\\ 
&\sim& f_i + x_i^0 \otimes f_i + Z_i f_i + Z_i x_i^0\otimes f_i,\label{eq:uc_2}\\
\nonumber\\
%%%%%%%%%%%%%%%%%
\text{Formula ``t": }\qquad Y_i &\sim&   (1, Z_i)^\T \otimes  (1, (x_i^0)^\T)^\T \otimes t_i \label{eq:uc_t}\\
\nonumber\\
%%%%%%%%%%%%%%%%%
\text{Formula ``tc":}\qquad Y_i &\sim&   (1, Z_i)^\T \otimes  (1, (x_i^0 - \bx_{(M_i)}^0)^\T)^\T \otimes t_i\label{eq:mp_agg_0_t} \\
&\sim & t_i +(x_i^0 - \bx_{(M_i)}^0) \otimes t_i + Z_i t_i + Z_i (x_i^0 - \bx_{(M_i)}^0)\otimes t_i.\label{eq:mp_agg_0_t_2}
\endy
In particular, 
\eqref{eq:mp_agg_0} replaces the $u_i^\mp(c)$ in  \eqref{eq:mp_agg_app_1} with $u_i^\mp$ in forming the centered fully interacted regression; 
\eqref{eq:uc} gives the uncentered variant of \eqref{eq:mp_agg_0}; 
\eqref{eq:uc_t} reparameterizes the $f_i$ component in \eqref{eq:uc} as $t_i$; 
\eqref{eq:mp_agg_0_t} centers the covariates $x_i^0$ in \eqref{eq:uc_t} by the missingness-pattern-specific means, $\bxz_{(M_i)}$. 
For easy reference, we index regressions \eqref{eq:uc}--\eqref{eq:mp_agg_0_t} by ``f", ``t", and ``tc", respectively, with letters ``f", ``t", and ``c"  indicating using $f_i$, using $t_i$, and centered $x_i^0$, respectively. 

%%%%%%%%%%%%%%%%%
We proceed with the proof in the following three steps:
\begine[(i)]
\item\label{step:i}  The coefficient of $Z_i$ from \eqref{eq:mp_agg_app_1}, $\ttl$, equals the coefficient of $Z_i$ from   \eqref{eq:mp_agg_0}, denoted by  $\ttlz  $. 
\item\label{step:ii} The $\htl^\mp$ from \eqref{eq:htau_mp} satisfies $
\htl^\mp = \rho^\T \hat\delta_\tc$, 
where $\hd_\tc$ is the coefficient of $Z_i t_i$ in \eqref{eq:mp_agg_0_t_2} and $\rho$ is the $2^J$-vector of $\{\rho_{(m)}: m \in \{0, 1\}^J\}$ in lexicographical order of $m$. 
%%%%%%%%%%%%%%
\item\label{step:iii} 
The coefficient of $Z_i$ from   \eqref{eq:mp_agg_0} satisfies 
\beginy\label{eq::step-3}
\ttlz   = (1, (\bar u^\mp)^\T) \beginp \hd_\ff \\ \hg_\ff \endp = \rho^\T \hd_\tc,
\endy
where $\hd_\ff$ and $\hg_\ff$ are the coefficients of $Z_i f_i$ and $Z_i x_i^0 \otimes f_i$ from \eqref{eq:uc_2}.
\ende
Steps \eqref{step:i}--\eqref{step:iii} together ensure
$
\ttl = \ttlz      =  \rho^\T \hd_\tc =  \htl^\mp  
$
and complete the proof. 
Denote by 
\beginy
&&w_i =  (1, Z_i)^\T \otimes ( 1, (u_i^\mp  - \bar u^\mp ) ^\T)^\T,\nonumber\\
&&\wif =  (1, Z_i)^\T \otimes (1, (u_i^\mp)^\T)^\T = (1, Z_i)^\T \otimes  (1, (x_i^0)^\T)^\T \otimes f_i,\nonumber\\
&&\wit =  (1, Z_i)^\T \otimes  (1, (x_i^0)^\T)^\T \otimes t_i,\nonumber\\
&&\witc = (1, Z_i)^\T \otimes  (1, (x_i^0 - \bx_{(M_i)}^0)^\T)^\T \otimes t_i
\label{eq:ws}
\endy
the regressor vectors of regressions \eqref{eq:mp_agg_0}--\eqref{eq:mp_agg_0_t}, respectively, with 
$\htheta$, $\htheta_\ff$, $\htheta_\tt$, and $\htheta_\tc$ as the corresponding coefficient vectors. We verify below steps \eqref{step:i}--\eqref{step:iii} one by one.

\paragraph{Proof of step \eqref{step:i}.}
From \eqref{eq:u_mp_f} and \eqref{eq:u_mp_f_0}, we have 
\begina
(1, u_i^\mp(c)^\T)^\T 
&= & (1, x_i^\imp(c)^\T)^\T \otimes f_i\\
&=& (1, (x_i^0)^\T + M_i^\T D_c)^\T\otimes f_i\\
&=& (1, (u_i^\mp)^\T)^\T +  (0,  M_i^\T D_c)^\T\otimes f_i,
\enda
where 
$
 (0,  M_i^\T D_c)^\T\otimes f_i=
  (0,  c_1M_{i1}, \ldots, c_J M_{iJ})^\T\otimes f_i 
$
is a linear combination of $f_i  \backslash\{1\} = (1, M_{i1})^\T \otimes \cdots \otimes (1, M_{iJ})^\T \backslash \{1\}$, which is in turn a linear combination of both $u_i^\mp(c)$ and $u_i^\mp$. 
This ensures $u_i^\mp(c)$ is a non-degenerate linear transformation of $u_i^\mp$.
The equivalence of $\ttl$ and $\ttlz  $ follows from the invariance of \citet{Lin13}'s estimator by Lemma \ref{lem:trivial}.

\paragraph{Proof of step \eqref{step:ii}.}
%First, 
%$t_i$ is essentially the vector of indicators of the $2^J$ possible missingness patterns. In particular, index the $Q = 2^J$ possible missingness patterns in $ \{0,1\}^J$  in lexicographical order as $\{m^{(q)}: q = \ot{Q} \}$ with $m^{(1)} = 0_J$, $m^{(2)} = (0_{J-1}^\T, 1)^\T$, $m^{(Q)} = 1_J$, etc. We have 
%\beginy\label{eq:t}
%t_i = (1(M_i = m^{(1)}), \ldots, 1(M_i = m^{(Q)}) )^\T \quad \text{with} \quad t_i^\T t_{i'} = 1(M_i = M_{i'})
%\endy
% for all $i, i' = \ot{N}$. 
%The result then follows from the properties of saturated {\ols}. 
%In particular, 
%the regressor vector of 
%\eqref{eq:mp_agg_0_t_2} 
%%equals
%%$
%%\witc = (1, Z_i) \otimes  (1, (x_i^0 - \bx_{(M_i)}^0)^\T)^\T \otimes t_i
%%$
%%by 
%in \eqref{eq:ws}
%and satisfies 
%\begina
%\witc^\T w_{i', \tc} &=& \left\{ (1, Z_i) \otimes  (1, (x_i^0 - \bx_{(M_i)}^0)^\T) \otimes t_i^\T\right\} \left\{  (1, Z_{i'})^\T \otimes  (1, (x_{i'}^0 - \bx_{(M_{i'})}^0)^\T)^\T \otimes t_{i'} \right\}\\
%&=& \left\{ (1, Z_i) (1, Z_{i'})^\T  \right\} \otimes\left\{  (1, (x_i^0 - \bx_{(M_i)}^0)^\T)   (1, (x_{i'}^0 - \bx_{(M_{i'})}^0)^\T)^\T \right\}\otimes (  t_i^\T   t_{i'})\\
%&=& 0
%\enda
%for $i$ and $i'$ with different missingness patterns by \eqref{eq:t}.  
The {\ols} \eqref{eq:mp_agg_0_t_2} is equivalent to $|\mathcal M|$ missingness-pattern-specific {\ols}, 
\begina
Y_i \sim (1, Z_i)^\T \otimes (1, (x_i^0 - \bar x^0_{(M_i)})^\T)^\T \qquad \text{for} \quad \{i: M_i = m\},
\enda
which are identical to those that produced $\htlm$  after removing the dimensions in $(x_i^0 - \bar x^0_{(M_i)})$ that are constantly zero. 
This ensures the element of $\hd_\tc$ that corresponds to $Z_i 1(M_i = m) \in Z_i t_i $  equals $\htlm$. Therefore, $
\htl^\mp = \summ \rho_{(m)} \htlm = \sum_{q=1}^Q \rho_{(m^{(q)})} \htau_{\lin,(m^{(q)})} = \rho^\T \hd_\tc$. The previous argument works well when $\mathcal M$ includes all missingness patterns. It is also rigorous in general because $\rho_{(m)} = 0$ and the element of $\hd_\tc$ that corresponds to $Z_i 1(M_i = m) \in Z_i t_i $ can be arbitrary for $m \not\in\mathcal M$.

% for $m \in \mathcal M$ and {\red \texttt{NA} for $m \not\in\mathcal M$. This, together with $\rho_{(m)} = 0$ for all $m \not\in\mathcal M$, ensures 
%$
%\htl^\mp = \summ \rho_{(m)} \htlm = \sum_{k=1}^K \rho_{(m^{(q)})} \htau_{\lin,(m^{(q)})} = \rho^\T \hd_\tc$. 
%}

\paragraph{Proof of step \eqref{step:iii}, the first identity in \eqref{eq::step-3}.}
%The regressor vectors of \eqref{eq:mp_agg_0} and \eqref{eq:uc} equal
%\begina
%&&w_i = (1, Z_i)^\T \otimes (1, (u_i^\mp - \bar u ^\mp)^\T )^\T
%= (1,  (u_i^\mp - \bar u ^\mp)^\T , Z_i, Z_i (u_i^\mp - \bar u ^\mp)^\T)^\T,\\
%%%%%%%%%%%%%%%%%%%%%%%%
%&&\wif = (1, Z_i)^\T \otimes (1, (u_i^\mp)^\T )^\T
%= (1,  (u_i^\mp)^\T , Z_i, Z_i (u_i^\mp)^\T)^\T,
%\enda
%respectively, from \eqref{eq:ws}. 
Let
$
\htheta = (\tilde \mu, \tilde\beta^\T, \ttl , \tilde  \xi^\T)^\T
$
be the vector of coefficients of \eqref{eq:mp_agg_0}, with $\tilde \mu$, $\tilde\beta$, $\ttl$, and $\tilde \xi$ 
corresponding to $1$, $(u_i^\mp - \bar u ^\mp)$, $Z_i$, and $Z_i (u_i^\mp - \bar u ^\mp)$, respectively, in the order of $w_i$. 
Let
$
\htheta_\ff = ( \tilde \mu_\ff, \tilde\beta_\ff^\T, \ttlz  , \tilde  \xi_\ff ^\T)^\T
$
be the vector of coefficients of \eqref{eq:uc}, with $\tilde \mu_\ff $, $\tilde\beta_\ff $, $\ttlz  $, and $\tilde  \xi_\ff $ 
corresponding to $1$, $u_i^\mp$, $Z_i$, and $Z_i u_i^\mp$, respectively, in the order of $\wif$. 
Let $L =  (1+J) 2^J - 1$ denote the length of $u_i^\mp$. 
It follows from 
\begina
\qquad\qquad\quad
\beginp
1\\
u_i^\mp
\endp = 
\beginp
1 & 0^\T_L  \\
\bar u^\mp & I_L 
\endp 
\beginp
1\\
u_i^\mp - \bar u ^\mp  
\endp
= \Gamma
\beginp
1\\
u_i^\mp - \bar u ^\mp
\endp,\quad \text{where}\  \Gamma = \beginp
1 & 0^\T_L  \\
\bar u^\mp & I_L 
\endp,
\enda
that 
$\wif  = (I_2 \otimes \Gamma) w_i$ by \eqref{eq:ws}
and thus 
\begina
\beginp
\tilde\mu \\
\tilde\beta \\
\ttl \\
\tilde \xi  
\endp
=\tilde\theta  = (I_2 \otimes  \Gamma^\T)\tilde\theta_\ff  = 
 \beginp
1 &  (\bar u^\mp)^\T \\
& I_L  & \\
&&1& (\bar u^\mp)^\T\\
&& & I_L 
 \endp\beginp
\tilde \mu_\ff  \\
\tilde \beta_\ff  \\
\ttlz   \\
\tilde \xi_\ff  
\endp
\enda
by Lemma \ref{lem:trivial}. 
This ensures 
\begina
\ttl =  (1, (\bar u^\mp)^\T) \beginp
\ttlz    \\
\tilde\xi_\ff 
\endp 
= 
 (1, (\bar u^\mp)^\T) \beginp
\hd_\ff \\
\hg_\ff 
\endp;
\enda 
the last equality follows from  
\begina
(Z_i, Z_i(u_i^\mp)^\T)^\T = Z_i (1, (u_i^\mp)^\T )^\T = Z_i (1, (\xiz)^\T)^\T \otimes f_i = (Z_i f_i , (Z_i\xiz \otimes f_i)^\T)^\T
\enda such that $((\ttlz  )^\T, \tilde\xi_\ff^\T)^\T$ and $(\hd_\ff^\T, \hg_\ff^\T)^\T$  are two expressions of the same coefficient vector of  $(Z_i, Z_i(u_i^\mp)^\T)^\T = (Z_i f_i , (Z_i\xiz \otimes f_i)^\T)^\T $.

\paragraph{Proof of step \eqref{step:iii}, the second identity in \eqref{eq::step-3}.}
%We next verify 
%\beginy\label{eq:ss}
% (1, (\bar u^\mp)^\T) 
%\beginp
%\hd_\ff \\
%\hg_\ff 
%\endp
% = \rho^\T \hd_\tc. 
%\endy
%from the proof of step \eqref{step:ii} $ be a shorthand notation for $
%First, recall that we index the $Q = 2^J$ possible missingness patterns in lexicographical order as $\{m^{(q)}: q = \ot{Q} \}$.  
Let  
$t_{(q) } = t_{(m^{(q)})}$ be the common value of $t_i$ for units with $M_i = m^{(q)}$, which 
has 1 in the $q$th dimension, and 0 elsewhere. 
Let $\bar x^0_{(q) }=\bar x^0_{(m^{(q)})}$, with the $j$th element denoted by $\bar x^0_{(q) , j} = \bar x^0_{(m^{(q)}), j} $.
Let 
\begina
G =
\beginp
G_1\\
G_2\\
\vdots\\
G_J
\endp, \quad \text{where}\quad 
G_j = 
\beginp
\bxz_{(1),j}\\
& \bxz_{(2),j} \\
&&\ddots\\
&&& \bxz_{(Q) ,j}
\endp,
\enda
be a $(JQ) \times Q$ matrix that satisfies
\beginy
G = \beginp
\bxz_{(1)} \otimes t_{(1)} , \ \bxz_{(2)} \otimes t_{(2)}, \ \ldots , \ \bxz_{(Q) } \otimes t_{(Q) }\endp. \label{eq:G}
\endy
It follows from the identity 
$
\bxz_{(M_i),j}\cdot 1(M_i = m^{(q)}) = \bxz_{(q) ,j} \cdot 1(M_i = m^{(q)})
$
that 
\begina
\bxz_{(M_i),j} \cdot t_i =   \beginp
 \bxz_{(M_i), j} \cdot 1(M_i = m^{(1)})\\
  \bxz_{(M_i), j} \cdot 1(M_i = m^{(2)})\\
   \vdots\\
     \bxz_{(M_i), j} \cdot 1(M_i = m^{(Q)})
     \endp
     = \beginp
 \bxz_{(1), j} \cdot 1(M_i = m^{(1)})\\
  \bxz_{(2), j} \cdot 1(M_i = m^{(2)})\\
   \vdots\\
     \bxz_{(Q) , j} \cdot 1(M_i = m^{(Q)})
     \endp = G_j t_i 
\enda
and thus 
\begina
 \bar x^0_{(M_i)}  \otimes t_i 
 = 
\beginp
 \bxz_{(M_i), 1} t_i \\
    \bxz_{(M_i), 2} t_i \\
    \vdots\\
      \bxz_{(M_i), J} t_i 
 \endp
 =
 \beginp
G_1 t_i \\
    G_2 t_i \\
    \vdots\\
     G_J t_i 
 \endp
 = G t_i. 
\enda
This ensures 
\begina
(1, (\xiz)^\T)^\T \otimes t_i  &=& 
\beginp
t_i \\ \xiz\otimes t_i 
\endp  
= 
\beginp
t_i \\ \bar x^0_{(M_i)}  \otimes t_i + (\xiz - \bar x^0_{(M_i)}) \otimes t_i 
\endp \\
& =&
 \beginp
 I_Q & 0_{Q \times (JQ)} \\
 G & I_{JQ}
 \endp
 \beginp
 t_i \\  (\xiz - \bar x^0_{(M_i)}) \otimes t_i 
 \endp
 \\
 &=& H  \left\{ \big(1, (\xiz - \bar x^0_{(M_i)})^\T \big)^\T  \otimes t_i \right\},  \qquad \text{where} \ H =  \beginp
 I_Q &0_{Q \times (JQ)}  \\
 G & I_{JQ}
 \endp. 
\enda
The regressor vectors of  \eqref{eq:uc_t} and \eqref{eq:mp_agg_0_t}  thus satisfy $\wit = (I_2 \otimes H)\witc $ by \eqref{eq:ws}. 
%\begina
%\wit 
%&=& (1, Z_i)^\T \otimes \left\{ (1, (\xiz)^\T)^\T \otimes t_i \right\} = (1, Z_i)^\T \otimes \left[ H  \left\{ (1, (\xiz - \bar x^0_{(M_i)})^\T )^\T  \otimes t_i \right\} \right]  \\
%&= & (I_2 \otimes H) \left\{  (1, Z_i)^\T \otimes (1, (\xiz - \bar x^0_{(M_i)})^\T )^\T  \otimes t_i \right\} \\
%& =&  (I_2 \otimes H)\witc. 
%\enda
By Lemma \ref{lem:trivial}, the corresponding coefficient vectors, denoted by $\htheta_\tt$ and $\htheta_\tc$, respectively, satisfy 
\begina
\beginp
\hat\lambda_\tc\\
\hat\nu_\tc\\
\hd_\tc\\
\hat\gamma_\tc
\endp
=
\htheta_\tc = (I_2 \otimes H^\T) \htheta_\tt = 
\beginp
I_Q & G^\T \\
& I_{JQ} \\ 
&& I_Q & G^\T\\
&& & I_{JQ}
\endp 
\beginp
\hat\lambda_\tt\\
\hat\nu_\tt\\
\hd_\tt\\
\hat\gamma_\tt
\endp,
\enda
where $\hat\lambda_\tc$, $\hat\nu_\tc$, $\hd_\tc$, and $\hg_\tc$ are the components of $\htheta_\tc$ corresponding to $t_i$, $(x_i^0 - \bx_{(M_i)}^0) \otimes t_i$, $Z_i t_i$, and $Z_i (x_i^0 - \bx_{(M_i)}^0) \otimes t_i$, respectively, in \eqref{eq:mp_agg_0_t_2}; likewise for 
$\hat\lambda_\tt$, $\hat\nu_\tt$, $\hd_\tt$, and $\hg_\tt$ as the components of $\htheta_\tt$ corresponding to $t_i$, $x_i^0 \otimes t_i$, $Z_i t_i$, and $Z_i x_i^0 \otimes t_i$, respectively, in decomposition
\begina
Y_i &\sim&  (1, Z_i)^\T \otimes  (1, (x_i^0)^\T)^\T \otimes t_i \nonumber \\
&\sim& t_i + x_i^0 \otimes t_i + Z_i t_i + Z_i x_i^0\otimes t_i  %,%\label{eq:uc_t_2}\\
\enda
of \eqref{eq:uc_t}. 
This ensures 
\beginy\label{eq:tc_t}
\hd_\tc = (I_Q, G^\T) 
\beginp
\hd_\tt\\
\hat\gamma_\tt
\endp.
\endy

On the other hand, it follows from \eqref{eq:ft} that the regressor  vectors of \eqref{eq:uc} and \eqref{eq:uc_t} satisfy $\wif =  ( I_{2(J+1)}\otimes \Phi )  \wit$ such that the coefficient vectors satisfy
\begina
\beginp
\hat\lambda_\tt\\
\hat\nu_\tt\\
\hd_\tt\\
\hat\gamma_\tt
\endp
=
\htheta_\tt 
= ( I_{2(J+1)}\otimes \Phi^\T )  \htheta_\ff
= 
 \beginp
 \Phi^\T & \\
& I_J \otimes \Phi^\T\\ 
&& \Phi^\T & \\
&& & I_J \otimes \Phi^\T
\endp 
\beginp
\hat\lambda_\ff\\
\hat\nu_\ff\\
\hd_\ff\\
\hat\gamma_\ff
\endp,
\enda
where 
$\hat\lambda_\ff$, $\hat\nu_\ff$ are the components of $\htheta_\ff$ corresponding to $f_i$, $x_i^0 \otimes f_i$ in \eqref{eq:uc_2}, in addition to $\hd_\tt$, and $\hg_\tt$ that correspond to $Z_i f_i$  and $Z_i x_i^0 \otimes f_i$, respectively.
This ensures 
\begina
\beginp
\hd_\tt\\
\hat\gamma_\tt
\endp
= 
 \beginp
 \Phi^\T & \\
& I_J \otimes \Phi^\T
\endp 
\beginp
\hd_\ff\\
\hat\gamma_\ff
\endp
\enda
such that, together with \eqref{eq:tc_t}, we have 
\begina
\rho^\T \hd_\tc 
= \rho^\T(I_Q, G^\T)  \beginp
 \Phi^\T & \\
& I_J \otimes \Phi^\T
\endp 
\beginp
\hd_\ff\\
\hat\gamma_\ff
\endp.
\enda
A sufficient condition for the second identity in \eqref{eq::step-3} is thus 
\beginy\label{eq:sss}
(1, (\bar u^\mp)^\T) = \rho^\T (I_Q, G^\T)
\beginp
\Phi^\T \\
 & I_J\otimes \Phi^\T
\endp = 
(\rho^\T , \rho^\T G^\T)
\beginp
\Phi^\T \\
 & I_J\otimes \Phi^\T
\endp.
\endy
%or equivalently, 
%\begina
%(1, (\bar u^\mp)^\T)^\T = 
%\beginp
%A^\T \\
% & I_J\otimes A^\T 
%\endp
%\beginp
%I\\
%G
%\endp
%\rho 
%= 
%\beginp
%A^\T \\
% & I_J\otimes A^\T 
%\endp
%\beginp
%\rho \\
%G\rho 
%\endp
%\enda
This is indeed correct. In particular, it follows from \eqref{eq:u_mp_f_0} and \eqref{eq:ft} that 
\begina
(1, (u_i^\mp)^\T)^\T 
&  =& (1, (x_i^0)^\T)^\T \otimes f_i =(1, (x_i^0)^\T)^\T \otimes (\Phi t_i) = (I_{J+1} \otimes \Phi ) \{(1, (x_i^0)^\T)^\T \otimes t_i\}\\
& =& 
\beginp
\Phi \\
& I_J \otimes \Phi
\endp
\beginp
t_i\\
x_i^0\otimes t_i
\endp.
\enda
This, together with $\rho = \meani t_i$ by \eqref{eq:t} and 
\begina
 \meani x_i^0\otimes t_i 
&=& N^{-1} \sumq  \sum_{i: M_i = m^{(q)}} x_i^0 \otimes t_i 
= N^{-1} \sumq  \left(\sum_{i: M_i = m^{(q)}} x_i^0\right) \otimes t_{(q) } \\
&=& N^{-1} \sumq  N_{(q) } \bar x_{(q) }^0 \otimes t_{(q) } =
\sumq  \rho_{(q) } \bar x_{(q) }^0 \otimes t_{(q) } = G\rho
\enda
by \eqref{eq:G}, ensures that 
\begina
(1, (\bar u^\mp)^\T)^\T = \meani (1, (u_i^\mp)^\T)^\T 
=
\beginp
\Phi \\
& I_J \otimes \Phi
\endp
\beginp
\meani t_i \\
\meani \xiz \otimes t_i 
\endp
= \beginp
\Phi \\
& I_J \otimes \Phi
\endp
\beginp
\rho \\
G \rho
\endp.
\enda
This verifies the second identity in \eqref{eq::step-3} via \eqref{eq:sss} and completes the proof.

\subsection{Aggregate regression for recovering $\htau_\text{F}^\mp$}\label{sec:mp_add_app}

%Recall $
%f_i = (1, M_{i1})^\T \otimes \cdots \otimes (1, M_{iJ})^\T$ as the $2^J$-vector of $(\prod_{j \in \mathcal J'} M_{ij})$ for all $\mathcal J' \subseteq \{\ot{J}\}$.
%Let $f_i'$ be the subvector of $f_i $ without the first constant term, with $f_i = (1, (f_i')^\T)^\T$. 
%Proposition \ref{prop:mp_agg_f} gives the aggregate regression for recovering $(\htf^\mp, \hse_\fisher^\mp)$. 

\begin{proposition}\label{prop:mp_agg_f}
The additive missingness-pattern estimators $\htf^\mp$ and $\hse_\fisher^\mp$ from \eqref{eq:htau_mp} and \eqref{eq:hse_mp} equal the coefficient of $Z_i$ and its associated robust standard error from 
\beginy\label{eq:add}
Y_i \sim (1, Z_i, x_i^\imp(c)^\T)^\T \otimes (1, (f_i' - \bar f_i')^\T)^\T ,
\endy
where $f_i'$ is the subvector of $f_i $ without the first constant term, with $f_i = (1, (f_i')^\T)^\T$. 
\end{proposition}

Direct comparison shows that \eqref{eq:add} involves  interactions between $Z_i$ and $f_i'$ and is thus different from $Y_i \sim 1 + Z_i + u_i^\mp(c)$. The proof of Proposition \ref{prop:mp_agg_f} follows from  the same logic as that of Proposition \ref{prop:mp_agg}. We give below a sketch of the main steps to avoid repetition. 

\begin{proof}[A sketch of the proof of Proposition \ref{prop:mp_agg_f}] 

Consider the below three variants of \eqref{eq:add} as key stepping stones for verifying the result:
\beginy
Y_i &\sim& (1, Z_i, (x_i^0)^\T) \otimes  (1, (f_i' - \bar f_i')^\T)^\T,
 \label{eq:add_0}\\ 
%%%%%%%%%%%
\text{Formula ``f": }\qquad Y_i &\sim&  (1, Z_i, (x_i^0)^\T) \otimes f_i , \label{eq:add_uc}\\ 
%%%%%%%%%%%%%%
%%%%%%%%%%%%%%%%%
\text{Formula ``t": }\qquad Y_i &\sim&  (1, Z_i, (x_i^0)^\T) \otimes t_i. \label{eq:add_t}
\endy
In particular, 
\eqref{eq:add_0} replaces the $x_i^\imp(c)$ in  \eqref{eq:add} with $x_i^0 = x_i^\imp(0_J)$; 
\eqref{eq:add_uc} gives the uncentered variant of \eqref{eq:add_0}; 
\eqref{eq:add_t} reparameterizes the $f_i$ component in \eqref{eq:add_uc} as $t_i$. 
Index regressions \eqref{eq:add_uc} and \eqref{eq:add_t} by ``f" and ``t", respectively. 
%%%%%%%%%%%%%%%%%
We proceed with the proof in the following three steps:
\begine[(i)]
\item\label{step:i_add}  The coefficient of $Z_i$ from \eqref{eq:add}, $\htf^\agg$, equals the coefficient of $Z_i$ from   \eqref{eq:add_0}, denoted by  $\htf^{\agg,0}$. 
\item\label{step:ii_add} The additive missingness-pattern estimator satisfies $
\htf^\mp = \rho^\T \hat\delta_\tt$, 
where $\hd_\tt$ is the coefficient of $Z_i t_i$ in \eqref{eq:add_t} and $\rho$ is the vector of $\{\rho_{(m)}: m \in \{0, 1\}^J\}$ in lexicographical order of $m$. 
%%%%%%%%%%%%%%
\item\label{step:iii_add} 
$
\htf^{\agg,0}    = \rho^\T \hd_\tc$.
\ende
Steps \eqref{step:i}--\eqref{step:iii} together ensure
$
\htf^\agg = \htf^{\agg,0}     =  \rho^\T \hd_\tt =  \htf^\mp  
$
and complete the proof. 
\end{proof}

\end{document}